\documentclass[aps,10pt,prl,superscriptaddress,twocolumn,footinbib,notitlepage,tightenlines,showpacs]{revtex4-1}

\usepackage{tensor}
\usepackage{amsmath}
\usepackage{amsthm}
\usepackage[dvipdfmx]{graphicx}
\usepackage{float}
\usepackage{comment}
\usepackage{amsfonts}
\usepackage{bm}
\usepackage{braket}
\usepackage[colorlinks=true,linkcolor=blue,urlcolor=blue]{hyperref}
\usepackage{color}
\usepackage{enumerate}
\usepackage{mathtools}
\usepackage{bbm}
\usepackage{mathtools}
\usepackage[all]{xy}

\DeclarePairedDelimiter\ceil{\lceil}{\rceil}
\DeclarePairedDelimiter\floor{\lfloor}{\rfloor}

\newtheorem{theorem}{Theorem}

\newtheorem{proposition}[theorem]{Proposition}
\newtheorem{corollary}[theorem]{Corollary}
\newtheorem{definition}[theorem]{Definition}
\newtheorem{lemma}[theorem]{Lemma}

\newcommand{\cov}{\overset{\mathrm{cov}}{\succ}}
\newcommand{\asucc}{\succ_{\mathrm{a}}}
\newcommand{\locc}{\overset{\mathrm{LOCC}}{\succ}}
\newcommand{\ii}{\mathrm{i}}
\newcommand{\dtv}{d_{\mathrm{TV}}}
\newcommand{\midd}{\,\middle| \,}

\newcommand{\titleofpaper}{Beyond i.i.d. in the Resource Theory of Asymmetry: An Information-Spectrum Approach for Quantum Fisher Information}

\begin{document}

\title{\titleofpaper}
\author{Koji Yamaguchi}
\email{koji.yamaguchi@uwaterloo.ca}
\affiliation{Department of Applied Mathematics, University of Waterloo, Waterloo, Ontario, N2L
3G1, Canada}
\affiliation{Department of Communication Engineering and Informatics,
University of Electro-Communications, 1-5-1 Chofugaoka, Chofu, Tokyo, 182-8585, Japan}

\author{Hiroyasu Tajima}
\affiliation{Department of Communication Engineering and Informatics,
University of Electro-Communications, 1-5-1 Chofugaoka, Chofu, Tokyo, 182-8585, Japan}
\affiliation{JST, PRESTO, 4-1-8 Honcho, Kawaguchi, Saitama, 332-0012, Japan}

\begin{abstract}
Energetic coherence is indispensable for various operations, including precise measurement of time and acceleration of quantum manipulations. Since energetic coherence is fragile, it is essential to understand the limits in distillation and dilution to restore damage. The resource theory of asymmetry (RTA) provides a rigorous framework to investigate energetic coherence as a resource to break time-translation symmetry. Recently, in the i.i.d. regime where identical copies of a state are converted into identical copies of another state, it has been shown that the convertibility of energetic coherence is governed by a standard measure of energetic coherence, called the quantum Fisher information (QFI). This fact means that QFI in the theory of energetic coherence takes the place of entropy in thermodynamics and entanglement entropy in entanglement theory. However, distillation and dilution in realistic situations take place in regimes beyond i.i.d., where quantum states often have complex correlations. Unlike entanglement theory, the conversion theory of energetic coherence in pure states in the non-i.i.d. regime has been an open problem. In this Letter, we solve this problem by introducing a new technique: an information-spectrum method for QFI. Two fundamental quantities, coherence cost and distillable coherence, are shown to be equal to the spectral QFI rates for arbitrary sequences of pure states. As a consequence, we find that both entanglement theory and RTA in the non-i.i.d. regime are understood in the information-spectrum method, while they are based on different quantities, i.e., entropy and QFI, respectively.
\end{abstract}


\maketitle

\textit{\textbf{Introduction.}}---
In quantum mechanics, different states can be superposed. Of vital importance is the superposition between different energy eigenstates, called energetic coherence. Energetic coherence is mandatory for creating accurate clocks \cite{giovannetti_quantum-enhanced_2001,giovannetti_quantum_2006,giacomini_quantum_2019,schnabel_quantum_2010,marvian_coherence_2020,woods_autonomous_2019}, accelerating quantum operations \cite{marvian_quantum_2016}, and measuring physical quantities noncommutative with conserved quantities \cite{wigner_messung_1952,araki_measurement_1960,yanase_optimal_1961,ozawa_conservation_2002,tajima_coherence-variance_2019}.
Recently, it has been shown that energetic coherence also plays important roles for gate implementation in quantum computation \cite{ozawa_conservative_2002,tajima_uncertainty_2018,tajima_coherence_2020,tajima_universal_2021,tajima_2022_universal}, quantum measurements \cite{tajima_coherence-variance_2019, tajima_2022_universal}, quantum error correction \cite{kubica_using_2021,zhou_new_2021,yang_optimal_2022,tajima_universal_2021,liu_quantum_2022,tajima_2022_universal}, and black hole physics \cite{tajima_universal_2021,tajima_2022_universal}.

Energetic coherence and other fundamental properties of quantum systems are better understood by treating them as resources for quantum tasks. 
Quantum resource theories (QRTs) provide a versatile framework for analyzing seemingly unrelated resources with different origins, including entanglement \cite{horodecki_quantum_2009}, athermality \cite{lostaglio_introductory_2019,sagawa_entropy_2022}, and energetic coherence \cite{gour_resource_2008,gour_measuring_2009,marvian_coherence_2020,marvian_operational_2022}. Unexpected similarities arise in different branches of QRTs \cite{chitambar_quantum_2019}, leading to a unified understanding of the underlying laws. Since valuable resources are often fragile, it is fundamental to develop a theoretical understanding of the distillation and dilution of resources to restore their damage. Here, distillation is the operation of extracting as much resource as possible from a given state, and dilution is the opposite (Fig.~\ref{fig:dist_dil}).
\begin{figure}
    \centering
    \includegraphics[width=8.6cm]{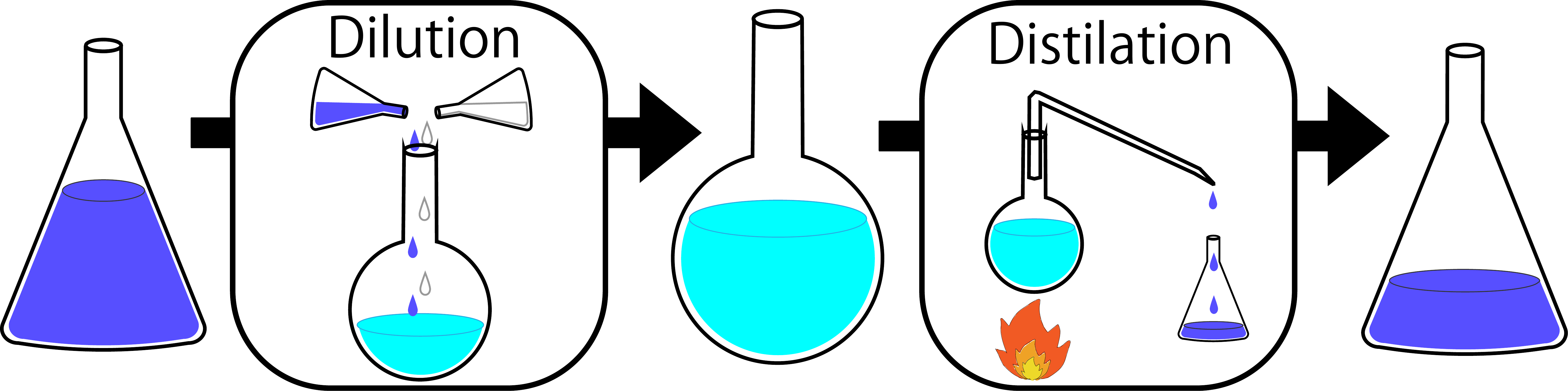}
    \caption{Schematic picture of dilution and distillation. In dilution, a given state or a given sequence of states (depicted as light blue liquid) is generated by consuming as little resource (depicted as dark blue liquid) as possible. In distillation, as much resource as possible is extracted from a given state or a given sequence of states.}
    \label{fig:dist_dil}
\end{figure}

Revealing the limits of distillation and dilution of energetic coherence is of great importance when assembling multiple inaccurate clocks into an accurate clock \cite{marvian_coherence_2020}. It has been studied \cite{marvian_coherence_2020,marvian_operational_2022} in the resource theory of asymmetry (RTA), a branch of QRTs that analyzes symmetries and conservation laws \cite{bartlett_reference_2007,gour_resource_2008,gour_measuring_2009,marvian_mashhad_symmetry_2012,marvian_operational_2022,tajima_uncertainty_2018,tajima_coherence-variance_2019,tajima_coherence_2020,tajima_universal_2021,takagi_correlation_2022,kudo_fisher_2022}. In the independent and identically distributed (i.i.d.) regime where identical copies of a pure state are converted to identical copies of another pure state, the conversion rate is shown to be given by the ratio of their quantum Fisher information (QFI), a quantifier of energetic coherence \cite{marvian_operational_2022}. In other words, QFI is in the position of entropy in the second law of thermodynamics. The same thermodynamic structure is known to exist in the i.i.d. regime for entanglement entropy in entanglement theory \cite{bennett_concentrating_1996}.

Towards practical applications, it is essential to extend the conversion theory in the i.i.d. regime to non-i.i.d. setting because a realistic resource often has complex correlations while an i.i.d. resource state has no correlation. In entanglement theory \cite{hayashi_general_2006,bowen_asymptotic_2008} and quantum thermodynamics \cite{tajima_large_2017,faist_macroscopic_2019,sagawa_asymptotic_2021}, conversion theories in the non-i.i.d. regime have been established. However, the counterpart in RTA remains elusive. 

The obstacle to analyzing non-i.i.d. regime in RTA is the limitation of the traditional information-spectrum method. This method gives a universal way of dealing with entropy-related problems for general states with arbitrary correlations in classical and quantum information theory, e.g., source coding, channel coding and hypothesis testing \cite{ogawa_strong_2000,nagaoka_information-spectrum_2007,hayashi_general_2002,hayashi_general_channel_2003,bowen_beyond_2006}. Furthermore, entanglement theory and quantum thermodynamics in the non-i.i.d. regime \cite{hayashi_general_2006, bowen_asymptotic_2008,jiao_asymptotic_2018,tajima_large_2017,sagawa_asymptotic_2021,faist_macroscopic_2019} are established with the information-spectrum method since they are based on entropy. However, a central measure in converting energetic coherence is QFI, which is quite different from entropy. Therefore, the non-i.i.d. theory in RTA has been out of the scope of the information-spectrum method. 

In this Letter, we establish the conversion theory of energetic coherence in non-i.i.d. pure states by constructing an information-spectrum approach for QFI. The key ingredients we introduce here are the followings: the spectral sup- and inf-QFI rates, the max- and min-QFIs, and asymmetric majorization. All of them clarify the correspondence in the conversion theories of entanglement and energetic coherence in the non-i.i.d. regime, which are characterized by entropies and QFIs, respectively. 
First, we prove a general formula for the coherence cost and the distillable coherence, i.e., the optimal conversion rates of a sequence of arbitrary pure states from and to a reference state. Concretely, they are shown to be equal to the spectral sup- and inf-QFI rates, respectively. This result corresponds to the general formula in entanglement theory \cite{hayashi_general_2006,bowen_asymptotic_2008}, asserting that the entanglement cost and the distillable entanglement are equal to the spectral sup- and inf-entropy rates.
Second, these spectral QFI rates are constructed as asymptotic rates of the smooth max- and min-QFIs. Their construction is parallel to that of spectral entropy rates, given as the asymptotic rates of the smooth max- and min-entropies with the smoothing technique \cite{datta_smooth_2009,renner_security_2005}. Third, the asymmetric majorization relation between energy distribution is shown to provide a necessary and sufficient condition for the exact convertibility among pure states in RTA. This result is the counterpart in RTA to Nielsen's theorem \cite{nielsen_conditions_1999}, which characterizes the pure-state convertibility in entanglement theory by the majorization relation of the Schmidt coefficients. 

Our findings highlight a clear correspondence in non-i.i.d. resource conversion in entanglement theory and RTA. See Figs~\ref{fig:spectral_rates} and \ref{fig:rtoa_rtoe}.  Although they treat quite different resources, i.e., entanglement and energetic coherence, both are understood within a unified framework of the information-spectrum method for each resource. 

\textit{\textbf{Resource theory of asymmetry (RTA).}}---
This Letter aims to construct a general theory of manipulating energetic coherence. To this end, we begin by identifying states with and without energetic coherence. Consider a quantum system $S$ and its Hamiltonian $H$. Energetic coherence means superposition between  eigenstates of $H$ with different eigenvalues. Thus, a state $\rho$ has energetic coherence iff the time evolution $e^{-\ii Ht}$ changes it. Conversely, a state without energetic coherence is symmetric under time evolution, i.e., $e^{-\ii Ht}\rho e^{\ii Ht}=\rho$ for any $t\in \mathbb{R}$. From these facts, we call a state without energetic coherence \textit{symmetric} and a state with energetic coherence \textit{asymmetric}. By definition, a state $\rho$ is symmetric iff $[\rho,H]=0$.

We next consider transformations of states to manipulate energetic coherence. A basic element is an operation which does not create energetic coherence in the sense that it transforms a symmetric state to a symmetric state. 
This condition is satisfied if the operation is described by a CPTP map $\mathcal{E}$ satisfying \footnote{In fact, any completely incoherence-preserving operation satisfies Eq.~\eqref{eq:covariant_channel_definition} \cite{marvian_coherence_2020}. Here, a channel $\mathcal{E}$ is called completely incoherence-preserving iff for any ancillary system $A$ with an arbitrary Hamiltonian $H_A$, the map $\mathcal{E}\otimes \mathcal{I}_A$ transforms any symmetric state to a symmetric state, where $\mathcal{I}_A$ denotes the identity map on $A$.}
\begin{align}
\mathcal{E}\left(e^{-\ii H t}\rho e^{\ii H t}\right)=e^{-\ii H t}\mathcal{E}(\rho)e^{\ii H t},\enskip\forall\rho,\forall t\in\mathbb{R}.\label{eq:covariant_channel_definition}
\end{align}
A channel ${\cal E}$ satisfying Eq.~\eqref{eq:covariant_channel_definition} is called covariant (under time evolution $e^{-\ii Ht}$). 

Based on these ideas, RTA is constructed as a resource theory of energetic coherence. 
The framework of a resource theory is determined by defining ``free states" that can be freely prepared and ``free operations" that can be freely performed.
In RTA, symmetric states are free states, and covariant operations are free operations. 
With these definitions, energetic coherence in asymmetric states becomes a resource. This structure in RTA is the same as in entanglement theory, where entanglement becomes a resource by defining separable states and local operations and classical communication (LOCC) as free states and free operations.

By adopting the above resource-theoretic perspective, coherence is quantified by resource measures, which monotonically decrease under covariant operations. 
A well-known and important one is the symmetric logarithmic derivative Fisher information \cite{helstrom_quantum_1969,holevo_probabilistic_2011} with respect to $\{e^{-\ii Ht}\rho e^{\ii Ht}\}_{t}$, given by
\begin{align}
    \mathcal{F}\left(\rho,H\right)\coloneqq 2 \sum_{i,j}\frac{(\lambda_i-\lambda_j)^2}{\lambda_i+\lambda_j}|\braket{i|H|j}|^2,
\end{align}
where $\rho=\sum_{i}\lambda_i\ket{i}\bra{i}$ is the eigenvalue decomposition. See, e.g., \cite{hansen_metric_2008,takagi_skew_2019,kudo_fisher_2022} for details and its generalization. Hereafter, we call this quantity quantum Fisher information (QFI), simply written as $\mathcal{F}(\rho)$. 
For a pure state, QFI equals four times the variance of $H$ \cite{Comment_convex_roof}.

Following the standard argument \cite{marvian_operational_2022}, we hereafter analyze a system with a Hamiltonian
\begin{align}
    H=\sum_{n=0}^\infty n\ket{n}\bra{n}\label{eq:Hamiltonian},
\end{align}
where $\{\ket{n}\}$ denotes an orthogonal basis.
With the method in Ref. \cite{marvian_operational_2022}, pure-state conversion theory in this system can be extended to a more general setup in RTA with arbitrary Hamiltonians \cite{SM}.

An essential characteristic of a pure state $\psi=\ket{\psi}\bra{\psi}$ in the manipulation of energetic coherence is its energy distribution $p_\psi=\{p_\psi(n)\}_{n=0}^\infty$, where $p_\psi(n)\coloneqq |\braket{n|\psi}|^2$.
This is because any pure state $\ket{\psi}$ can be mapped to $\sum_{n=0}^\infty \sqrt{p_\psi(n)}\ket{n}$ by an energy-conserving unitary operation, which is covariant and invertible. 
In fact, necessary and sufficient conditions for the exact convertibility between pure states have been obtained in terms of the energy distributions \cite{gour_resource_2008,marvian_mashhad_symmetry_2012,marvian_theory_2013}.

From a practical viewpoint, the exact conversion is typically impossible and too restrictive. 
Therefore, it is common to explore the convertibility with vanishing error in the asymptotic regime. 
We adopt the trace distance $D(\rho,\sigma)\coloneqq\frac{1}{2}\|\rho-\sigma\|_1$ as a quantifier of error, where $\|A\|_1\coloneqq \mathrm{Tr}(\sqrt{A^\dag A})$. For $\epsilon\in(0,1]$, we denote $\rho\approx_\epsilon \sigma $ iff two states $\rho$ and $\sigma$ satisfy $D(\rho,\sigma)\leq\epsilon$. For two sequences of states $\widehat{\rho}=\{\rho_m\}_m$ and $\widehat{\sigma}=\{\sigma_m\}_m$, we denote $\widehat{\rho}\cov_{\epsilon}\widehat{\sigma}$ iff there exists a sequence of covariant channels $\{\mathcal{E}_m\}_m$ such that $\mathcal{E}_m(\rho_m)\approx_\epsilon \sigma_m$ for all sufficiently large $m$. If $\widehat{\rho}\cov_{\epsilon}\widehat{\sigma}$ holds for all $\epsilon\in(0,1]$, we say $\widehat{\rho}$ is asymptotically convertible to $\widehat{\sigma}$ and denote $\widehat{\rho}\cov\widehat{\sigma}$. For simplicity, we only analyze systems with Hamiltonian given by Eq.~\eqref{eq:Hamiltonian}. Our main theorem on the pure-state conversion (Theorem~\ref{thm:cost_dist_Finf_F0}) for this setup can be extended to a more general setup with arbitrary Hamiltonians. 
Of course, this includes the i.i.d. case, where a Hamiltonian is given by a sum of copies of a free Hamiltonian of a subsystem. See \cite{SM} for a general formula.

In the analysis of asymptotic convertibility, we adopt a coherence bit, i.e., a qubit with Hamiltonian $\ket{1}\bra{1}$ in a state $\ket{\phi_{\mathrm{coh}}}\coloneqq (\ket{0}+\ket{1})/\sqrt{2}$ as a reference. 
There are two fundamental resource measures: the coherence cost and the distillable coherence. 
They are defined as the optimal rates for converting a sequence of states from and to coherence bits, i.e., 
\begin{align}
        C_{\mathrm{cost}}\left(\widehat{\rho}\right)&\coloneqq \inf\left\{R\mid \widehat{\phi_{\mathrm{coh}}}(R)\cov \widehat{\rho}\right\},\\
     C_{\mathrm{dist}}\left(\widehat{\rho}\right)&\coloneqq \sup\left\{R\mid \widehat{\rho}\cov \widehat{\phi_{\mathrm{coh}}}(R)\right\},
\end{align}
where $\widehat{\phi_{\mathrm{coh}}}(R)\coloneqq \{\phi_{\mathrm{coh}}^{\otimes \ceil{R m}}\}_m$ for $R>0$ and $\phi_{\mathrm{coh}}\coloneqq\ket{\phi_{\mathrm{coh}}}\bra{\phi_{\mathrm{coh}}}$. Note that the infimum of the empty set is formally defined as $+\infty$.

Finally, we introduce several notations for later convenience.
For $a=\{a(n)\}_{n\in\mathbb{Z}}$,
we denote $a\geq 0$ iff $a(n)\geq 0$ for all $n\in\mathbb{Z}$. 
A product sequence $a*b$ is defined by
$a*b(n)\coloneqq\sum_{k\in\mathbb{Z}}a(k)b(n-k)$. 
Similarly, we define $[a*b]_s^t(n)\coloneqq \sum_{k=s}^{t}a(k)b(n-k)$.

For a given sequence $q=\{q(n)\}_{n\in\mathbb{Z}}$, another sequence $\widetilde{q}=\{\widetilde{q}(n)\}_{n\in\mathbb{Z}}$ satisfying
\begin{align}
    \delta_{0,n}=\widetilde{q}*q(n)\label{eq:delta_tildeq_q}
\end{align}
plays a central role in our analysis. Here, $\delta_{m,n}$ is the Kronecker delta. If there exists a finite $n_\star\coloneqq \min\{n\mid q(n)>0\}$, such a sequence is constructed as
\begin{align}
    \widetilde{q}(n)\coloneqq 
    \begin{cases}  
    0&(n< -n_\star)\\
    \frac{1}{q(n_\star)} &(n= -n_\star)\\
    -\frac{1}{q(n_\star)}[\widetilde{q}*q]_{-n_\star}^{n-1}(n_\star+n)&(n>-n_\star)
    \end{cases}.\label{eq:definition_tildeq}
\end{align}

Note that $\widetilde{q}(n)$ is defined recursively for $n>-n_\star$. If $q$ is an energy distribution, $n_\star\geq 0$ exists. In this case, $\widetilde{q}$ satisfies $\sum_n\widetilde{q}(n)=1$. However, it is not a probability distribution in general since it can contain negative elements. Such a sequence $\widetilde{q}$ is utilized to define central quantifiers of our analysis, the max- and min-QFI, just below. 

We also introduce a generalized Poisson distribution $\mathrm{P}_\lambda=\{\mathrm{P}_{\lambda}(n)\}_{n\in\mathbb{Z}}$ for $\lambda\in\mathbb{R}$, where $\mathrm{P}_{\lambda}(n)\coloneqq e^{-\lambda}\lambda^n/n!$ for $n\geq 0$ and $\mathrm{P}_{\lambda}(n)\coloneqq0$ for $n<0$. For $\lambda\geq 0$, $\mathrm{P}_\lambda$ is the ordinary Poission distribution. Although $\mathrm{P}_\lambda$ with negative $\lambda$ is not a probability distribution, this notation is useful since $\widetilde{\mathrm{P}}_\lambda=\mathrm{P}_{-\lambda}$ \cite{SM}.



\textit{\textbf{Main results.}}---
Now, let us construct an information-spectrum theory for QFI and show our main results.
We first introduce key quantities. For a pure state $\psi$, we define the max-QFI $\mathcal{F}_{\mathrm{max}}$ and the min-QFI $\mathcal{F}_{\mathrm{min}}$ by
\begin{align}
    \mathcal{F}_{\mathrm{max}}(\psi)&\coloneqq \inf\left\{4\lambda\mid\mathrm{P}
    _\lambda* \widetilde{p_{\psi}}\geq 0\right\},\label{eq:max_F_pure}\\
    \mathcal{F}_{\mathrm{min}}(\psi)&\coloneqq\sup \left\{4\lambda\mid p_{\psi}*\mathrm{P}_{-\lambda}\geq 0\right\}.
\end{align}
They quantify the amounts of coherence in $\psi$ transformable from and to a state whose energy distribution follows a Poisson distribution \cite{SM}.
For a general state $\rho$, we define the max-QFI by $\mathcal{F}_{\mathrm{max}}(\rho)\coloneqq\inf_{\Phi_\rho} \mathcal{F}_{\mathrm{max}}(\Phi_\rho)$, where the infimum is taken over the set of all purifications $\Phi_\rho$ of $\rho$ and the Hamiltonians of the auxiliary system with integer eigenvalues. This notation is consistent with that for pure states \cite{SM}. 

The max- and min-QFI have similar properties to the max- and min-entropies in entanglement theory \cite{SM}.
For example, they provide the upper and lower bounds for QFI:
\begin{align}
    \mathcal{F}_{\mathrm{max}}(\psi)\geq \mathcal{F}(\psi)\geq \mathcal{F}_{\mathrm{min}}(\psi). \label{eq:Finf_F_F0}
\end{align}

For a general sequence of pure states $\widehat{\psi}=\{\psi_m\}$, the spectral sup- and inf-QFI rates are defined as
\begin{align}
    \overline{\mathcal{F}}(\widehat{\psi})&\coloneqq \lim_{\epsilon\to+0}\limsup_{m\to\infty}\frac{1}{m}\mathcal{F}_{\max}^\epsilon (\psi_m),\label{eq:spectral_sup_rate}\\
    \underline{\mathcal{F}}(\widehat{\psi})&\coloneqq \lim_{\epsilon\to+0}\liminf_{m\to\infty}\frac{1}{m}\mathcal{F}_{\min}^\epsilon (\psi_m)\label{eq:spectral_inf_rate},
\end{align}
where $\mathcal{F}_{\mathrm{max}}^\epsilon (\psi)\coloneqq \inf_{\rho\in B^{\epsilon}(\psi)}\mathcal{F}_{\mathrm{max}}(\rho)$ and $\mathcal{F}_{\mathrm{min}}^\epsilon (\psi)\coloneqq \sup_{\phi\in B_{\mathrm{pure}}^\epsilon(\psi)}\mathcal{F}_{\mathrm{min}}(\phi)$ are smooth max- and min-QFI. Here, we defined $B^\epsilon(\psi)\coloneqq \{\rho:\text{states}|D(\rho,\psi)\leq \epsilon\}$ and $B_{\mathrm{pure}}^\epsilon(\psi)\coloneqq \{\phi\text{: pure states}|D(\phi,\psi)\leq \epsilon\}$.
Note that $\limsup_{m\to\infty}\frac{1}{m}\mathcal{F}_{\mathrm{max}}^\epsilon (\psi_m)$ and $\liminf_{m\to\infty}\frac{1}{m}\mathcal{F}_{\mathrm{min}}^\epsilon (\psi_m)$ monotonically increases and decreases as $\epsilon$ becomes smaller, and hence limit values $\overline{\mathcal{F}}(\psi)$ and $\underline{\mathcal{F}}(\psi)$ exist.

The main theorem of this Letter is the following:
\begin{theorem}\label{thm:cost_dist_Finf_F0}
For a general sequence of pure states $\widehat{\psi}=\{\psi_m\}$, the coherence cost and the distillable coherence are equal to the spectral sup- and inf-QFI rates, respectively. That is,
\begin{align}
C_{\mathrm{cost}}(\widehat{\psi})=\overline{\mathcal{F}}(\widehat{\psi}),\quad C_{\mathrm{dist}}(\widehat{\psi})=\underline{\mathcal{F}}(\widehat{\psi}).    
\end{align}
\end{theorem}
As a corollary of Theorem~\ref{thm:cost_dist_Finf_F0}, we immediately get \cite{SM}
\begin{align}
    &\widehat{\psi}\cov\widehat{\phi}\implies \overline{\mathcal{F}}(\widehat{\psi})\geq \overline{\mathcal{F}}(\widehat{\phi}),\quad \underline{\mathcal{F}}(\widehat{\psi})\geq \underline{\mathcal{F}}(\widehat{\phi}),\label{eq:asymptotic_nec}\\
    &\underline{\mathcal{F}}(\widehat{\psi})> \overline{\mathcal{F}}(\widehat{\phi})\implies \widehat{\psi}\cov\widehat{\phi}.\label{eq:asymptotic_suf}
\end{align}
Replacing $\overline{\mathcal{F}}$, $\underline{\mathcal{F}}$ and $\cov$ by $\overline{S}$, $\underline{S}$ and $\locc$, the same relations as Eqs.~\eqref{eq:asymptotic_nec} and \eqref{eq:asymptotic_suf} hold in entanglement theory. 
Here, $\overline{S}$ and $\underline{S}$ denote the spectral sup- and inf-entropy rates, while $\widehat{\psi}\locc \widehat{\phi}$ means that $\widehat{\psi}$ is asymptotically convertible to $\widehat{\phi}$ by LOCC \cite{SM}. 

Theorem~\ref{thm:cost_dist_Finf_F0} for a system with Hamiltonian in Eq.~\eqref{eq:Hamiltonian} can be extended to an arbitrary sequence of systems with any Hamiltonians in pure states having a finite period \cite{SM}. In particular, the spectral QFI rates $\overline{\mathcal{F}}$ and $\underline{\mathcal{F}}$ are equal to QFI $\mathcal{F}$ in the i.i.d. setting \cite{SM}, which reproduces the result in earlier i.i.d. studies \cite{gour_resource_2008,marvian_operational_2022}. 
We remark that $\overline{S}$ and $\underline{S}$ are equal to entanglement entropy in the i.i.d. regime in entanglement theory \cite{SM}.

These results show that the spectral sup- and inf-QFI rates, $\overline{\mathcal{F}} $ and $\underline{\mathcal{F}} $, in RTA play the same roles as the spectral sup- and inf-entropy rates, $\overline{S}$ and $\underline{S}$, in entanglement theory \cite{SM}. See Fig.~\ref{fig:spectral_rates}.
In other words, RTA in the non-i.i.d. regime has the same structure on convertibility as Lieb-Yngvason's non-equilibrium theory \cite{lieb_entropy_2013}, based on QFI-related quantities rather than entropies.

\begin{figure}
    \centering
    \includegraphics[width=8.6cm]{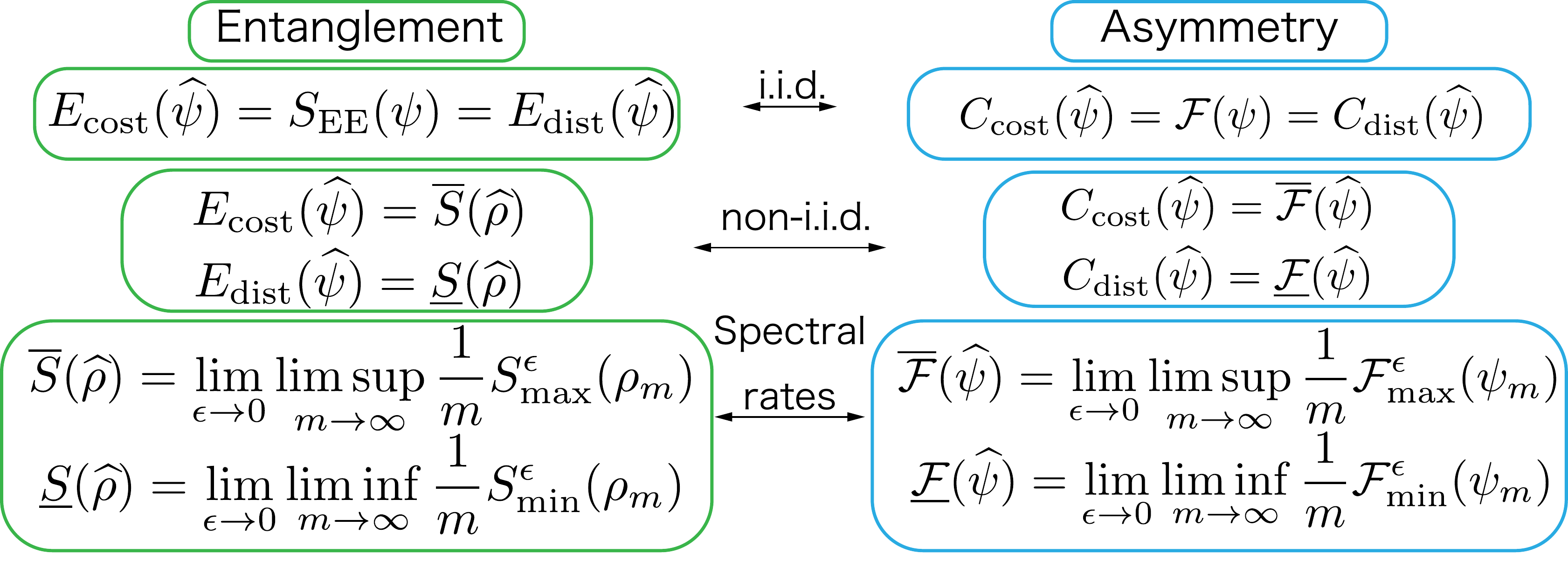}
    \caption{Comparison of entanglement theory and RTA on the asymptotic convertibility \cite{SM}. 
    The quantities $E_{\mathrm{cost}}$ and $E_{\mathrm{dist}}$ denote the entanglement cost and the distillable entanglement.
    For a sequences of general bipartite pure states $\widehat{\psi}=\{\psi_{AB,m}\}_m$, we define $\widehat{\rho}=\{\rho_m\}_m$, where $\rho_m\coloneqq \mathrm{Tr}_B(\psi_{AB,m})$.
    The quantities $S_{\mathrm{EE}}(\psi)$, $S_{\mathrm{max}}^\epsilon(\rho)$ and $S_{\mathrm{min}}^\epsilon(\rho)$ denote entanglement entropy, the smooth max- and min-entropies.}
    \label{fig:spectral_rates}
\end{figure}

Theorem~\ref{thm:cost_dist_Finf_F0} for the coherence cost is directly extended to general states, including mixed states. That is, defining $\overline{\mathcal{F}}\left(\widehat{\rho}\right)\coloneqq \lim_{\epsilon\to+0}\limsup_{m\to\infty}\frac{1}{m}\mathcal{F}_{\mathrm{max}}^\epsilon (\rho_m)$, where $\mathcal{F}_{\mathrm{max}}^\epsilon (\rho)\coloneqq\inf_{\sigma\in B^\epsilon(\rho)}\mathcal{F}_{\mathrm{max}}(\sigma)$, the following holds \cite{SM}:
\begin{theorem}\label{thm:cost_general}
For a general sequence of states $\widehat{\rho}=\{\rho_m\}$, it holds $C_{\mathrm{cost}}\left(\widehat{\rho}\right)=\overline{\mathcal{F}}\left(\widehat{\rho}\right)$.
\end{theorem}

\textit{\textbf{One-shot convertibility between pure states.}}---
We here define a notion of asymmetric majorization, which we abbreviate a-majorization, as follows:
\begin{definition}
For probability distributions $p=\{p(n)\}_{n\in\mathbb{Z}}$ and $q=\{q(n)\}_{n\in\mathbb{Z}}$, we say that $p$ a-majorizes $q$ iff $p*\widetilde{q}\geq 0$ hold. In this case, we denote $p\asucc q$. 
\end{definition}

For comparison, we review the definition of majorization. 
A probability distribution $p=\{p(i)\}_{i=1}^d$ majorizes another probability distribution $q=\{q(i)\}_{i=1}^d$ iff $\sum_{l=1}^k p^\downarrow(l)\geq \sum_{l=1}^k q^\downarrow(l)$ for all $k=1,\cdots,d$, where $\downarrow$ indicates that the distributions are rearranged in decreasing order so that $p^\downarrow(i)\geq p^\downarrow(j)$ and $q^\downarrow(i)\geq q^\downarrow(j)$ for $i>j$. 

The a-majorization has properties similar to the ordinary majorization \cite{SM}. Among them, a significant one is the following:
\begin{theorem}\label{thm:convertibility_a_majorization}
A pure state $\psi$ is convertible to a pure state $\phi$ by a covariant operation iff $p_\psi\asucc p_\phi$. 
\end{theorem}

This is the counterpart in RTA to Nielsen's theorem in entanglement theory \cite{nielsen_conditions_1999}:
A bipartite pure state $\psi$ is convertible to a bipartite pure state $\phi$ by LOCC iff $\lambda_\psi\prec \lambda_\phi$, where $\lambda_\psi$ and $\lambda_\phi$ are the probability distributions given by the Schmidt coefficients of $\psi$ and $\phi$, respectively. This correspondence is the motivation for introducing the terminology of a-majorization.
See Fig.~\ref{fig:rtoa_rtoe}.

We remark that other necessary and sufficient conditions on one-shot convertibility in RTA were proven in earlier studies \cite{gour_resource_2008,marvian_theory_2013,marvian_mashhad_symmetry_2012}.
Our contribution here is to provide the one-shot convertibility condition in terms of a-majorization to make it useful for our purpose to analyze the asymptotic convertibility in the non-i.i.d. regime. 
In particular, this reformulation makes the correspondence between RTA and entanglement theory clearer. 

\begin{figure}
    \centering
    \includegraphics[width=8.6cm]{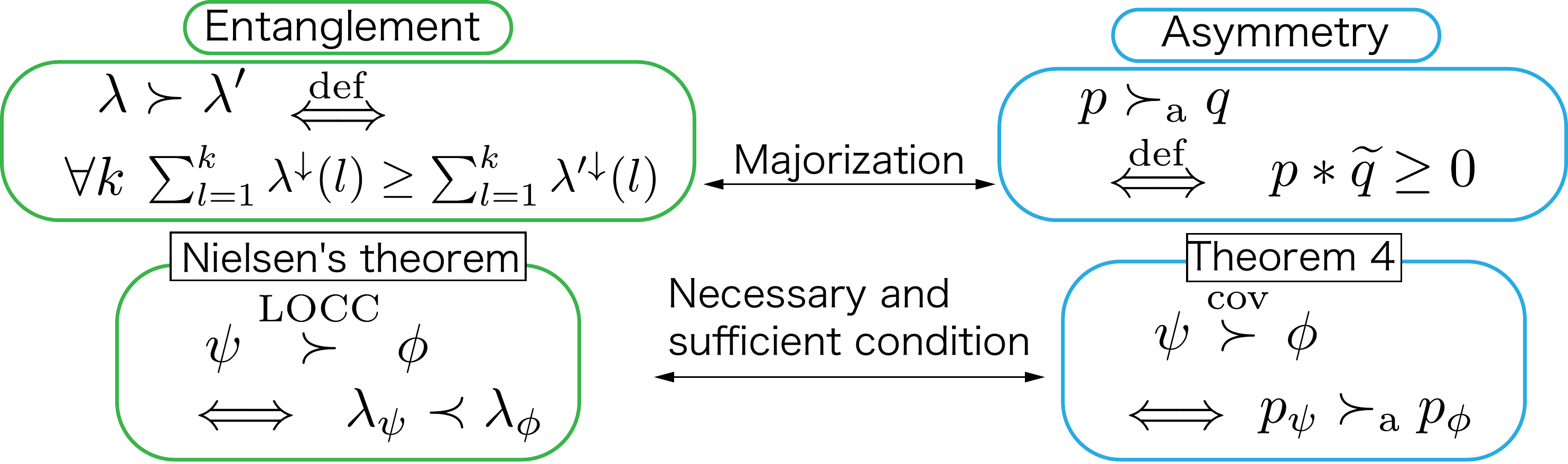}
    \caption{Comparison of entanglement theory and RTA on the one-shot convertibility. If a pure state $\psi$ is exactly convertible to another pure state $\phi$ by LOCC (resp. covariant operations), we denote $\psi \locc \phi$ (resp. $\psi\cov\phi$). }
    \label{fig:rtoa_rtoe}
\end{figure}

\textit{\textbf{Proof of Theorem~\ref{thm:cost_dist_Finf_F0}.}}---
For a Poisson distribution $\mathrm{P}_{\lambda}$, we denote $\chi_{\lambda}\coloneqq \sum_{n,n'=0}^\infty\sqrt{\mathrm{P}_\lambda(n)\mathrm{P}_\lambda(n')}\ket{n}\bra{n'}$ and $\widehat{\chi}_{\lambda}\coloneqq \{\chi_{\lambda m}\}_m$. This sequence is interconvertible with $\widehat{\phi_{\mathrm{coh}}}(R)$ by covariant operations, i.e., $\widehat{\phi_{\mathrm{coh}}}(R)\cov\widehat{\chi}_{R/4}$ and $\widehat{\chi}_{R/4}\cov\widehat{\phi_{\mathrm{coh}}}(R)$ \cite{SM}. 

The followings are key lemmas \cite{SM}:
\begin{lemma}\label{lem:purification_majorization}
Let $\mathcal{E}$ be a covariant channel. A state $\mathcal{E}(\chi_\lambda)$ has a purification $\Psi$ such that $p_{\Psi}=\mathrm{P}_\lambda$, where the Hamiltonian of the ancilla added to purify $\mathcal{E}(\chi_\lambda)$ has integer eigenvalues. 
\end{lemma}
\begin{lemma}\label{lem:purification_smooth}
Let $\psi$ and $\phi$ be pure states. Assume that a covariant channel satisfies $\mathcal{E}(\psi)\approx_\epsilon\phi$. Then there exists a pure state $\psi'$ such that $\psi'\in B^{2\epsilon^{1/4}}_{\mathrm{pure}}(\psi)$ and $\quad p_{\psi'}\asucc p_\phi$.
\end{lemma}

To show $C_{\mathrm{cost}}(\widehat{\psi})=\overline{\mathcal{F}}(\widehat{\psi})$, we introduce $C_{\mathrm{cost}}^\epsilon\left(\widehat{\rho}\right)\coloneqq \inf\left\{R\middle| \widehat{\phi_{\mathrm{coh}}}(R)\cov_{\epsilon}\widehat{\rho}\right\}$. 
Defining $4\lambda^{\epsilon/2}\coloneqq C_{\mathrm{cost}}^{\epsilon/2}(\widehat{\psi})$, for any $\delta>0$, there exists $\delta'$ such that $\delta> \delta'\geq 0$ and $\widehat{\phi_{\mathrm{coh}}}(4(\lambda^{\epsilon/2}+\delta'))\cov_{\epsilon/2}\widehat{\psi}$. Since $\widehat{\chi}_{(\lambda^{\epsilon/2}+\delta')}\cov_{\epsilon/2}\widehat{\phi_{\mathrm{coh}}}(4(\lambda^{\epsilon/2}+\delta'))$, we have $\widehat{\chi}_{(\lambda^{\epsilon/2}+\delta)}\cov_{\epsilon}\widehat{\psi}$, where we have used the fact that $\mathrm{P}_{\lambda}\asucc \mathrm{P}_{\lambda'}$ holds for any $\lambda\geq \lambda'$ \cite{SM}. 
From Lemma~\ref{lem:purification_majorization}, for all sufficiently large $m$, there exists a state $\rho_m\in B^\epsilon(\psi_m)$ whose purification $\Phi_m$ satisfies $\mathrm{P}_{(\lambda^{\epsilon/2}+\delta)m}= p_{\Phi_m}$. Therefore, $4(\lambda^{\epsilon/2}+\delta)m\geq \mathcal{F}_{\mathrm{max}}^\epsilon\left(\psi_m\right)$ for sufficiently large $m$, which implies
\begin{align}
    \forall \delta>0, \quad C_{\mathrm{cost}}^{\epsilon/2}(\widehat{\psi})+4\delta\geq \limsup_{m\to\infty}\frac{1}{m}\mathcal{F}_{\mathrm{max}}^\epsilon(\psi_m).
\end{align}
As $\epsilon\to +0$, we get $C_{\mathrm{cost}}(\widehat{\psi})\geq \overline{\mathcal{F}}(\widehat{\psi})$.

To show the opposite inequality, 
we define $4\lambda_m^{\epsilon/2}\coloneqq \mathcal{F}^{\epsilon/2}_{\mathrm{max}}(\psi_m)$. For any $\delta>0$, there exist $\delta_m'$, satisfying $\delta>\delta_m'\geq0$, such that there exist a state $\rho_m\in B^{\epsilon/2}(\psi_m)$ and its purification $\Phi_m$ satisfying $\mathrm{P}_{(\lambda^{\epsilon/2}_m+\delta_m')}\asucc p_{\Phi_m}$. 
Note that for all sufficiently large $m$, $m(\lambda^{\epsilon/2}+\delta)\geq \lambda_m^{\epsilon/2}$ holds for $\lambda^{\epsilon/2}\coloneqq \limsup_{m\to\infty}\frac{1}{m}\lambda_m^{\epsilon/2}$. Therefore, we get $\mathrm{P}_{(\lambda^{\epsilon/2}+2\delta)m}\asucc p_{\Phi_m}$, where we have used $m\delta > \delta_m'$.
Since $\widehat{\phi_{\mathrm{coh}}}(4(\lambda^{\epsilon/2}+2\delta))\cov_{\epsilon/2}\widehat{\chi}_{(\lambda^{\epsilon/2}+2\delta)}$, we have $\widehat{\phi_{\mathrm{coh}}}(4(\lambda^{\epsilon/2}+2\delta))\cov_{\epsilon/2}\{\Phi_m\}_m$. 
Since the partial trace is a covariant operation, we have $\widehat{\phi_{\mathrm{coh}}}(4(\lambda^{\epsilon/2}+2\delta))\cov_{\epsilon}\widehat{\psi}$. Therefore, 
\begin{align}
     \forall \delta>0, \quad 
     C_{\mathrm{cost}}^\epsilon (\widehat{\psi})\leq \limsup_{m\to\infty}\frac{1}{m}\mathcal{F}_{\mathrm{max}}^{\epsilon/2}(\psi_m)+8\delta.
\end{align}
As $\epsilon\to +0$, we get $ C_{\mathrm{cost}} (\widehat{\psi})\leq \overline{\mathcal{F}}(\widehat{\psi})$. Therefore, $C_{\mathrm{cost}} (\widehat{\psi})=\overline{\mathcal{F}}(\widehat{\psi})$.

To show $C_{\mathrm{dist}}(\widehat{\psi})=\underline{\mathcal{F}}(\widehat{\psi})$, we introduce $C_{\mathrm{dist}}^\epsilon(\widehat{\psi})\coloneqq \sup\{R\mid \widehat{\psi}\cov_{\epsilon} \widehat{\phi_{\mathrm{coh}}}(R)\}$.
Defining $4\lambda^{\epsilon/2}\coloneqq C_{\mathrm{dist}}^{\epsilon/2}(\widehat{\psi})$, for any $\delta>0$, there exists $\delta'$ such that $\delta>\delta'\geq 0$ and $\widehat{\psi}\cov_{\epsilon/2}\widehat{\phi_{\mathrm{coh}}}(4(\lambda^{\epsilon/2}-\delta'))$. Since $\widehat{\phi_{\mathrm{coh}}}(4(\lambda^{\epsilon/2}-\delta'))\cov_{\epsilon/2}\widehat{\chi}_{\lambda^{\epsilon/2}-\delta'}$, we get $\widehat{\psi}\cov_{\epsilon}\widehat{\chi}_{\lambda^{\epsilon/2}-\delta'}$. 
From Lemma~\ref{lem:purification_smooth}, for all sufficiently large $m$, there exist pure states $\psi_m'\in B_{\mathrm{pure}}^{2\epsilon^{1/4}}(\psi_m)$ such that $\psi_m'\asucc \mathrm{P}_{(\lambda^{\epsilon/2}-\delta)m}$, where we used $\delta>\delta'$. Therefore, 
\begin{align}
     \forall \delta>0, \quad C_{\mathrm{dist}}^{\epsilon/2}(\widehat{\psi}) -4\delta\leq \liminf_{m\to\infty}\frac{1}{m}\mathcal{F}_{\mathrm{min}}^{2\epsilon^{1/4}}(\widehat{\psi}).
\end{align}
As $\epsilon\to +0$, we get $C_{\mathrm{dist}}(\widehat{\psi})\leq \underline{\mathcal{F}}(\widehat{\psi})$.

To show the opposite inequality,
we define $4\lambda_m^{\epsilon/2}\coloneqq \mathcal{F}_{\mathrm{min}}(\psi_m)$. 
For any $\delta>0$, there exists $\delta'_m$, satisfying $\delta>\delta'_m\geq0$, such that there exists a pure state $\psi_m'\in B_{\mathrm{pure}}^{\epsilon/2} (\psi_m)$ satisfying $p_{\psi'_m}\asucc\mathrm{P}_{\lambda_m^{\epsilon/2}-\delta'_m}$. 
For all sufficiently large $m$, $\lambda_m^{\epsilon/2}\geq m (\lambda^{\epsilon/2}-\delta)$, where $\lambda^{\epsilon/2}\coloneqq \liminf_{m\to\infty}\frac{1}{m}\lambda_m^{\epsilon/2}$. Since $m\delta>\delta_m'$, we have $p_{\psi'_m}\asucc\mathrm{P}_{(\lambda^{\epsilon/2}-2\delta)m}$. By using
$\widehat{\chi}_{\lambda^{\epsilon/2}-2\delta}\cov_{\epsilon/2}\widehat{\phi_{\mathrm{coh}}}(4(\lambda^{\epsilon/2}-2\delta))$ and $\psi_m'\in B^{\epsilon/2}(\psi_m)$, we get $\widehat{\psi}\cov_\epsilon\widehat{\phi_{\mathrm{coh}}}(4(\lambda^{\epsilon/2}-2\delta))$, which implies
\begin{align}
    \forall \delta>0,\quad C_{\mathrm{dist}}^{\epsilon}(\widehat{\psi})\geq \liminf_{m\to\infty}\frac{1}{m}\mathcal{F}_{\mathrm{min}}^{\epsilon/2}(\psi_m)-8\delta.
\end{align}
As $\epsilon \to +0$, we get $C_{\mathrm{dist}}(\widehat{\psi})\geq \underline{\mathcal{F}}(\widehat{\psi})$. Therefore, $C_{\mathrm{dist}}(\widehat{\psi})= \underline{\mathcal{F}}(\widehat{\psi})$.

\textit{\textbf{Conclusion and Discussions.}}---
In this Letter, we established the pure-state conversion theory in RTA in the asymptotic non-i.i.d. regime. Unlike entanglement theory, the traditional information-spectrum method for entropy cannot be applied to RTA since its standard measure, QFI, is quite different from entropy. 
To overcome this issue, we constructed an information-spectrum approach for QFI by carefully analyzing the correspondence between RTA and entanglement theory. It opens the possibility of exploring a unified understanding of asymptotic conversion theory in each branch of quantum resource theories by extending the information-spectrum method for its resource measure. 
Such an extension may trigger research that has been out of the scope of the information-spectrum method.
We speculate that the information-spectrum approach for QFI can be helpful in research areas where QFI plays an essential role, such as in non-equilibrium thermodynamics \cite{ito_stochastic_2020} and general resource theories \cite{tan_fisher_2021}. 


\let\oldaddcontentsline\addcontentsline
\renewcommand{\addcontentsline}[3]{}

\begin{acknowledgments}
The authors thank Achim Kempf for a fruitful discussion. 
KY acknowledges support from the JSPS Overseas Research Fellowships. 
HT acknowledges supports from JSPS Grants-in-Aid for Scientific Research (JP19K14610), JST PRESTO (JPMJPR2014), and JST MOONSHOT (JPMJMS2061).
\end{acknowledgments}


%

\let\addcontentsline\oldaddcontentsline

\appendix
\widetext

\clearpage

\setcounter{page}{1}
\begin{center}
{\large \bf Supplemental Material for \protect \\ 
``\titleofpaper''}\\
\vspace*{0.3cm}
Koji Yamaguchi$^{1}$ and Hiroyasu Tajima$^{2,3}$\\
\vspace*{0.1cm}
$^{1}${\small \it Department of Applied Mathematics, University of Waterloo, Waterloo, Ontario, N2L
3G1, Canada}
\\
$^{2}${\small \it Department of Communication Engineering and Informatics, University of Electro-Communications, 1-5-1 Chofugaoka, Chofu, Tokyo, 182-8585, Japan}
\\
$^{3}${\small \it JST, PRESTO, 4-1-8 Honcho, Kawaguchi, Saitama, 332-0012, Japan}
\end{center}

\renewcommand{\theequation}{S.\arabic{equation}}
\setcounter{equation}{0}

\tableofcontents

\vspace{1cm}
In this Supplemental Material, we first complete the proof of theorems in the main text. We then relate the asymptotic convertibility of pure states in harmonic oscillator systems to that of pure states having a finite period. By using these results, we extend Theorem~\ref{thm:cost_dist_Finf_F0} to an arbitrary sequence of systems with any Hamiltonian in pure states having a finite period. Finally, we show that the information-spectral QFI rates are equal to the Fisher information for i.i.d. sequence of pure states with a finite period, which is consistent with the results in the i.i.d. setting \cite{gour_resource_2008,marvian_operational_2022}.

\section{The generating function of the Poisson distribution and its reciprocal}
For a sequence $q=\{q(n)\}_{n\in\mathbb{Z}}$, we have defined $\widetilde{q}$ as a sequence that satisfies Eq.~\eqref{eq:delta_tildeq_q}. Although it can be constructed recursively by Eq.~\eqref{eq:definition_tildeq}, the method of the generating function \cite{wilf_generatingfunctionology_2011} is sometimes useful. For simplicity, we assume that $n_\star=0$. That is, $q(n)=0$ for $n<0$ and $q(0)\neq 0$. 
A generating function of $q=\{q(n)\}_{n=0}^\infty$ is defined as a formal series given by
\begin{align}
    f(z)\coloneqq \sum_{n=0}^\infty q(n)z^n. 
\end{align}
Its reciprocal $1/f(z)$ that satisfies
\begin{align}
    1=f(z)\times 1/f(z)\label{eq:reciprocal}
\end{align}
for all $z$ exists if and only if $q(0)\neq 0$ \cite{wilf_generatingfunctionology_2011}. Let $a=\{a(n)\}_{n=0}^\infty$ be a sequence generated by $1/f(z)$. From Eq.~\eqref{eq:reciprocal}, it satisfies
\begin{align}
    \delta_0=a*q,
\end{align}
where we have defined $a(n)=0$ for $n<0$ and $\delta_0$ as a sequence defined by $\delta_0=\{\delta_{n,0}\}_{n\in\mathbb{Z}}$. In other words, $a$ is the same as the sequence $\widetilde{q}$ defined in Eq.~\eqref{eq:definition_tildeq}. 
We remark that a similar technique, based on characteristic functions instead of generating functions, has been used in \cite{marvian_theory_2013,marvian_mashhad_symmetry_2012} to derive a necessary and sufficient condition on the one-shot convertibility in RTA.

Let us apply this method to the Poisson distribution.
For a sequence $\{\mathrm{P}_{\lambda}(n)\}_{n=0}^\infty$, its generating function is calculated as
\begin{align}
    f(z)=\sum_{n=0}^\infty \mathrm{P}_{\lambda}(n)z^n=e^{-\lambda}e^{\lambda z}.
\end{align}
Therefore, its reciprocal is given by
\begin{align}
    1/f(z)=e^{-(-\lambda)}e^{(-\lambda) z}.
\end{align}
Since
\begin{align}
    \sum_{n=0}^\infty e^{-(-\lambda)}\frac{(-\lambda)^n}{n!}z^n=e^{-(-\lambda)}e^{(-\lambda) z}=1/f(z),
\end{align}
the sequence generated by $1/f(z)$ is given by $\{\mathrm{P}_{-\lambda}(n)\}_{n=0}^\infty$. That is, $\widetilde{\mathrm{P}_{\lambda}}=\mathrm{P}_{-\lambda}$.

This result can also be checked directly: Let $\lambda$ and $\lambda'$ be real parameters. A straightforward calculation shows that
\begin{align}
    \sum_{k\in\mathbb{Z}}\mathrm{P}_{\lambda}(k)\mathrm{P}_{\lambda'}(n-k)
    &=\begin{cases}
    \sum_{k=0}^n\frac{\lambda^k}{k!}\frac{\lambda^{\prime (n-k)}}{(n-k)!}&\quad (n\geq 0)\\
    0&\quad (n<0)
    \end{cases}\\
    &=\mathrm{P}_{\lambda+\lambda'}(n).\label{eq:poisson_convolution}
\end{align}
Since $\mathrm{P}_0=\delta_{0}$, we get $\widetilde{\mathrm{P}_{\lambda}}=\mathrm{P}_{-\lambda}$. 

Another immediate consequence of Eq.~\eqref{eq:poisson_convolution} is the fact that
\begin{align}
    \lambda\geq \lambda'\geq 0\Longleftrightarrow\mathrm{P}_{\lambda}\asucc\mathrm{P}_{\lambda'}.\label{eq:poisson_NScond}
\end{align}
This is because $\mathrm{P}_{\sigma}(n)$ is non-negative for all $n$ if and only if $\sigma\geq 0$.

\section{Proof of Theorem~\ref{thm:convertibility_a_majorization}}
To prove Theorem~\ref{thm:convertibility_a_majorization}, a key theorem is the following:
\begin{theorem}\label{thm:c-majorization}
For two probability distributions $\{p(n)\}_{n=0}^\infty$ and $\{q(n)\}_{n=0}^\infty$, conditions (i) and (ii) are equivalent:
\begin{enumerate}[(i)]
    \item There exists a probability distribution $\{w(k)\}_{k\in\mathbb{Z}}$ such that $p=\sum_{k\in\mathbb{Z}}w(k)\Upsilon_k q$, where $\Upsilon_k$ is a shift operator on probability distribution such that $(\Upsilon_kp)(n)=p(n-k)$.
    \item $p\asucc q$.
\end{enumerate}
\end{theorem}
This theorem corresponds to the Hardy-Littlewood-P\'olya theorem \cite{hardy_simple_1929} in the theory of majorization, which states that the following conditions (a) and (b) are equivalent: (a) There exists a doubly stochastic matrix $D$ such that $p=Dq$. (b) $p$ majorizes $q$, i.e., $p\succ q$. The correspondence becomes more clear by using Birkhoff's theorem \cite{birkhoff_tres_1946}, which states that any doubly stochastic matrix $D$ can be written as $D=\sum_{k}r(k)P_k$, where $\{r(k)\}$ is a probability distribution, $P_k$ are the permutation matrices, and the sum is taken over the set of all permutation matrices. For details, see, e.g., \cite{marshall_inequalities_2011}.
\begin{proof}[Proof of Theorem~\ref{thm:c-majorization}]
For a probability distribution $\{q(n)\}_{n\in\mathbb{Z}}$, the sequence $\widetilde{q}$ defined in Eq.\eqref{eq:definition_tildeq} satisfies Eq.~\eqref{eq:delta_tildeq_q}.
Therefore,
\begin{align}
    \sum_{n\in\mathbb{Z}}\delta_0(n)=\sum_{n\in\mathbb{Z}}\sum_{k\in\mathbb{Z}}\widetilde{q}(k)q(n-k)=\left(\sum_{k\in\mathbb{Z}}\widetilde{q}(k)\right)\left(\sum_{n\in\mathbb{Z}}q(n)\right),
\end{align}
which implies
\begin{align}
    \sum_{k\in\mathbb{Z}}\widetilde{q}(k)=\frac{1}{\sum_{n\in\mathbb{Z}}q(n)}=1.
\end{align}

For a probability distribution $p=\{p(n)\}_{n\in\mathbb{Z}}$, from Eq.~\eqref{eq:delta_tildeq_q}, we have
\begin{align}
    p(n)&=\sum_{l\in\mathbb{Z}}\delta_{0}(n-l)p(l)\\
    &=\sum_{l\in\mathbb{Z}}\sum_{k\in\mathbb{Z}}\widetilde{q}(k)q(n-l-k)p(l)\\
    &=\sum_{k\in\mathbb{Z}}\left(\sum_{l\in\mathbb{Z}}p(l)\widetilde{q}(k-l)\right)q(n-k)\\
    &=\sum_{k\in\mathbb{Z}}\left(p*\widetilde{q}\right)(k)q(n-k).
\end{align}
In addition,
\begin{align}
    \sum_{k\in\mathbb{Z}}\left(p*\widetilde{q}\right)(k)=\sum_{l\in\mathbb{Z}}p(l)\sum_{k\in\mathbb{Z}}\widetilde{q}(k-l)=1
\end{align}
holds since $\sum_{k\in\mathbb{Z}}p(k)=\sum_{k\in\mathbb{Z}}\widetilde{q}(k)=1$. 

Suppose that 
\begin{align}
    p(n)=\sum_{k\in\mathbb{Z}}w(k)q(n-k)\label{eq:weighted_sum}
\end{align}
holds for some $w=\{w(k)\}_{k\in\mathbb{Z}}$. Then we have
\begin{align}
    \sum_{l\in\mathbb{Z}}p(l)\widetilde{q}(k-l)=\sum_{l\in\mathbb{Z}}\sum_{n\in\mathbb{Z}}w(n)q(l-n)\widetilde{q}(k-l)=\sum_{n\in\mathbb{Z}}w(n)\delta_{k-n,0}=w(k)\label{eq:w},
\end{align}
i.e., $w=p*\widetilde{q}$. 
Thus, the sequence $w$ satisfying Eq.~\eqref{eq:weighted_sum} is unique and given by $w=p*\widetilde{q}$. 

Since $w=p*\widetilde{q}$ satisfies $\sum_{k\in\mathbb{Z}}w(k)=1$, there exists a sequence $\{w(k)\}_{k\in\mathbb{Z}}$ that satisfies Eq.~\eqref{eq:weighted_sum}, $w(k)\geq 0$ and $\sum_{k\in\mathbb{Z}}w(k)=1$ if and only if
\begin{align}
    p*\widetilde{q}\geq 0.
\end{align}
\end{proof}

Theorem~\ref{thm:convertibility_a_majorization} is obtained as a corollary of Theorem~\ref{thm:c-majorization} and a theorem in \cite{gour_resource_2008} on the convertibility:
\begin{theorem}[Theorem~3 in \cite{gour_resource_2008}]\label{thm:Gour_Spekkens}
A pure state $\psi$ is convertible to a pure state $\phi$ by a covariant operation if and only if there exists a probability distribution on integers $\{w(k)\}_{k\in\mathbb{Z}}$ such that $p_{\psi}=\sum_{k\in\mathbb{Z}}w(k)\Upsilon_kp_\phi$.
\end{theorem}

\section{The properties of a-majorization $\asucc$}
We first show that the binary relation $\asucc$ is a preorder. 
For any probability distribution $p=\{p(n)\}_{n=0}^\infty$, it holds $p\asucc p$ since $p*\widetilde{p}=\delta_0\geq 0$. For probability distributions $p=\{p(n)\}_{n=0}^\infty$, $q=\{q(n)\}_{n=0}^\infty$ and $r=\{r(n)\}_{n=0}^\infty$ such that $p\asucc q$ and $q\asucc r$, we have
\begin{align}
    \sum_{l\in\mathbb{Z}}p(l)\widetilde{r}(k-l)&=\sum_{l,n\in\mathbb{Z}}p(l)\widetilde{r}(k-n)\delta_{n-l,0}\\
    &=\sum_{l,n\in\mathbb{Z}}p(l)\widetilde{r}(k-n)\sum_{m\in\mathbb{Z}}q(n-m)\widetilde{q}(m-l)\\
    &=\sum_{m\in\mathbb{Z}}\left(\sum_{l\in\mathbb{Z}}p(l)\widetilde{q}(m-l)\right)\left(\sum_{n\in\mathbb{Z}}q(n-m)\widetilde{r}(k-n)\right)\geq 0
\end{align}
for all $k\in\mathbb{Z}$, which implies that $p\asucc r$ holds. Therefore, the binary relation $\asucc$ is a preorder. It should be noted that the majorization relation $\succ$ is also a preorder.

We further show the following proposition:
\begin{proposition}\label{prop:a_majorization_preorder}
For probability distributions $p$ and $q$, $p\asucc q$ and $q\asucc p$ hold if and only if there exists a shift operator $\Upsilon_k$ with an integer $k\in\mathbb{Z}$ such that $p=\Upsilon_k q$. 
\end{proposition}
This corresponds to the fact that $\lambda\succ \lambda'$ and $\lambda'\succ\lambda$ hold if and only if there exists a permutation matrix $P$ such that $\lambda=P\lambda' $ in ordinary majorization theory. 

\begin{proof}[Proof of Proposition~\ref{prop:a_majorization_preorder}]
Let us first show the claim for probability distributions $p$ and $q$ such that $n_\star^{(p)}=n_\star^{(q)}=0$, where $n_\star^{(p)}\coloneqq \min\{n\mid p(n)\neq 0\}$ and $n_\star^{(q)}\coloneqq \min\{n\mid q(n)\neq 0\}$. 

Assume that $p*\widetilde{q}\geq 0$ and $q*\widetilde{p}\geq 0$ hold, that is,
\begin{align}
    \sum_{k=0}^np(k)\widetilde{q}(n-k)\geq 0,\quad \sum_{k=0}^nq(k)\widetilde{p}(n-k)\geq 0
\end{align}
for all $n\geq 0$.

For $n=0$, we have
\begin{align}
    p(0)\widetilde{q}(0)\geq 0,\quad q(0)\widetilde{p}(0)\geq 0.
\end{align}
Using $\widetilde{p}(0)=1/p(0)$ and $\widetilde{q}(0)=1/q(0)$, we get
\begin{align}
    p(0)\geq q(0),\quad q(0)\geq p(0).
\end{align}
Therefore, $p(0)=q(0)$. 

Suppose that $p(l)=q(l)$ holds for all $0\leq l\leq n$.
For $n+1$, from the assumption, we have
\begin{align}
    p(n+1)\widetilde{q}(0)+\sum_{k=0}^np(k)\widetilde{q}(n+1-k)\geq 0,\quad  q(n+1)\widetilde{p}(0)+\sum_{k=0}^nq(k)\widetilde{p}(n+1-k)\geq 0.
\end{align}
By using $p(l)=q(l)$ ($0\leq \forall l\leq n$), 
they imply
\begin{align}
    p(n+1)\widetilde{q}(0)+\sum_{k=0}^nq(k)\widetilde{q}(n+1-k)\geq 0,\quad  q(n+1)\widetilde{p}(0)+\sum_{k=0}^np(k)\widetilde{p}(n+1-k)\geq 0.
\end{align}
By construction, $\widetilde{p}$ and $\widetilde{q}$ satisfy
\begin{align}
    \sum_{k=0}^{n+1}p(k)\widetilde{p}(n+1-k)=\sum_{k=0}^{n+1}q(k)\widetilde{q}(n+1-k)=0
\end{align}
for $n\geq 0$. Therefore,
\begin{align}
    p(n+1)\widetilde{q}(0)-q(n+1)\widetilde{q}(0)\geq 0,\quad  q(n+1)\widetilde{p}(0)-p(n+1)\widetilde{p}(0)\geq 0,
\end{align}
which implies $p(n+1)=q(n+1)$, where we have used $\widetilde{p}(0)\neq 0$ and $\widetilde{q}(0)\neq 0$. Therefore, $p=q$. 

On the other hand, if $p=q$, then $p*\widetilde{q}=q*\widetilde{p}=\delta_0\geq 0$. Therefore, we have proved the claim under the assumption that $n_\star^{(p)}=n_\star^{(q)} =0$. 

To generalize this result for general probability distributions $p$ and $q$, consider shifted distributions
\begin{align}
    p'\coloneqq \Upsilon_{n_\star^{(p)}}p,\quad  q'\coloneqq \Upsilon_{n_\star^{(q)}}q.
\end{align}
They satisfy $n_{\star}^{(p')}=n_\star^{(q')}=0$. Note that Eq.\eqref{eq:delta_tildeq_q} implies 
\begin{align}
    \widetilde{p'}=\Upsilon_{-n_\star^{(p)}}\widetilde{p},\quad \widetilde{q'}=\Upsilon_{-n_\star^{(q)}}\widetilde{q}. 
\end{align}
Since conditions $p*\widetilde{q}\geq 0$ and $q*\widetilde{p}\geq 0$ are equivalent to $p'*\widetilde{q'}\geq 0$ and $q'*\widetilde{p'}\geq 0$, they hold if and only if
\begin{align}
    p'=q'
\end{align}
or equivalently, 
\begin{align}
    p=\Upsilon_{k}q,\quad k\coloneqq n_\star^{(q)}-n_\star^{(p)}.
\end{align}

\end{proof}

\section{Properties of $\mathcal{F}_{\mathrm{max}}$ and $\mathcal{F}_{\mathrm{min}}$ for pure states}
Let us prove the monotonicity of the variance under $\asucc$. Let $p$ and $q$ be probability distributions such that $p\asucc q$. Since $p=\sum_{k\in\mathbb{Z}}w(k)\Upsilon_k  q$ for a probability distribution $w=p*\widetilde{q}\geq 0$, we have
\begin{align}
    \mu_p&\coloneqq\sum_{n\in\mathbb{Z}} np(n)=\sum_{k\in\mathbb{Z}}n \sum_{k\in\mathbb{Z}}w(k)q(n-k)\\
    &=\sum_{k,n\in\mathbb{Z}}\left(k+(n-k)\right) w(k)q(n-k)\\
    &=\mu_w+\mu_q,
\end{align}
Similarly, it holds that
\begin{align}
    \mathrm{Var}(p)&=\sum_{n\in\mathbb{Z}}(n-\mu_p)^2 p(n)\\
    &=\sum_{k,n\in\mathbb{Z}}\left(k-\mu_w+(n-k)-\mu_q\right)^2w(k)q(n-k)\\
    &=\sum_{k\in\mathbb{Z}}(k-\mu_w)^2w(k)+\mathrm{Var}(q)\\
    &\geq \mathrm{Var}(q),
\end{align}
where in the last line, we have used $w(k)\geq0$ for all $k$. 
Of course, this monotonicity is expected from the fact that QFI monotonically decreases under a covariant operation and that QFI is four times the variance for pure states. 

Equation~\eqref{eq:Finf_F_F0} is a consequence of this monotonicity. 
To prove it, let us first show an easy but useful lemma:
\begin{lemma}\label{lem:Finf_F0_achievability}
Let $\psi$ be a pure state. 
If $\mathcal{F}_{\mathrm{max}}(\psi)<+\infty$, then for any $\lambda$ such that $4\lambda>\mathcal{F}_{\mathrm{max}}(\psi)$, it holds $\mathrm{P}_\lambda\asucc p_\psi$. Similarly, if $\mathcal{F}_{\mathrm{min}}(\psi)>0$, then for any $ \sigma$ such that $\mathcal{F}_{\mathrm{min}}(\psi)>4\sigma\geq 0$, it holds $p_\psi\asucc \mathrm{P}_{\sigma}$. 
\end{lemma}
\begin{proof}
Fix any $\lambda\in\mathbb{R}$ such that $4\lambda>\mathcal{F}_{\mathrm{max}}(\psi)$. From the definition of $\mathcal{F}_{\mathrm{max}}(\psi)$, there exists a real parameter $\lambda'$ such that $4\lambda\geq4\lambda'\geq\mathcal{F}_{\mathrm{max}}(\psi)$ and $\mathrm{P}_{\lambda'}\asucc p_\psi$. From Eq.~\eqref{eq:poisson_NScond}, $4\lambda\geq 4\lambda'$ implies $\mathrm{P}_{\lambda}\asucc \mathrm{P}_{\lambda'}$. Therefore, we get $\mathrm{P}_{\lambda}\asucc p_\psi$. 

Fix $\sigma$ such that $\mathcal{F}_{\mathrm{min}}(\psi)>4\sigma\geq 0$. Then there exists $\sigma'$ such that $\mathcal{F}_{\mathrm{min}}(\psi)\geq4 \sigma'\geq 4\sigma$ and $p_\psi\asucc \mathrm{P}_{\sigma'}$. Again, from Eq.~\eqref{eq:poisson_NScond}, we have $\mathrm{P}_{\sigma'}\asucc \mathrm{P}_{\sigma}$, which implies $p_\psi\asucc\mathrm{P}_{\sigma}$.
\end{proof}

Now note that $\mathcal{F}(\chi_\lambda)=4\lambda$ holds for a pure state $\chi_\lambda=\sum_{n,n'=0}^\infty \sqrt{\mathrm{P}_\lambda(n)\mathrm{P}_\lambda(n')}\ket{n}\bra{n'}$. Therefore, $\mathcal{F}_{\mathrm{max}}(\psi)$ and $\mathcal{F}_{\mathrm{min}}(\psi)$ are the amount of coherence in $\psi$ that can be converted from and to the pure state $\chi_\lambda$ whose energy distribution is given by a Poisson distribution $\mathrm{P}_\lambda$. 

\begin{proposition}[Equation~\eqref{eq:Finf_F_F0}]
For any pure state $\psi$, it holds
\begin{align}
    \mathcal{F}_{\mathrm{max}}(\psi)\geq \mathcal{F}(\psi)\geq \mathcal{F}_{\mathrm{min}}(\psi).\label{eq:f_0_f_f_inf} 
\end{align}
\end{proposition}
\begin{proof}
If $\mathcal{F}_{\mathrm{max}}(\psi)=+\infty$, then $\mathcal{F}_{\mathrm{max}}(\psi)\geq \mathcal{F}(\psi)$ holds. Suppose that $\mathcal{F}_{\mathrm{max}}(\psi)<+\infty$. From Lemma~\ref{lem:Finf_F0_achievability}, $\mathrm{P}_{\lambda}\asucc p_\psi$ holds for any $\lambda$ such that $4\lambda>\mathcal{F}_{\mathrm{max}}(\psi)$. From the monotonicity of the variance, we have $4\lambda\geq \mathcal{F}(\psi)$. Since $\lambda$ is arbitrary as long as $4\lambda\geq \mathcal{F}_{\mathrm{max}}(\psi)$ is satisfied, we have $\mathcal{F}_{\mathrm{max}}(\psi)\geq \mathcal{F}(\psi)$. 

If $\mathcal{F}_{\mathrm{min}}(\psi)=0$, then $\mathcal{F}(\psi)\geq\mathcal{F}_{\mathrm{min}}(\psi) $ trivially holds. Suppose that $\mathcal{F}_{\mathrm{min}}(\psi)\neq 0$. From Lemma~\ref{lem:Finf_F0_achievability}, $p_\psi\asucc\mathrm{P}_\sigma$ holds for any $\sigma$ such that $\mathcal{F}_{\mathrm{min}}(\psi)>4\sigma$. From the monotonicity of the variance, we have $\mathcal{F}(\psi)\geq 4\sigma$. Since $\sigma$ is arbitrary as long as $\mathcal{F}_{\mathrm{min}}(\psi)>4\sigma$ is satisfied, we have $\mathcal{F}(\psi)\geq \mathcal{F}_{\mathrm{min}}(\psi)$. 
\end{proof}
Furthermore, it should be noted that when the energy distribution of a pure state $\psi$ is given by a Poisson distribution, it holds
\begin{align}
    \mathcal{F}_{\mathrm{max}}(\psi)=\mathcal{F}(\psi)=\mathcal{F}_{\mathrm{min}}(\psi). \label{eq:f_0_f_f_inf_equality}
\end{align}

We also prove the following proposition, which corresponds to Eqs.~\eqref{eq:asymptotic_nec} and \eqref{eq:asymptotic_suf} in the one-shot regime:
\begin{proposition}\label{prop:conversion_oneshot_N_and_S}
For any pure states $\psi$ and $\phi$, we have
\begin{enumerate}[(1)]
    \item If $\psi\cov\phi$, then $\mathcal{F}_{\mathrm{max}} (\psi)\geq \mathcal{F}_{\mathrm{max}} (\phi)$ and $\mathcal{F}_{\mathrm{min}}(\psi)\geq \mathcal{F}_{\mathrm{min}}(\phi)$.
    \item If $\mathcal{F}_{\mathrm{min}} (\psi)> \mathcal{F}_{\mathrm{max}}(\phi)$, then $\psi\cov \phi$. 
\end{enumerate}
\end{proposition}
\begin{proof}
(1):
Assume that $\psi\cov\phi$ holds. If $\mathcal{F}_{\mathrm{max}}(\psi)=+\infty$, then $\mathcal{F}_{\mathrm{max}}(\psi)\geq \mathcal{F}_{\mathrm{max}}(\phi)$ holds.  Suppose that $\mathcal{F}_{\mathrm{max}}(\psi)<\infty$. 
Let $\lambda$ be any real number such that $4\lambda> \mathcal{F}_{\mathrm{max}}(\psi)$. From Lemma~\ref{lem:Finf_F0_achievability}, this means $\mathrm{P}_{\lambda}\asucc p_{\psi}$. Since $\asucc$ is a preorder, it implies that $\mathrm{P}_{\lambda}\asucc p_{\phi}$ and therefore $4\lambda \geq \mathcal{F}_{\mathrm{max}}(\phi)$. Since $\lambda$ is an arbitrary real number such that $4\lambda> \mathcal{F}_{\mathrm{max}}(\psi)$, we have $\mathcal{F}_{\mathrm{max}}(\psi)\geq \mathcal{F}_{\mathrm{max}}(\phi)$. 

If $\mathcal{F}_{\mathrm{min}}(\phi)=0$, $ \mathcal{F}_{\mathrm{min}}(\psi)\geq \mathcal{F}_{\mathrm{min}}(\phi)$ trivially holds since $\mathcal{F}_{\mathrm{min}}(\psi)$ is nonnegative for any pure state. Suppose that $\mathcal{F}_{\mathrm{min}}(\phi)>0$. Let $\sigma$ be an arbitrary real number such that $\mathcal{F}_{\mathrm{min}}(\phi)>4\sigma\geq 0$. From Lemma~\ref{lem:Finf_F0_achievability}, $p_\phi\asucc \mathrm{P}_{\sigma}$.
Since $\asucc$ is preorder, we have $p_\psi\asucc \mathrm{P}_{\sigma}$ and therefore $\mathcal{F}_{\mathrm{min}}(\psi)\geq 4\sigma$. Since $\sigma$ is an arbitrary number such that $\mathcal{F}_{\mathrm{min}}(\phi)>4\sigma>0$, we get $\mathcal{F}_{\mathrm{min}}(\psi)\geq \mathcal{F}_{\mathrm{min}}(\phi)$. 

(2): Assume that pure states $\psi$, $\phi$ satisfy $\mathcal{F}_{\mathrm{min}}(\psi)>\mathcal{F}_{\mathrm{max}}(\phi)$. Fix a real number $\lambda$ such that $\mathcal{F}_{\mathrm{min}}(\psi)>4\lambda>\mathcal{F}_{\mathrm{max}}(\phi)$. Lemma~\ref{lem:Finf_F0_achievability} shows that $p_\psi\asucc \mathrm{P}_{\lambda}$ and $\mathrm{P}_\lambda \asucc p_\phi$. Since $\asucc$ is a preorder, we have $p_\psi\asucc p_\phi$, or equivalently, $\psi\cov\phi$. 
\end{proof}

For comparison, here we summarize the properties of the max- and min-entropies without proof. For a review and details, see, e.g., \cite{sagawa_entropy_2022} and the references therein.  
Hereafter, the base of the logarithm is set to $2$. 
For a state $\rho$, the $\alpha$-R\'enyi entropy is defined by
\begin{align}
    S_\alpha(\rho)\coloneqq \frac{1}{1-\alpha}\log\mathrm{Tr}\left(\rho^\alpha\right).
\end{align}
The limits $\alpha\to 0$ and $\alpha\to\infty$ yield the max-entropy $S_{\mathrm{max}}(\rho)$ and the min-entropy $S_{\mathrm{min}}(\rho)$, which are given by
\begin{align}
    S_{\mathrm{max}}(\rho)\coloneqq \log \left(\mathrm{rank}\left(\rho\right)\right),\quad S_{\mathrm{min}}(\rho)\coloneqq -\log\left(\|\rho\|_\infty\right),
\end{align}
where $\mathrm{rank}\left(\rho\right)$ and $\|\rho\|_\infty$ denote the rank of $\rho$ and the maximum eigenvalue of $\rho$, respectively. 
These entropies satisfy
\begin{align}
    S_{\mathrm{max}}(\rho)\geq S(\rho)\geq S_{\mathrm{min}}(\rho), 
\end{align}
where $S(\rho)\coloneqq -\mathrm{Tr}(\rho\log\rho)$ is the von Neumann entropy.
This corresponds to Eq.~\eqref{eq:f_0_f_f_inf}. 

In entanglement theory, a maximally entangled state
\begin{align}
    \ket{\Phi_d}_{AB}=\frac{1}{\sqrt{d}}\sum_{i=1}^d\ket{i}_A\ket{i}_B,\label{eq:maximally_entangled_state}
\end{align}
is a reference state adopted in the literature, where $d$ denotes the dimension of the Hilbert spaces for each subsystem $A$ and $B$, while $\{\ket{i}_{Z}\}_{i=1}^d$ is an orthonormal basis for each subsystem $Z=A,B$. 
The reduced state is the maximally mixed state
\begin{align}
    \rho_d\coloneqq \mathrm{Tr}_B\left(\Phi_d\right)=\frac{1}{d}\mathbb{I}_d,\quad \Phi_d\coloneqq \ket{\Phi_d}\bra{\Psi_d}
\end{align}
where $\mathbb{I}_d$ denotes the identity operator. In this case, we have
\begin{align}
    S_{\mathrm{max}}(\rho_d)= S(\rho_d)= S_{\mathrm{min}}(\rho_d).
\end{align}
This corresponds to Eq.~\eqref{eq:f_0_f_f_inf_equality}. 

If a state $\rho$ is convertible to another state $\sigma$ by LOCC, we denote $\rho\locc \sigma$.
An important theorem on the one-shot convertibility between pure states by LOCC is the following:
\begin{theorem}
Let $\psi_{AB}$ and $\phi_{AB}$ be bipartite pure states. Define the reduced states $\rho_A\coloneqq \mathrm{Tr}_B(\psi_{AB})$ and $\sigma_A\coloneqq \mathrm{Tr}(\phi_{AB})$. The following two hold:
\begin{enumerate}[(1)]
    \item If $\psi_{AB}\locc \phi_{AB}$, then $S_{\mathrm{max}}(\rho_A)\geq S_{\mathrm{max}}(\sigma_A)$ and $S_{\mathrm{min}}(\rho_A)\geq S_{\mathrm{min}}(\sigma_A)$.
    \item If $S_{\mathrm{min}}(\rho_A)\geq S_{\mathrm{max}}(\sigma_A)$, then $\psi_{AB}\locc \phi_{AB}$. 
\end{enumerate}
\end{theorem}
This is the counterpart of Proposition~\ref{prop:conversion_oneshot_N_and_S}. 

\section{$\mathcal{F}_{\mathrm{max}}$ for general states}
In the main text, we have defined
\begin{align}
    \mathcal{F}_{\mathrm{max}}(\rho)\coloneqq \inf_{\Psi_\rho,H_A}\mathcal{F}_{\mathrm{max}}(\Psi_\rho),\label{eq:max_F_purification}
\end{align}
where the infimum is taken over the sets of all purifications $\Psi_\rho$ of $\rho$ and Hamiltonians $H_A$ with integer eigenvalues of the auxiliary system $A$ that is added to purify $\rho$. 

For any pure state $\psi$, its purification is given by $\psi\otimes\xi$ for some pure state $\xi$. Since the partial trace is a covariant operation, we have $p_{\psi\otimes\xi}\asucc p_{\psi}$. By using Proposition~\ref{prop:conversion_oneshot_N_and_S}, we get $\mathcal{F}_{\mathrm{max}}(\psi\otimes\xi)\geq \mathcal{F}_{\mathrm{max}}(\psi)$. Therefore, Eq.~\eqref{eq:max_F_purification} is consistent with Eq.~\eqref{eq:max_F_pure}.

\section{Proof of Eqs.~\eqref{eq:asymptotic_nec} and \eqref{eq:asymptotic_suf}}
Equations~\eqref{eq:asymptotic_nec} and \eqref{eq:asymptotic_suf} are obtained as a corollary of Theorem~\ref{thm:cost_dist_Finf_F0} and the following proposition:
\begin{proposition}
For any sequences of pure states $\widehat{\psi}=\{\psi_m\}_m$ and $\widehat{\phi}=\{\phi_m\}_m$, the followings hold:
\begin{enumerate}[(1)]
    \item $\widehat{\psi}\cov\widehat{\phi}\implies C_{\mathrm{cost}}(\widehat{\psi})\geq C_{\mathrm{cost}}(\widehat{\phi}),\quad C_{\mathrm{dist}}(\widehat{\psi})\geq C_{\mathrm{dist}}(\widehat{\phi})$.
    \item $C_{\mathrm{dist}}(\widehat{\psi})> C_{\mathrm{cost}}(\widehat{\phi})\implies \widehat{\psi}\cov\widehat{\phi}$.
\end{enumerate}
\end{proposition}
\begin{proof}
(1): For $R\coloneqq C_{\mathrm{cost}}(\widehat{\psi})$ and any positive number $\delta>0$, there exists $\delta'$ such that $\delta>\delta'\geq 0$ and
\begin{align}
    \widehat{\phi_{\mathrm{coh}}}(R+\delta')\cov\widehat{\psi}. 
\end{align}
Since $\widehat{\psi}\cov \widehat{\phi}$ holds from the assumption, we have $\widehat{\phi_{\mathrm{coh}}}(R+\delta')\cov\widehat{\phi}$. Therefore, 
\begin{align}
    C_{\mathrm{cost}}(\widehat{\psi})+\delta'\geq C_{\mathrm{cost}}(\widehat{\phi}).
\end{align}
Since $\delta>0$ is arbitrary and $\delta>\delta'$, we have $C_{\mathrm{cost}}(\widehat{\psi})\geq C_{\mathrm{cost}}(\widehat{\phi})$. In a similar way, $C_{\mathrm{dist}}(\widehat{\psi})\geq C_{\mathrm{dist}}(\widehat{\phi})$ is proven. 

(2): Fix $R$ such that $C_{\mathrm{dist}}(\widehat{\psi})>R> C_{\mathrm{cost}}(\widehat{\phi})$. There exists a real number $R'$ such that $C_{\mathrm{dist}}(\widehat{\psi})\geq R'>R$ and
\begin{align}
    \widehat{\psi}\cov \widehat{\phi_{\mathrm{coh}}}(R').
\end{align}
Similarly, there exists a real number $R''$ such that $R>R''\geq C_{\mathrm{cost}}(\widehat{\phi})$ and
\begin{align}
    \widehat{\phi_{\mathrm{coh}}}(R'')\cov\widehat{\phi}.
\end{align}
Since $R'>R''$ implies $\widehat{\phi_{\mathrm{coh}}}(R')\cov\widehat{\phi_{\mathrm{coh}}}(R'')$, we get $\widehat{\psi}\cov\widehat{\phi}$. 
\end{proof}

\section{Facts on entanglement theory and spectral entropy rates}
We here provide results in entanglement theory in the literature without proof. 

Let $\psi_{AB}$ and $\phi_{AB}$ be bipartite pure states. We define the density operators for the subsystem $A$ as $\rho_A\coloneqq \mathrm{Tr}\left(\psi_{AB}\right)$ and $\sigma_A\coloneqq\mathrm{Tr}\left(\phi_{AB}\right)$. The entanglement entropy $S_{\mathrm{EE}}$ of $\psi_{AB}$ is given by the von Neumann entropy $S(\rho_A)$ of the reduced state $\rho_A$, i.e.,
\begin{align}
    S_{\mathrm{EE}}(\psi_{AB})\coloneqq S\left(\rho_A\right)\coloneqq -\mathrm{Tr}_{A}\left(\rho_{A}\log \rho_A\right),
\end{align}
where the base of the logarithm is 2. 
Consider sequences of i.i.d. pure states $\widehat{\psi_{AB}}\coloneqq \{\psi_{AB}^{\otimes m}\}_m$ and $\widehat{\phi_{AB}}\coloneqq \{\phi_{AB}^{\otimes\ceil{R m}}\}_m$, where $R>0$. 
We say that $\widehat{\psi_{AB}}\coloneqq \{\psi_{AB}^{\otimes m}\}_m$ can be asymptotically converted to $\widehat{\phi_{AB}}(R)\coloneqq \{\phi_{AB}^{\otimes\ceil{R m}}\}_m$ if and only if there exists a sequence of local operations and classical communications (LOCC) $\widehat{\mathcal{E}}=\{\mathcal{E}_m\}_m$ such that
\begin{align}
    \lim_{m\to\infty}D\left(\mathcal{E}_m\left(\psi_{AB}^{\otimes m}\right),\phi_{AB}^{\otimes\ceil{R m}}\right)=0.
\end{align}
In this case, we denote $\widehat{\psi_{AB}}\locc\widehat{\phi_{AB}}(R)$. 
It is known \cite{bennett_concentrating_1996} that the conversion from $\widehat{\psi_{AB}}$ to $\widehat{\phi_{AB}}(R)$ by LOCC is possible if $R\leq S_{\mathrm{EE}}(\psi_{AB})/S_{\mathrm{EE}}(\phi_{AB})$ and impossible if $R> S_{\mathrm{EE}}(\psi_{AB})/S_{\mathrm{EE}}(\phi_{AB})$. In other words, pure states are interconvertible and the optimal rate is given by the ratio of the entanglement entropies in the i.i.d. regime. 

To analyze the asymptotic convertibility in a more general setup, it is common to adopt a maximally entangled state as a reference. 
Let us first define the entanglement cost. 
For a given sequence of pure states $\widehat{\psi}=\{\psi_m\}_m$, we say that a rate $R$ is achievable in a dilution process if and only if there exists a sequence of nonnegative numbers $\{N_m\}_m$ such that $\widehat{\Phi}(\{N_m\})\locc \widehat{\psi}$ and
\begin{align}
    \limsup_{m\to\infty}\frac{1}{m}\log N_m\leq R,
\end{align}
where we have defined $\widehat{\Phi}(\{N_m\})\coloneqq \{\Phi_{N_m}\}_m$ for the maximally entangled state defined in Eq.~\eqref{eq:maximally_entangled_state}. 
The entanglement cost of $\widehat{\psi}$ is defined by
\begin{align}
    E_{\mathrm{cost}}(\widehat{\psi})\coloneqq \inf\left\{R\mid R\text{ is achievable in a  dilution process}\right\}.
\end{align}
In a similar way, we can define the distillable entanglement. 
We say that a rate $R$ is achievable in a distillation process if and only if there exists a sequence of nonnegative numbers $\{N_m\}_m$ such that $\widehat{\psi}\locc\widehat{\Phi}(\{N_m\})$ and
\begin{align}
    \liminf_{m\to\infty}\frac{1}{m}\log N_m\geq R.
\end{align}
The distillable entanglement is defined as
\begin{align}
    E_{\mathrm{dist}}(\widehat{\psi})\coloneqq \sup\left\{R\mid R\text{ is achievable in a distillation process}\right\}.
\end{align}

The spectral sup- and inf-entropy rates in the quantum case have been developed in different contexts, e.g., in \cite{nagaoka_information-spectrum_2007,ogawa_strong_2000,hayashi_general_2003,hayashi_general_2006,bowen_asymptotic_2008,datta_smooth_2009}. 
Here we provide one of the alternative but equivalent definitions, which is based on the smoothing technique \cite{datta_smooth_2009,renner_security_2005}. 
For a given sequence of states $\widehat{\rho}=\{\rho_m\}_m$, its spectral sup- and inf-entropy rates are defined by
\begin{align}
    \overline{S}(\widehat{\rho})\coloneqq\lim_{\epsilon\to 0} \limsup_{m\to\infty}\frac{1}{m}S_{\mathrm{max}}^\epsilon(\rho_m),\quad \underline{S}(\widehat{\rho})\coloneqq\lim_{\epsilon\to 0} \liminf_{m\to\infty}\frac{1}{m}S_{\mathrm{min}}^\epsilon(\rho_m),
\end{align}
where the smooth max- and min-entroies are defined by
\begin{align}
    S_{\mathrm{max}}^\epsilon(\rho)\coloneqq \inf_{\sigma\in B^\epsilon(\rho)}S_{\mathrm{max}}(\sigma),\quad S_{\mathrm{min}}^\epsilon(\rho)\coloneqq  \sup_{\sigma\in B^\epsilon(\rho)}S_{\mathrm{min}}(\sigma)
\end{align}
for $B^\epsilon(\rho)\coloneqq \left\{\sigma\text{: quantum states} \middle| D(\sigma,\rho)\leq \epsilon\right\}$.

For a sequence of general pure states $\widehat{\psi_{AB}}=\{\psi_{AB,m}\}_m$, let us define a sequence of reduced states by $\widehat{\rho_A}=\{\rho_{A,m}\}_m$, where $\rho_{A,m}\coloneqq \mathrm{Tr}_B(\psi_{AB,m})$. 
It is shown \cite{hayashi_general_2006,bowen_asymptotic_2008} that
\begin{align}
    E_{\mathrm{cost}}(\widehat{\psi_{AB}})=\overline{S}(\widehat{\rho_A}),\quad E_{\mathrm{dist}}(\widehat{\psi_{AB}})=\underline{S}(\widehat{\rho_A}).
\end{align}
For the convertibility between sequences of pure states $\widehat{\psi_{AB}}=\{\psi_{AB,m}\}_m$ and $\widehat{\phi_{AB}}=\{\phi_{AB,m}\}_m$, the following two hold \cite{jiao_asymptotic_2018}:
\begin{align}
    \widehat{\psi_{AB}}\locc \widehat{\phi_{AB}}&\implies \overline{S}(\widehat{\rho_A})\geq \overline{S}(\widehat{\sigma_A}),\quad \underline{S}(\widehat{\rho_A})\geq \underline{S}(\widehat{\sigma_A})\\
   \underline{S}(\widehat{\rho_A})> \overline{S}(\widehat{\sigma_A}) &\implies \widehat{\psi_{AB}}\locc \widehat{\phi_{AB}},
\end{align}
where we have defined $\widehat{\sigma_A}\coloneqq \{\sigma_{A,m}\}_m$ and $\sigma_{A,m}\coloneqq \mathrm{Tr}_B(\phi_{AB,m})$. 
They are the counterparts of Eqs.~\eqref{eq:asymptotic_nec} and \eqref{eq:asymptotic_suf} in entanglement theory. 

In particular, for an i.i.d. sequence of pure states $\widehat{\psi_{AB}}=\{\psi_{AB}^{\otimes m}\}_m$, the spectral entropy rates are equal to the entanglement entropy:
\begin{align}
    \overline{S}(\widehat{\rho_A})=S_{\mathrm{EE}}(\psi_{AB})=\underline{S}(\widehat{\rho_A}).
\end{align}

\section{Interconversion between $\widehat{\phi}_{\mathrm{coh}}(R)$ and $\widehat{\chi}_\lambda$}
We here show the following:
\begin{lemma}\label{lem:cbit_poisson}
A sequence $\widehat{\phi}_{\mathrm{coh}}(R)$ is interconvertible to $\widehat{\chi}_{R/4}$. That is, for any $\epsilon>0$, $\widehat{\phi}_{\mathrm{coh}}(R)\cov_\epsilon\widehat{\chi}_{R/4}$ and $\widehat{\chi}_{R/4}\cov_\epsilon\widehat{\phi}_{\mathrm{coh}}(R)$. 
\end{lemma}

Lemma~\ref{lem:cbit_poisson} is proved based on the arguments in \cite{marvian_operational_2022}. The following lemma connects the closeness of energy distributions and the convertibility of pure states. 
\begin{lemma}[\cite{marvian_operational_2022}]\label{lem:covariant_dTV}
Let $\psi$ and $\phi$ be pure states. There exists a covariant unitary operator $U$ such that
\begin{align}
    D\left(U\psi U^\dag,\phi\right)\leq \sqrt{2d_{\mathrm{TV}}\left(p_\psi,p_\phi\right)}.
\end{align}
\end{lemma}
Here, the total variation distance between two probability distributions is defined by
\begin{align}
    d_{\mathrm{TV}}\left(p,q\right)\coloneqq \frac{1}{2}\sum_{n\in\mathbb{Z}}\left|p(n)-q(n)\right|.
\end{align}

In the i.i.d. regime, the translated Poisson distribution plays an important role, which is defined as follows:
\begin{definition}[The translated Poisson distribution]
For $\mu\in\mathbb{R}$ and $\sigma^2\geq0$, the translated Poisson distribution is defined by
\begin{align}
     \mathrm{TP}_{\mu,\sigma^2}(n)\coloneqq\mathrm{P}_{\sigma^2+\gamma}(n-s)=
    \begin{cases}
    e^{-(\sigma^2+\gamma)}\frac{(\sigma^2+\gamma)^{n-s}}{(n-s)!}\quad &(\mathrm{For }\,n=s,s+1,s+2,s+3,\cdots)\\
    0 &(\mathrm{Otherwise}),
    \end{cases}
\end{align}
where $s\coloneqq \floor{\mu-\sigma^2}$ is an integer and $\gamma\coloneqq \mu-\sigma^2- \floor{\mu-\sigma^2}$ is a parameter satisfying $0\leq \gamma<1$. 
The mean and variance are given by $\mu$ and $\sigma^2+\gamma$, respectively. Alternatively, the translated Poisson distribution is written as $\mathrm{TP}_{\mu,\sigma^2}=\Upsilon_{\floor{\mu-\sigma^2}}\mathrm{P}_{\sigma^2+\gamma}$. 
\end{definition}

It is known that the sum of integer-valued random variables converges to the translated Poisson distribution \cite{barbour_total_2002,marvian_operational_2022}. 
Let $\{Z_i\}_{i=1}^m$ be a set of independent random variables with mean $\mu_i\coloneqq \mathbb{E}Z_i$ and variance $\sigma^2_i\coloneqq \mathbb{E}((Z_i-\mu_i)^2)$. 
Assume that its absolute third moment $\mathbb{E}|Z_i^3|$ is finite. 
Let $\mathcal{L}(Z)$ denote the probability distribution of a random variable $Z$. 
We define
\begin{align}
    v_i&\coloneqq\min\left\{\frac{1}{2},1- d_{\mathrm{TV}}\left(\mathcal{L}(Z_i),\mathcal{L}(Z_i+1)\right)\right\},\\
    \psi_i&\coloneqq \sigma_i^2 \mathbb{E}\left(Z_i(Z_i-1)\right)+|\mu_i-\sigma^2_i|\mathbb{E}\left((Z_i-1)(Z_i-2)\right)+\mathbb{E}\left(|Z_i(Z_i-1)(Z_i-2)|\right).
\end{align}
For the sum of the random variables $W\coloneqq \sum_{i=1}^mX_i$ with mean $\tilde{\mu}\coloneqq \mathbb{E}(W)=\sum_{i=1}^m\mu_i$ and variance $\tilde{\sigma}^2\coloneqq \mathbb{E}((W-\tilde{\mu})^2)=\sum_{i=1}^m\sigma^2_i$, 
the following theorem holds:
\begin{theorem}[Corollary~3.2 in \cite{barbour_total_2002}, Theorem~7 in \cite{marvian_operational_2022}]\label{TP_distr_convergence}
Suppose that the random variable $Z_i$ satisfies $\sigma_i^2\geq a>0$, $v_i\geq b >0$ and $\sigma_i^{-2}\psi_i\leq c<\infty$ for any $1\leq i\leq m$ with some parameters $a,b$ and $c$. Then 
\begin{align}
    \dtv\left(\mathcal{L}(W),\mathrm{TP}_{\tilde{\mu},\tilde{\sigma}^2}\right)\leq \frac{c}{\sqrt{mb-\frac{1}{2}}}+\frac{2}{m a}.
\end{align}
\end{theorem}
Applying this theorem to the energy distribution of $\phi^{\otimes \ceil{Rm}}_{\mathrm{coh}}$, we get
\begin{align}
    \lim_{m\to\infty}d_{\mathrm{TV}}\left(p_{\phi^{\otimes \ceil{Rm}}_{\mathrm{coh}}},\mathrm{TP}_{\frac{1}{2}\ceil{Rm},\frac{1}{4}\ceil{Rm}}\right)=0. \label{eq:coh_TP}
\end{align}

For Poisson distributions with different variance, the upper bound on the total variation distance is provided in the following theorem:
\begin{theorem}[Equation~(5) in Ref.~\cite{roos_improvements_2003}, Equation~(2.2) in Ref.~\cite{adell_exact_2006}, Lemma~8 in Ref.~\cite{marvian_operational_2022}]~\label{thm:comparison_Poisson}\\
\begin{align}
    \dtv\left(\mathrm{P}_{\sigma^2},\mathrm{P}_{\sigma^{\prime2}}\right)\leq \min\left\{x,\sqrt{\frac{2}{e}}\left(\sqrt{\sigma^2+x}-\sigma\right)\right\},
\end{align}
where $x\coloneqq |\sigma^2-\sigma^{\prime 2}|$.
\end{theorem}

Let us define $s_m\coloneqq \floor{\frac{1}{2}\ceil{Rm}-\frac{1}{4}\ceil{Rm}}$ and $\gamma_m\coloneqq\frac{1}{2}\ceil{Rm}-\frac{1}{4}\ceil{Rm}-\floor{\frac{1}{2}\ceil{Rm}-\frac{1}{4}\ceil{Rm}}$.
From Eq.~\eqref{eq:coh_TP} and Theorem~\ref{thm:comparison_Poisson}, we have
\begin{align}
    &d_{\mathrm{TV}}\left(p_{\phi^{\otimes \ceil{Rm}}_{\mathrm{coh}}},\Upsilon_{s_m}\mathrm{P}_{\frac{1}{4}Rm}\right)\\
    &\leq d_{\mathrm{TV}}\left(p_{\phi^{\otimes \ceil{Rm}}_{\mathrm{coh}}},\mathrm{TP}_{\frac{1}{2}\ceil{Rm},\frac{1}{4}\ceil{Rm}}\right)+d_{\mathrm{TV}}\left(\mathrm{TP}_{\frac{1}{2}\ceil{Rm},\frac{1}{4}\ceil{Rm}},\Upsilon_{s_m}\mathrm{P}_{\frac{1}{4}Rm}\right)\\
    &=d_{\mathrm{TV}}\left(p_{\phi^{\otimes \ceil{Rm}}_{\mathrm{coh}}},\mathrm{TP}_{\frac{1}{2}\ceil{Rm},\frac{1}{4}\ceil{Rm}}\right)+d_{\mathrm{TV}}\left(\mathrm{P}_{\frac{1}{4}\ceil{Rm}+\gamma_m},\mathrm{P}_{\frac{1}{4}Rm}\right)\\
    &\to 0\quad \label{eq:convergence_iid_energy_distr}
\end{align}
as $m\to\infty$. 

Let us define $q_m\coloneqq \Upsilon_{s_m}\mathrm{P}_{\frac{1}{4}Rm}$ and pure states $\kappa_m\coloneqq \sum_{n,n'\in\mathbb{Z}}q_m(n)q_m(n')\ket{n}\bra{n'}$. Since $\mathrm{P}_{\frac{1}{4}Rm}\asucc q_m$ and $q_m\asucc \mathrm{P}_{\frac{1}{4}Rm}$ hold, there exist covariant channels $\Lambda_m$ and $\Lambda_m'$ such that $\chi_{\frac{1}{4}Rm}=\Lambda_m(\kappa_m)$ and $\kappa_m=\Lambda_m'(\chi_{\frac{1}{4}Rm})$. 
On the other hand, according to Lemma~\ref{lem:covariant_dTV}, there exist covariant channels $\Theta_m$ and $\Theta_m'$ such that 
\begin{align}
    D\left(\Theta_m\left(\phi^{\otimes \ceil{Rm}}_{\mathrm{coh}}\right),\kappa_m\right)&\leq \sqrt{2d_{\mathrm{TV}}\left(p_{\phi^{\otimes \ceil{Rm}}_{\mathrm{coh}}},\Upsilon_{\frac{1}{2}\ceil{Rm}-\frac{1}{4}\ceil{Rm}}\mathrm{P}_{\frac{1}{4}Rm}\right)}\\
    D\left(\Theta_m'\left(\kappa_m\right),\phi^{\otimes \ceil{Rm}}_{\mathrm{coh}}\right)&\leq \sqrt{2d_{\mathrm{TV}}\left(p_{\phi^{\otimes \ceil{Rm}}_{\mathrm{coh}}},\Upsilon_{\frac{1}{2}\ceil{Rm}-\frac{1}{4}\ceil{Rm}}\mathrm{P}_{\frac{1}{4}Rm}\right)}.
\end{align}
Defining $\mathcal{E}_m\coloneqq \Lambda_m\circ \Theta_m$ and $\mathcal{E}_m'\coloneqq \Theta_m'\circ\Lambda_m'$, we have
\begin{align}
    \lim_{m\to\infty}D\left(\mathcal{E}_m\left(\phi^{\otimes \ceil{Rm}}\right),\chi_{\frac{1}{4}Rm}\right)=\lim_{m\to\infty}D\left(\phi^{\otimes \ceil{Rm}},\mathcal{E}_m'\left(\chi_{\frac{1}{4}Rm}\right)\right)=0,
\end{align}
which concludes the proof of Lemma~\ref{lem:cbit_poisson}. 

\section{Proof of Lemma~\ref{lem:purification_majorization}} 
\begin{proof}
Let $\mathcal{E}$ be a covariant channel with respect to the Hamiltonian given by Eq.~\eqref{eq:Hamiltonian}. From the covariant Stinespring dilation theorem \cite{keyl_optimal_1999,marvian_mashhad_symmetry_2012}, there exists an ancillary system $A$, a symmetric pure state $\eta_A$ of $A$, the Hamiltonian $H_A$ and a covariant unitary operator $U$ such that
\begin{align}
    \mathcal{E}\left(\cdots\right)=\mathrm{Tr}_A\left(U\left(\cdots\otimes \eta_A\right)U^\dag\right).
\end{align}
In our setup, the Hamiltonian of the original system is assumed to be given by Eq.~\eqref{eq:Hamiltonian}. In this case, the Hamiltonian $H_A$ of the auxiliary system has integer eigenvalues. Furthermore, without loss of generality, it is possible to assume that $\eta_A$ is an energy eigenstate with a vanishing eigenvalue of $H_A$, i.e., $H_A\ket{\eta_A}=0$. 

For $\chi_{\lambda}=\ket{\chi_\lambda}\bra{\chi_\lambda}$ and $\ket{\chi_\lambda}\coloneqq\sum_{n=0}^\infty\mathrm{P}_\lambda(n)\ket{n}$, we define $\Phi\coloneqq U\left(\chi_\lambda\otimes \eta_A\right)U^\dag$. Since $U$ is covariant and $H_A\ket{\eta_A}=0$ holds, we have $p_\Phi=p_{\chi_\lambda}=\mathrm{P}_\lambda$. In addition, $\Phi$ is a purification of $\mathcal{E}(\chi_\lambda)$, which concludes the proof of Lemma~\ref{lem:purification_majorization}. 
\end{proof}

Lemma~\ref{lem:purification_majorization} shows that $\mathcal{F}_\infty(\rho)$ for a general state $\rho$ quantifies the minimum amount of coherence in $\chi_\lambda$ required to generate $\rho$.

\section{Poof of Lemma~\ref{lem:purification_smooth}}

Instead of Lemma~\ref{lem:purification_smooth}, we here prove a slightly improved lemma: 
\begin{lemma}\label{lem:purifcation_smooth_tighter}
Let $\psi$ and $\phi$ be pure states. 
If there exists a covariant channel $\mathcal{E}$ such that $\mathcal{E}(\psi)\approx_\epsilon\phi$ for $0\leq \epsilon\leq 1$, then there exists a pure state $\psi'$ such that $\psi'\in B^{f(\epsilon)}_{\mathrm{pure}}(\psi)$ and $\quad p_{\psi'}\asucc p_\phi$ for $f(\epsilon)\coloneqq  \sqrt{2\sqrt{1-(1-\epsilon)^2}}$. 
\end{lemma}

In the proof Lemma~\ref{lem:purifcation_smooth_tighter}, we use the following:
\begin{lemma}\label{lem:tr_to_dTV}
Let $\psi$ and $\phi$ be arbitrary states. For a covariant channel $\mathcal{E}$, if
\begin{align}
    \epsilon\geq  D\left(\mathcal{E}\left(\psi\right),\phi\right)
\end{align}
for some $1\geq \epsilon\geq 0$, 
then there exists a probability distribution $\{w(k)\}_{k\in\mathbb{Z}}$ such that
\begin{align}
    d_{\mathrm{TV}}\left(p_\psi,\sum_{k\in\mathbb{Z}}w(k)\Upsilon_k p_{\phi}\right)\leq \sqrt{1-(1-\epsilon)^2}.
\end{align}
\end{lemma}

\begin{proof}
It is known \cite{gour_resource_2008} that any covariant quantum channel has a Kraus representation with Kraus operators $\{K_{k,u}\}$ such that
\begin{align}
    K_{k,u}&\coloneqq \sum_{n=\max\{0,-k\}}^\infty c_n^{(k,u)}\ket{n-k}\bra{n}\\
    &=\sum_{n\in\mathbb{Z}}c_n^{(k,u)}\ket{n-k}\bra{n},
\end{align}
where in the second line, for notational simplicity, we defined $c_n^{(k,u)}=0$ if $n<\max\{0,-k\}$.
For the channel to be trace-preserving, the coefficients must satisfy
\begin{align}
    \sum_{k,u} |c_n^{(k,u)}|^2=1 
\end{align}
for all $n\geq 0$. 
Note that 
\begin{align}
    \sum_{k,u}K_{k,u}^\dag\ket{n-k}\bra{n-k}K_{k,u}&= \sum_{k,u} |c_n^{(k,u)}|^2\ket{n}\bra{n}=\ket{n}\bra{n}.
\end{align}

For a given pure state $\psi$, we define
\begin{align}
    q^{(k,u)}\coloneqq \|K_{k,u}\ket{\psi}\|^2
\end{align}
For each $(k,u)$ such that $q^{(k,u)}\neq 0$, we define a normalized pure state
\begin{align}
    \ket{\phi}^{(k,u)}\coloneqq \frac{1}{ \sqrt{q^{(k,u)}}}K_{k,u}\ket{\psi}. 
\end{align}
The energy distributions for pure states are related as
\begin{align}
    p_\psi(n)&\coloneqq \braket{\psi|\ket{n}\bra{n}|\psi}\\
    &=\sum_{k,u}\braket{\psi|K_{k,u}^\dag\ket{n-k}\bra{n-k}K_{k,u}|\psi}\\
    &=\sum_{k,u}q^{
    (k,u)}\braket{\phi^{(k,u)}|\ket{n-k}\bra{n-k}|\phi^{(k,u)}}\\
    &=\sum_{k,u}q^{
    (k,u)}p_{\phi^{(k,u)}}(n-k),
\end{align}
or equivalently, 
\begin{align}
    p_\psi=\sum_{k,u}q^{
    (k,u)}\Upsilon_kp_{\phi^{(k,u)}}=\sum_kw_k\Upsilon_kp_{\phi^{(k,u)}},
\end{align}
where we have defined 
\begin{align}
    w(k)\coloneqq \sum_uq^{
    (k,u)}.
\end{align}
For this probability distribution $\{w(k)\}_{k\in\mathbb{Z}}$, we have
\begin{align}
    \dtv\left(p_\psi,\sum_kw(k)\Upsilon_kp_\phi\right)
    &=\dtv\left(\sum_{k,u}q^{(k,u)}\Upsilon_kp_{\phi^{(k,u)}},\sum_{k,u}q^{(k,u)}\Upsilon_kp_\phi\right)\\
    &=\sum_{k,u}q^{(k,u)} \dtv\left(\Upsilon_kp_{\phi^{(k,u)}},\Upsilon_kp_\phi\right)\\
    &=\sum_{k,u}q^{(k,u)} \dtv\left(p_{\phi^{(k,u)}},p_\phi\right)\\
    &\leq \sqrt{\sum_{k,u}q^{(k,u)} \dtv^2\left(p_{\phi^{(k,u)}},p_\phi\right)}.
\end{align}

By using the the Fuchs–van de Graaf inequalities, we have
\begin{align}
    \epsilon&
    \geq D(\mathcal{E}(\psi),\phi)\\
    &\geq 1-\sqrt{F(\mathcal{E}(\psi),\phi)}\\
    &=1-\sqrt{\sum_{k,u}q^{(k,u)}F(\phi^{(k,u)},\phi)}\\
    &\geq 1-\sqrt{1-\sum_{k,u}q^{(k,u)}D^2(\phi^{(k,u)},\phi)}\\
    &\geq 1-\sqrt{1-\sum_{k,u}q^{(k,u)}\dtv^2(p_{\phi^{(k,u)}},p_{\phi})},
\end{align}
which implies that 
\begin{align}
    \sum_{k,u}q^{(k,u)}\dtv^2(p_{\phi^{(k,u)}},p_{\phi})\leq 1-(1-\epsilon)^2.
\end{align}

Therefore, we finally get
\begin{align}
    \dtv\left(p_\psi,\sum_kw(k)\Upsilon_kp_\phi\right)\leq \sqrt{1-(1-\epsilon)^2}.
\end{align}
\end{proof}

\begin{proof}[Proof of Lemma~\ref{lem:purifcation_smooth_tighter}]
From Lemma~\ref{lem:tr_to_dTV}, there exists a probability distirubiton $\{w(k)\}_{k\in\mathbb{Z}}$ such that 
\begin{align}
    d_{\mathrm{TV}}\left(p_{\psi},\sum_{k\in\mathbb{Z}}w(k)\Upsilon_k p_\phi\right)\leq \sqrt{1-(1-\epsilon)^2}. 
\end{align}
Defining $q\coloneqq w* p_{\phi}$ and $\ket{\psi''}\coloneqq \sum_{n}q(n)\ket{n}$, from Lemma~\ref{lem:covariant_dTV}, 
there exists a covariant unitary operator $U$ such that 
\begin{align}
    D\left(\psi,\psi'\right)\leq \sqrt{2d_{\mathrm{TV}}\left(p_{\psi},q\right)}\leq \sqrt{2\sqrt{1-(1-\epsilon)^2}},\quad \psi'\coloneqq U\psi''U^\dag.
\end{align}
Since $q\asucc p_\phi$ and $p_{\psi'}\asucc q$, we have $p_{\psi'}\asucc p_{\phi}$. 
\end{proof}
In the main text, we used a looser bound $2\epsilon^{1/4}$ instead of $\sqrt{2\sqrt{1-(1-\epsilon)^2}}$ to simplify the notation.

\section{Proof of Theorem~\ref{thm:cost_general}}
Theorem~\ref{thm:cost_general} can be proven in the same way as the first part of Theorem~\ref{thm:cost_dist_Finf_F0}. For completeness, we here repeat the proof.

\begin{proof}[Proof of Theorem~\ref{thm:cost_general}]

To show $C_{\mathrm{cost}}(\widehat{\rho})=\overline{\mathcal{F}}(\widehat{\rho})$, let us define $C_{\mathrm{cost}}^\epsilon\left(\widehat{\rho}\right)\coloneqq \inf\left\{R\middle| \widehat{\phi_{\mathrm{coh}}}(R)\cov_{\epsilon}\widehat{\rho}\right\}$. 
Defining $4\lambda^{\epsilon/2}\coloneqq C_{\mathrm{cost}}^{\epsilon/2}(\widehat{\rho})$, for any parameter $\delta>0$, there exists a real number $\delta'$ such that $\delta> \delta'\geq 0$ and $\widehat{\phi_{\mathrm{coh}}}(4(\lambda^{\epsilon/2}+\delta'))\cov_{\epsilon/2}\widehat{\rho}$. Since $\widehat{\chi}_{(\lambda^{\epsilon/2}+\delta')}\cov_{\epsilon/2}\widehat{\phi_{\mathrm{coh}}}(4(\lambda^{\epsilon/2}+\delta'))$, we have $\widehat{\chi}_{(\lambda^{\epsilon/2}+\delta)}\cov_{\epsilon}\widehat{\rho}$, where we have used Eq.~\eqref{eq:poisson_NScond}. 
From Lemma~\ref{lem:purification_majorization}, for all sufficiently large $m$, there exists a state $\sigma_m\in B^\epsilon(\rho_m)$ whose purification $\Phi_m$ satisfies $\mathrm{P}_{(\lambda^{\epsilon/2}+\delta)m}= p_{\Phi_m}$. Therefore, $4(\lambda^{\epsilon/2}+\delta)m\geq \mathcal{F}_{\mathrm{max}}^\epsilon\left(\rho_m\right)$ for sufficiently large $m$, which implies
\begin{align}
    C_{\mathrm{cost}}^{\epsilon/2}(\widehat{\psi})+4\delta\geq \limsup_{m\to\infty}\frac{1}{m}\mathcal{F}_{\mathrm{max}}^\epsilon(\rho_m).
\end{align}
Since $\delta>0$ is arbitrary, we have $C_{\mathrm{cost}}(\widehat{\rho})\geq \overline{\mathcal{F}}(\widehat{\rho})$ in the limit of $\epsilon\to +0$. 

To show the opposite inequality, 
let us define $4\lambda_m^{\epsilon/2}\coloneqq \mathcal{F}^{\epsilon/2}_{\mathrm{max}}(\rho_m)$. For any $\delta>0$, there exist $\delta_m'$, satisfying $\delta>\delta_m'\geq0$, such that there exist a state $\sigma_m\in B^{\epsilon/2}(\rho_m)$ and its purification $\Phi_m$ satisfying $\mathrm{P}_{(\lambda^{\epsilon/2}_m+\delta_m')}\asucc p_{\Phi_m}$. 
Note that for all sufficiently large $m$, $m(\lambda^{\epsilon/2}+\delta)\geq \lambda_m^{\epsilon/2}$ holds for $\lambda^{\epsilon/2}\coloneqq \limsup_{m\to\infty}\frac{1}{m}\lambda_m^{\epsilon/2}$. Therefore, we get $\mathrm{P}_{(\lambda^{\epsilon/2}+2\delta)m}\asucc p_{\Phi_m}$, where we have used $m\delta > \delta_m'$.
Since $\widehat{\phi_{\mathrm{coh}}}(4(\lambda^{\epsilon/2}+2\delta))\cov_{\epsilon/2}\widehat{\chi}_{(\lambda^{\epsilon/2}+2\delta)}$, we have $\widehat{\phi_{\mathrm{coh}}}(4(\lambda^{\epsilon/2}+2\delta))\cov_{\epsilon/2}\{\Phi_m\}_m$. 
Since the partial trace is a covariant operation, we have $\widehat{\phi_{\mathrm{coh}}}(4(\lambda^{\epsilon/2}+2\delta))\cov_{\epsilon}\widehat{\rho}$. Therefore, 
\begin{align}
    C_{\mathrm{cost}}^\epsilon (\widehat{\rho})\leq \limsup_{m\to\infty}\frac{1}{m}\mathcal{F}_{\mathrm{max}}^{\epsilon/2}(\rho_m)+8\delta.
\end{align}
Since $\delta>0$ is arbitrary, as $\epsilon\to +0$, we have $ C_{\mathrm{cost}} (\widehat{\rho})\leq \overline{\mathcal{F}}(\widehat{\rho})$, and therefore $C_{\mathrm{cost}} (\widehat{\rho})=\overline{\mathcal{F}}(\widehat{\rho})$.
\end{proof}

\section{Convertibility between pure states with finite periods}
So far, we have constructed asymptotic conversion theory in RTA in the non-i.i.d. regime under the assumption that the Hamiltonians are given by the one for a harmonic oscillator system. Here, improving the arguments in \cite{gour_resource_2008,marvian_operational_2022}, we show that the convertibility between pure states with finite periods can be analyzed with a harmonic oscillator system with the Hamiltonian in Eq.~\eqref{eq:Hamiltonian}. 

To begin with, we analyze a one-shot setting. Let $A$ and $B$ be quantum systems associated with Hilbert spaces $\mathcal{H}_A$ and $\mathcal{H}_B$. We assume that their Hamiltonians $H_A$ and $H_B$ are bounded from below. Let $\mathcal{L}(\mathcal{H})$ denote the set of all linear operators on a Hilbert space $\mathcal{H}$. A quantum channel $\mathcal{E}_{A\to B}:\mathcal{L}(\mathcal{H}_A)\to \mathcal{L}(\mathcal{H}_B)$ is called covariant if and only if 
\begin{align}
    e^{-\ii H_Bt}\left(\mathcal{E}_{A\to B}\left(\rho_A\right)\right)e^{\ii H_Bt}=\mathcal{E}_{A\to B}\left(e^{-\ii H_A t}\rho_A e^{\ii H_A t}\right)
\end{align}
holds for all states $\rho_A$ of the system $A$ and for all $t\in\mathbb{R}$. Note that the covariance is defined with respect to the Hamiltonians of the input and output systems of the channel. Clarifying this point, we denote $(\rho_A,H_A)\cov (\sigma_B,H_B)$ if and only if there exists a covariant channel $\mathcal{E}_{A\to B}$ such that $\mathcal{E}_{A\to B}(\rho_A)=\sigma_B$. Similarly, we denote $(\rho_A,H_A)\cov_\epsilon (\sigma_B,H_B)$ for $\epsilon\in(0,1]$ if and only if there exists a covariant channel $\mathcal{E}_{A\to B}$ such that $D(\mathcal{E}_{A\to B}(\rho_A),\sigma_B)\leq \epsilon$. When no confusion arises, we simply write $\rho\cov \sigma$ and $\rho\cov_\epsilon\sigma$ by omitting the Hamiltonian as in the main text. 

In the following, we analyze the conversion from a pure state in $A$ to a pure state in $B$. For a pure state $\psi_A$, its period is defined by
\begin{align}
    \tau\coloneqq \inf_{t>0}\left\{t\midd \Braket{\psi_A|e^{-\ii H_A t}|\psi_A}|=1\right\}.
\end{align}
Assuming that $\tau$ is finite and non-zero, a state $\psi_A$ is mapped by a covariant operation to a state with a finite period $\tau/k$ for some positive integer $k$ or a symmetric state. Of particular interest here is the former case since the latter case is trivial as any state can be mapped to any symmetric state with a covariant operation. 
Therefore, we analyze the convertibility under the assumption that $\phi_{B}$ has a finite period $\tau' =\tau/k$ with some positive integer $k$. 

Let $H_A=\sum_{E\in\mathrm{Spec}(H_A)} E\Pi_E^{(A)}$ be the spectral decomposition of the Hamiltonian, where $\mathrm{Spec}(H_A)$ is the set of different eigenvalues of $H_A$ and $\Pi_E^{(A)}$ denotes the projector to the eigenspace with eigenvalue $E$. For a given pure state $\psi_A$, we define
\begin{align}
    \mathrm{Spec}_\psi(H_A)\coloneqq \left\{E\in \mathrm{Spec}(H_A)\middle| \Pi_E^{(A)}\psi_A\neq 0\right\}.
\end{align}
 When the pure state $\psi_A$ has a finite period $\tau$, the eigenvalues in $\mathrm{Spec}_\psi(H_A)$ are expressed as $\frac{2\pi}{\tau}n+E_0$ with some $n\in\mathbb{Z}_{\geq 0}$ and a constant $E_0$. Shifting the Hamiltonian by a constant, without loss of generality, we can set $E_0=0$. 
In the same way, we introduce
\begin{align}
    \mathrm{Spec}_\phi(H_B)\coloneqq \left\{E\in \mathrm{Spec}(H_B)\middle| \Pi_E^{(B)}\phi_B\neq 0\right\}
\end{align}
for a pure state $\phi_B$.
When the pure state $\phi_B$ has period $\tau/k$ for some positive integer $k$, shifting the Hamiltonian $H_B$ by a constant, without loss of generality, all eigenvalues in $\mathrm{Spec}_\phi(H_B)$ are expressed as $\frac{2\pi}{\tau}kn$ for some $n\in\mathbb{Z}_{\geq 0}$. 
Finally, by multiplying the Hamiltonians $H_A$ and $H_B$ by $\tau/(2\pi)$, we can assume that $\mathrm{Spec}_\psi(H_A)\subset\mathbb{Z}_{\geq 0}$ and $\mathrm{Spec}_\phi(H_B)\subset \mathbb{Z}_{\geq 0}$. 

For pure states $\psi_{A}$ and $\phi_{B}$, we define their energy distributions by
\begin{align}
   p_{\psi}(n)&\coloneqq \braket{\psi|\Pi_n^{(A)}|\psi},\quad n \in \mathrm{Spec}_\psi(H_{A}), \label{eq:energy_distr_psi}\\
    p_{\phi}(n)&\coloneqq \braket{\phi|\Pi_n^{(B)}|\phi} \quad n \in \mathrm{Spec}_\phi(H_{B}).\label{eq:energy_distr_phi}
\end{align}
By using the energy distributions, the pure states $\psi_{A}$ and $\phi_{B}$ are expanded as
\begin{align}
    \ket{\psi}_A=\sum_{n\in\mathrm{Spec}_\psi(H_{A})}\sqrt{p_{\psi}(n)}\ket{\psi_{n}}_A,\quad \ket{\phi}_{B}=\sum_{n\in\mathrm{Spec}_\phi(H_{B})}\sqrt{p_{\phi}(n)}\ket{\phi_{n}}_B,
\end{align}
where $\ket{\psi_{n}}_A$ and $\ket{\phi_{n}}_B,$ are eigenvectors of $H_A$ and $H_B$ with eigenvalue $n$, respectively. 

Let us introduce a harmonic oscillator system whose Hamiltonian is given by
\begin{align}
    H_{\mathrm{HO}}=\sum_{n=0}^\infty n\ket{n}\bra{n}_{\mathrm{HO}},
\end{align}
where $\{\ket{n}_{\mathrm{HO}}\}_{n=0}^\infty$ denotes the Fock basis of the Harmonic oscillator system. 
By using energy distributions defined in Eqs.~\eqref{eq:energy_distr_psi} and \eqref{eq:energy_distr_phi}, we define pure states
\begin{align}
    \ket{\psi'}_{\mathrm{HO}}&\coloneqq \sum_{n\in\mathrm{Spec}_\psi(H_A)}\sqrt{p_{\psi}(n)}\ket{n}_{\mathrm{HO}}\in\mathcal{H}_{\mathrm{HO}},\quad 
    \ket{\phi'}_{\mathrm{HO}}\coloneqq \sum_{n\in\mathrm{Spec}_\phi(H_B)}\sqrt{p_{\phi}(n)}\ket{n}_{\mathrm{HO}}\in\mathcal{H}_{\mathrm{HO}},
\label{eq:pure_state_in_HO}
\end{align}
where $\mathcal{H}_{\mathrm{HO}}=\mathrm{Span}\{\ket{n}_{\mathrm{HO}}\}_{n=0}^\infty$ denotes the Hilbert space for the harmonic oscillator system.

We now prove a lemma on one-shot convertibility among pure states with and without an error. 
\begin{lemma}\label{lem:one_shot_conversion}
    Let $\psi_A$ and $\phi_B$ be pure states on $A$ and $B$ with finite periods $\tau$ and $\tau'$, respectively. Assume that $\tau'=\tau/k$ for some positive integer $k$. Define pure states $\psi'_{\mathrm{HO}}$ and $\phi'_{\mathrm{HO}}$ in a harmonic oscillator system by Eq.~\eqref{eq:pure_state_in_HO}. Then, the following two hold:
\begin{enumerate}[(i)]
    \item $(\psi_A,H_A)\cov (\phi_B,H_B)\Longleftrightarrow(\psi'_{\mathrm{HO}},H_{\mathrm{HO}})\cov (\phi'_{\mathrm{HO}},H_{\mathrm{HO}})$.
    \item For any $\epsilon\in(0,1]$, $(\psi_A,H_A)\cov_\epsilon (\phi_B,H_B)\Longleftrightarrow(\psi'_{\mathrm{HO}},H_{\mathrm{HO}})\cov_\epsilon (\phi'_{\mathrm{HO}},H_{\mathrm{HO}})$
\end{enumerate}
\end{lemma}
\begin{proof}
Let us first introduce maps relating the original systems $A$ and $B$ to a harmonic oscillator system. 
Defining operators
\begin{align}
    \Pi^{(\psi)}&\coloneqq \sum_{n\in\mathrm{Spec}_\psi(H_A)}\ket{n}_{\mathrm{HO}}\bra{\psi_{n}}_A:\mathcal{H}_{A}\to\mathcal{H}_{\mathrm{HO}},\quad 
     \Pi^{(\phi)}\coloneqq \sum_{n\in\mathrm{Spec}_\phi(H_B)}\ket{n}_{\mathrm{HO}}\bra{\phi_{n}}_B:\mathcal{H}_B\to\mathcal{H}_{\mathrm{HO}},
\end{align}
we have
\begin{align}
    \Pi^{(\psi)}\ket{\psi}_A&=\ket{\psi'}_{\mathrm{HO}},\quad \Pi^{(\psi)\dag}\ket{\psi'}_{\mathrm{HO}}=\ket{\psi}_A,\quad
    \Pi^{(\phi)}\ket{\phi}_B=\ket{\phi'}_{\mathrm{HO}},\quad \Pi^{(\phi)\dag}\ket{\phi'}_{\mathrm{HO}}=\ket{\phi}_B.
\end{align}
Furthermore, they satisfy
\begin{align}
    \Pi^{(\psi)}H_A=H_{\mathrm{HO}}\Pi^{(\psi)},\quad \Pi^{(\phi)}H_B=H_{\mathrm{HO}}\Pi^{(\phi)}.\label{eq:pi_Hamiltonian_commute}
\end{align}
We introduce maps
\begin{align}
    \mathcal{M}_{A\to\mathrm{HO}}\left(\cdots\right) &\coloneqq \Pi^{(\psi)}\left(\cdots\right)\Pi^{(\psi)\dag}+\mathrm{Tr}\left(\left(\cdots\right)\left(\mathbb{I}_A-\sum_{n\in\mathrm{Spec}_\psi(H_A)}\ket{\psi_n}\bra{\psi_{n}}_A\right)\right)\rho_{\mathrm{HO}}^{(\mathrm{free})}:\mathcal{L}(\mathcal{H}_A)\to\mathcal{L}(\mathcal{H}_{\mathrm{HO}}),\\
    \mathcal{M}_{\mathrm{HO}\to A}\left(\cdots\right) &\coloneqq \Pi^{(\psi)\dag}\left(\cdots\right)\Pi^{(\psi)}+\mathrm{Tr}\left(\left(\cdots\right)\left(\mathbbm{I}_{\mathrm{HO}}-\sum_{n\in\mathrm{Spec}_\psi(H_A)}\ket{n}\bra{n}_{\mathrm{HO}}\right)\right)\rho_A^{(\mathrm{free})}:\mathcal{L}(\mathcal{H}_{\mathrm{HO}})\to\mathcal{L}(\mathcal{H}_A),\\
    \mathcal{M}_{B\to\mathrm{HO}}\left(\cdots\right) &\coloneqq \Pi^{(\phi)}\left(\cdots\right)\Pi^{(\phi)\dag}+\mathrm{Tr}\left(\left(\cdots\right)\left(\mathbb{I}_B-\sum_{n\in\mathrm{Spec}_\phi(H_B)}\ket{\phi_n}\bra{\phi_{n}}_B\right)\right)\rho_{\mathrm{HO}}^{(\mathrm{free})}:\mathcal{L}(\mathcal{H}_B)\to\mathcal{L}(\mathcal{H}_{\mathrm{HO}}),\\
    \mathcal{M}_{\mathrm{HO}\to B}\left(\cdots\right) &\coloneqq \Pi^{(\phi)\dag}\left(\cdots\right)\Pi^{(\phi)}+\mathrm{Tr}\left(\left(\cdots\right)\left(\mathbbm{I}_{\mathrm{HO}}-\sum_{n\in\mathrm{Spec}_\phi(H_B)}\ket{n}\bra{n}_{\mathrm{HO}}\right)\right)\rho_B^{(\mathrm{free})}:\mathcal{L}(\mathcal{H}_{\mathrm{HO}})\to\mathcal{L}(\mathcal{H}_B),
\end{align}
where $\rho_{\mathrm{HO}}^{(\mathrm{free})}$, $\rho_A^{(\mathrm{free})}$ and $\rho_B^{(\mathrm{free})}$ are any symmetric states, and $\mathbb{I}$ denotes the identity operator for each system. From Eq.~\eqref{eq:pi_Hamiltonian_commute}, it can be confirmed that these maps are covariant channels. In addition, they satisfy
\begin{align}
    \mathcal{M}_{A\to\mathrm{HO}}(\psi_A)=\psi'_{\mathrm{HO}},\quad \mathcal{M}_{\mathrm{HO}\to A}(\psi'_{\mathrm{HO}})=\psi_A,\quad \mathcal{M}_{B\to\mathrm{HO}}(\phi_B)=\phi'_{\mathrm{HO}},\quad \mathcal{M}_{\mathrm{HO}\to B}(\phi'_{\mathrm{HO}})=\phi_B.\label{eq:m_maps_purestates}
\end{align}

To prove the claim (i), let us first assume $(\psi_A,H_A)\cov (\phi_B,H_B)$. Then there exists a covariant channel $\mathcal{E}_{A\to B}:\mathcal{L}(\mathcal{H}_A)\to\mathcal{L}(\mathcal{H}_B)$ such that $\mathcal{E}_{A\to B}(\psi_A)=\phi_B$. Introducing a covariant channel
\begin{align}
    \mathcal{E}'_{\mathrm{HO}\to\mathrm{HO}}\coloneqq \mathcal{M}_{B\to\mathrm{HO}}\circ \mathcal{E}_{A\to B}\circ\mathcal{M}_{\mathrm{HO}\to A}:\mathcal{L}(\mathcal{H}_{\mathrm{HO}})\to \mathcal{L}(\mathcal{H}_{\mathrm{HO}}),\label{eq:cov_channel_HO_to_HO}
\end{align}
we get
\begin{align}
    \mathcal{E}'_{\mathrm{HO}\to\mathrm{HO}}\left(\psi'_{\mathrm{HO}}\right)&=\mathcal{M}_{B\to\mathrm{HO}}\circ \mathcal{E}_{A\to B}\circ\mathcal{M}_{\mathrm{HO}\to A}\left(\psi'_{\mathrm{HO}}\right)\\
    &=\mathcal{M}_{B\to\mathrm{HO}}\circ \mathcal{E}_{A\to B}(\psi_A)\\
    &=\mathcal{M}_{B\to\mathrm{HO}}(\phi_B)\\
    &=\phi'_{\mathrm{HO}},
\end{align}
where we used Eq.~\eqref{eq:m_maps_purestates} in the second and forth lines. Therefore, $(\psi'_{\mathrm{HO}},H_{\mathrm{HO}})\cov(\phi'_{\mathrm{HO}},H_{\mathrm{HO}})$.
The relations among these maps are summarized in the following diagram:
\begin{equation}
    \begin{gathered}
        \xymatrix@R=40pt@C=55pt{
    \ket{\psi}_{A}=\sum_{n}\sqrt{p_\psi(n)}\ket{\psi_n}\ar[r]^{\mathcal{E}_{A\to B}} & \ket{\phi}_B=\sum_{n}\sqrt{p_\phi(n)}\ket{\phi_n}\ar[d]^{\mathcal{M}_{B\to \mathrm{HO}}}\\
    \ket{\psi'}_{\mathrm{HO}}=\sum_{n}\sqrt{p_\psi(n)}\ket{n}_{\mathrm{HO}}\ar[u]^{\mathcal{M}_{\mathrm{HO} \to A}}\ar[r]_{\mathcal{E}'_{\mathrm{HO}\to\mathrm{HO}}}&\ket{\phi'}_{\mathrm{HO}}=\sum_{n}\sqrt{p_\phi(n)}\ket{n}_{\mathrm{HO}}
    }
    \end{gathered}
\end{equation}

Conversely, assume $(\psi'_{\mathrm{HO}},H_{\mathrm{HO}})\cov(\phi'_{\mathrm{HO}},H_{\mathrm{HO}})$. Then there exists a covariant channel $\Lambda'_{\mathrm{HO}\to\mathrm{HO}}:\mathcal{L}(\mathcal{H}_{\mathrm{HO}})\to\mathcal{L}(\mathcal{H}_{\mathrm{HO}})$ such that $\Lambda'_{\mathrm{HO}\to\mathrm{HO}}(\psi'_{\mathrm{HO}})=\phi'_{\mathrm{HO}}$. Defining a covariant channel
\begin{align}
    \Lambda_{A\to B}\coloneqq \mathcal{M}_{\mathrm{HO}\to B}\circ\Lambda'_{\mathrm{HO}\to\mathrm{HO}}\circ\mathcal{M}_{A\to\mathrm{HO}} :\mathcal{L}(\mathcal{H}_A)\to\mathcal{L}(\mathcal{H}_{B}),\label{eq:cov_channel_A_to_B}
\end{align}
it holds $\Lambda_{A\to B}(\psi_A)=\phi_B$, implying that $(\psi_A,H_A)\cov (\phi_B,H_B)$. The relations among these maps are summarized in the following diagram:
\begin{equation}
    \begin{gathered}
        \xymatrix@R=40pt@C=55pt{
     \ket{\psi}_{A}=\sum_{n}\sqrt{p_\psi(n)}\ket{\psi_n}\ar[d]_{\mathcal{M}_{A\to \mathrm{HO}}} \ar[r]^{\Lambda_{A\to B}} & \ket{\phi}_B=\sum_{n}\sqrt{p_\phi(n)}\ket{\phi_n}\\
   \ket{\psi'}_{\mathrm{HO}}=\sum_{n}\sqrt{p_\psi(n)}\ket{n}_{\mathrm{HO}}\ar[r]_{\Lambda'_{\mathrm{HO}\to\mathrm{HO}}}&\ket{\phi'}_{\mathrm{HO}}=\sum_{n}\sqrt{p_\phi(n)}\ket{n}_{\mathrm{HO}}\ar[u]_{\mathcal{M}_{\mathrm{HO}\to B}}
    }
    \end{gathered}
\end{equation}

The claim (ii) can be shown by extending the proof for claim (i) with the data-processing inequality. Assume $(\psi_A,H_A)\cov_\epsilon (\phi_B,H_B)$. Then there exists a covariant channel $\mathcal{E}_{A\to B}:\mathcal{L}(\mathcal{H}_A)\to\mathcal{L}(\mathcal{H}_B)$ such that $D(\mathcal{E}_{A\to B}(\psi_A),\phi_B)\leq \epsilon$. For a covariant channel $\mathcal{E}_{\mathrm{HO}\to\mathrm{HO}}$ defined in Eq.~\eqref{eq:cov_channel_HO_to_HO}, we have
\begin{align}
    D\left(\mathcal{E}'_{\mathrm{HO}\to\mathrm{HO}}\left(\psi'_{\mathrm{HO}}\right),\phi'_{\mathrm{HO}}\right)&=D\left(\mathcal{M}_{B\to\mathrm{HO}}\circ \mathcal{E}_{A\to B}(\psi_A),\mathcal{M}_{B\to\mathrm{HO}}(\phi_B)\right)\nonumber\\
    &\leq D\left(\mathcal{E}_{A\to B}(\psi_A),\phi_B\right),
\end{align}
implying that $(\psi'_{\mathrm{HO}},H_{\mathrm{HO}})\cov_\epsilon (\phi'_{\mathrm{HO}},H_{\mathrm{HO}})$. 
Conversely, assume $(\psi'_{\mathrm{HO}},H_{\mathrm{HO}})\cov_\epsilon (\phi'_{\mathrm{HO}},H_{\mathrm{HO}})$. Then there exists a covariant channel $\Lambda_{\mathrm{HO}\to \mathrm{HO}}:\mathcal{L}(\mathcal{H}_{\mathrm{HO}})\to\mathcal{L}(\mathcal{H}_{\mathrm{HO}})$ such that $D(\Lambda_{\mathrm{HO}\to \mathrm{HO}}(\psi'_{\mathrm{HO}}),\phi'_{\mathrm{HO}})\leq \epsilon$. For a covariant channel $\Lambda_{A\to B}$ defined in Eq.~\eqref{eq:cov_channel_A_to_B}, we get 
\begin{align}
    D\left(\Lambda_{A\to B}\left(\psi_A\right),\phi_B\right)&=D\left(\mathcal{M}_{\mathrm{HO}\to B}\circ\Lambda'_{\mathrm{HO}\to\mathrm{HO}}(\psi'_{\mathrm{HO}}),\mathcal{M}_{\mathrm{HO}\to B}(\phi'_{\mathrm{HO}})\right)\nonumber\\
    &\leq D\left(\Lambda'_{\mathrm{HO}\to\mathrm{HO}}(\psi'_{\mathrm{HO}}),\phi'_{\mathrm{HO}}\right),
\end{align}
implying that $(\psi_A,H_A)\cov_\epsilon (\phi_B,H_B)$. 
\end{proof}
Note that claim (ii) in this theorem is valid even when one adopts any quantum distance satisfying the data-processing inequality instead of the trance distance to quantify the error.

We here extend the above arguments to asymptotic convertibility. Let $\widehat{\rho}_A=\{\rho_{A_m}\}_m$ and $\widehat{\sigma}_B=\{\sigma_{B_m}\}_m$ be sequence of states of systems $A_m$ and $B_m$, respectively. We assume that the Hamiltonians $\widehat{H}_A=\{H_{A_m}\}_m$ and $\widehat{H}_A=\{H_{B_m}\}_m$ are bounded below. For $\epsilon\in(0,1]$, we denote $(\widehat{\psi}_A,\widehat{H}_A)\cov_\epsilon (\widehat{\phi}_B,\widehat{H}_B)$ if and only if there exists a sequence of covariant operations $\{\mathcal{E}_m\}_m$ such that $D(\mathcal{E}_m(\rho_m),\sigma_m)\leq \epsilon$ for all sufficiently large $m$. If $(\widehat{\psi}_A,\widehat{H}_A)\cov_\epsilon (\widehat{\phi}_B,\widehat{H}_B)$ holds for all $\epsilon\in (0,1]$, then we denote $(\widehat{\psi}_A,\widehat{H}_A)\cov(\widehat{\phi}_B,\widehat{H}_B)$. If no confusion arises, we omit Hamiltonians and simply write $\widehat{\rho}\cov_\epsilon \widehat{\sigma}$ and $\widehat{\rho}\cov\widehat{\sigma}$ as in the main text.

As a corollary of Lemma~\ref{lem:one_shot_conversion}, we prove that the asymptotic convertibility among pure states can be analyzed with harmonic oscillator systems. 
\begin{corollary}\label{cor:asymptotic_convertibility_general_to_ho}
Let $\widehat{\psi}_A=\{\psi_{A_m}\}_m$ and $\widehat{\phi}_B=\{\phi_{B_m}\}_m$ be sequences of pure states of systems with Hamiltonians $\widehat{H}_A=\{H_{A_m}\}_m$ and $\widehat{H}_B=\{H_{B_m}\}_m$, respectively. Assume that the periods $\tau_m$ and $\tau_m'$ of $\psi_{A_m}$ and $\phi_{B_m}$ are finite and satisfy $\tau_m'=\tau_m/k_m$ for some positive integers $k_m$. Without loss of generality, by multiplying and shifting the Hamiltonians $H_{A_m}$ and $H_{B_m}$, we can assume that $\mathrm{Spec}_{\psi_{A_m}}(H_{A_m})\subset\mathbb{Z}_{\geq 0}$ and $\mathrm{Spec}_{\phi_{B_m}}(H_{B_m})\subset\mathbb{Z}_{\geq 0}$. Define sequences of pure states in the harmonic oscillator systems $\widehat{\psi}_{\mathrm{HO}}=\{\psi_{\mathrm{HO},m}'\}_m$ and $\widehat{\phi}_{\mathrm{HO}}=\{\phi_{\mathrm{HO},m}'\}_m$ by
\begin{align}
    \ket{\psi_{\mathrm{HO},m}'}&\coloneqq \sum_{n\in\mathrm{Spec}_{\psi_{A_m}}(H_{A_m})}\sqrt{p_{\psi_{A_m}}(n)}\ket{n}_{\mathrm{HO}}\in\mathcal{H}_{\mathrm{HO}},\quad 
    \ket{\phi_{\mathrm{HO},m}'}\coloneqq \sum_{n\in\mathrm{Spec}_{\phi_{B_m}}(H_{B_m})}\sqrt{p_{\phi_m}(n)}\ket{n}_{\mathrm{HO}}\in\mathcal{H}_{\mathrm{HO}},
\end{align}
with Hamiltonian $\widehat{H}_{\mathrm{HO}}\coloneqq \{H_{m}^{(\mathrm{HO})}\}_m$ and $H_{m}^{(\mathrm{HO})}\coloneqq \sum_{n=0}^\infty\ket{n}\bra{n}$ for all $m$. Then, it holds
\begin{align}
    \forall \epsilon\in(0,1],\quad (\widehat{\psi}_A,\widehat{H}_A)\cov_\epsilon(\widehat{\phi}_B,\widehat{H}_B)\Longleftrightarrow(\widehat{\psi}_{\mathrm{HO}}',\widehat{H}_{\mathrm{HO}})\cov_\epsilon(\widehat{\phi}_{\mathrm{HO}}',\widehat{H}_{\mathrm{HO}}),
\end{align}
or equivalently
\begin{align}
    (\widehat{\psi}_A,\widehat{H}_A)\cov(\widehat{\phi}_B,\widehat{H}_B)\Longleftrightarrow(\widehat{\psi}_{\mathrm{HO}}',\widehat{H}_{\mathrm{HO}})\cov(\widehat{\phi}_{\mathrm{HO}}',\widehat{H}_{\mathrm{HO}}).
\end{align}
\end{corollary}
\begin{proof}
    $(\widehat{\psi}_A,\widehat{H}_A)\cov_\epsilon(\widehat{\phi}_B,\widehat{H}_B)$ implies that there exist covariant operations $\{\mathcal{E}_{A_m\to B_m}\}_m$ and $M>0$ such that $D(\mathcal{E}_{A_m\to B_m}(\psi_{A_m}),\phi_{B_m})\leq \epsilon$ for all $m\geq M$. From Lemma~\ref{lem:one_shot_conversion}, there also exist covariant operations $\{\mathcal{E}'_{\mathrm{HO}\to\mathrm{HO}}\}$ such that and $M>0$ such that $D(\mathcal{E}'_{\mathrm{HO}\to\mathrm{HO}}(\psi'_{\mathrm{HO},m}),\phi_{\mathrm{HO},m}')\leq \epsilon$ for all $m\geq M$ and hence $(\widehat{\psi}'_{\mathrm{HO}},\widehat{H}_{\mathrm{HO}})\cov_\epsilon(\widehat{\phi}'_{\mathrm{HO}},\widehat{H}_{\mathrm{HO}})$. In the same way, it is proven that $(\widehat{\psi}'_{\mathrm{HO}},\widehat{H}_{\mathrm{HO}})\cov_\epsilon(\widehat{\phi}'_{\mathrm{HO}},\widehat{H}_{\mathrm{HO}})$ implies $(\widehat{\psi}_A,\widehat{H}_A)\cov_\epsilon(\widehat{\phi}_B,\widehat{H}_B)$. 
\end{proof}

\section{General formula for the coherence cost and the distillable coherence for pure states of arbitrary systems}
Based on the results in the previous section, we here generalize Theorem~\ref{thm:cost_dist_Finf_F0}. Let $(\widehat{\psi},\widehat{H})$ be a set of sequences of pure states $\widehat{\psi}=\{\psi_m\}$ and Hamiltonians $\widehat{H}=\{H_m\}$. We assume that $\psi_m$ has a finite period $\tau$ for all $m$. As a reference state with period $\tau$, we adopt the coherence bit
\begin{align}
    \ket{\phi_{\mathrm{coh}}}\coloneqq \frac{1}{\sqrt{2}}\left(\ket{0}+\ket{1}\right)
\end{align}
with Hamiltonian
\begin{align}
    H_{\mathrm{coh}}=\frac{2\pi}{\tau}\ket{1}\bra{1}. 
\end{align}
For $R>0$, the i.i.d. sequence of coherence bits are defined by
\begin{align}
    \widehat{\phi_{\mathrm{coh}}}(R)\coloneqq\{ \phi_{\mathrm{coh}}^{\otimes \ceil{Rm}}\}_m,\quad \phi_{\mathrm{coh}}\coloneqq \ket{\phi_{\mathrm{coh}}}\bra{\phi_{\mathrm{coh}}}
\end{align}
with a sequence of Hamiltonians
\begin{align}
    \widehat{H}_{\mathrm{coh}}(R)&\coloneqq \{H_{\mathrm{coh},m}(R)\}_m,\quad 
    H_{\mathrm{coh},m}(R)\coloneqq \sum_{i=1}^{\ceil{Rm}}H_{\mathrm{coh}}^{(i)},\quad 
    H_{\mathrm{coh}}^{(i)}\coloneqq \mathbb{I}^{\otimes i-1}\otimes H_{\mathrm{coh}}\otimes \mathbb{I}^{\otimes \ceil{Rm}-i}.
\end{align}

We here denote the coherence cost and the distillable coherence by
\begin{align}
    C_{\mathrm{cost}}\left(\widehat{\psi},\widehat{H}\right)&\coloneqq \inf\left\{R\middle|\left(\widehat{\phi_{\mathrm{coh}}}(R),\widehat{H}_{\mathrm{coh}}(R)\right)\cov \left(\widehat{\psi},\widehat{H}\right)\right\},\\
     C_{\mathrm{dist}}\left(\widehat{\psi},\widehat{H}\right)&\coloneqq \sup\left\{R\middle|\left(\widehat{\psi},\widehat{H}\right)\cov \left(\widehat{\phi_{\mathrm{coh}}}(R),\widehat{H}_{\mathrm{coh}}(R)\right)\right\}
\end{align}
to emphasize that these quantities depend not only on the states but also the Hamiltonians. By shifting the Hamiltonians, we can assume that $\mathrm{Spec}_{\psi_m}(H_m)=\frac{2\pi}{\tau} n$ for some $n\in\mathbb{Z}_{\geq 0}$. For the sequence of pure states $\widehat{\psi}=\{\psi_m\}$, we define a sequence of states $\widehat{\psi}_{\mathrm{HO}}=\{\psi_{\mathrm{HO},m}\}_m$ in a harmonic oscillator system by
\begin{align}
    \ket{\psi_{\mathrm{HO},m}}&\coloneqq \sum_{\frac{2\pi}{\tau}n\in\mathrm{Spec}_{\psi_m}(H_m)}\sqrt{p_{\psi_m}\left(\frac{2\pi}{\tau}n\right)}\ket{n}_{\mathrm{HO}}\in\mathcal{H}_{\mathrm{HO}}\label{eq:harmonic_oscillator_copy}
\end{align}
with Hamiltonian
\begin{align}
    H_{\mathrm{HO}}=\frac{2\pi}{\tau}\sum_{n=0}^\infty n\ket{n}\bra{n}.
\end{align}
Note that if we multiply all the Hamiltonians by $\frac{\tau}{2\pi}$, we can assume that the period is $2\pi$ without loss of generality.
By using the sequence of states $\widehat{\psi}=\{\psi_m\}$, we define
\begin{align}
    \overline{\mathcal{F}}\left(\widehat{\psi},\widehat{H}\right)\coloneqq \overline{\mathcal{F}}\left(\widehat{\psi}_{\mathrm{HO}}\right),\quad \underline{\mathcal{F}}\left(\widehat{\psi},\widehat{H}\right)\coloneqq \underline{\mathcal{F}}\left(\widehat{\psi}_{\mathrm{HO}}\right),
\end{align}
where the right hand sides are the spectral sup- and inf-QFI rates defined in Eqs.~\eqref{eq:spectral_sup_rate} and \eqref{eq:spectral_inf_rate} in the main text.

From Theorem~\ref{thm:cost_dist_Finf_F0} and Corollary~\ref{cor:asymptotic_convertibility_general_to_ho}, we get a general formula for the coherence cost and the distillable coherence for pure states in general systems with arbitrary Hamiltonians:
\begin{theorem}
    For any sequences of pure states  $\widehat{\psi}=\{\psi_m\}$ and Hamiltonians $\widehat{H}=\{H_m\}$ with a finite period, it holds
    \begin{align}
        C_{\mathrm{cost}}\left(\widehat{\psi},\widehat{H}\right)=\overline{\mathcal{F}}\left(\widehat{\psi},\widehat{H}\right),\quad C_{\mathrm{dist}}\left(\widehat{\psi},\widehat{H}\right)=\underline{\mathcal{F}}\left(\widehat{\psi},\widehat{H}\right).
    \end{align}
\end{theorem}

As a consistency check of the result in the i.i.d. regime \cite{gour_resource_2008,marvian_operational_2022}, we show that
\begin{align}
    \overline{\mathcal{F}}\left(\widehat{\psi}_{\mathrm{iid}},\widehat{H}_{\mathrm{iid}}\right)=\overline{\mathcal{F}}\left(\widehat{\psi}_{\mathrm{iid}},\widehat{H}_{\mathrm{iid}}\right)=\mathcal{F}(\psi,H)\label{eq:iid_consistency}
\end{align}
for i.i.d. sequence of pure states $\widehat{\psi}_{\mathrm{iid}}=\{\psi^{\otimes m}\}$ and Hamiltonians $\{\widehat{H}_{\mathrm{iid}}\}=\{H_{\mathrm{iid},m}\}$ with period $2\pi$, where $H_{\mathrm{iid},m}\coloneqq \sum_{i=1}^m\mathbb{I}^{\otimes i-1}\otimes H\otimes \mathbb{I}^{\otimes m-i}$. First, following Eq.~\eqref{eq:harmonic_oscillator_copy}, let us introduce 
\begin{align}
    \ket{\psi_{\mathrm{HO},m}}&\coloneqq \sum_{n\in\mathrm{Spec}_{\psi^{\otimes m}}(H_{\mathrm{iid},m})}\sqrt{p_{\psi^{\otimes m}}\left(n\right)}\ket{n}_{\mathrm{HO}}
\end{align}
generalizing Eq.~\eqref{eq:convergence_iid_energy_distr}, we find that
\begin{align}
    \lim_{m\to\infty}\dtv\left(p_{\psi^{\otimes m}}, \mathrm{P}_{m \frac{1}{4}\mathcal{F}(\psi,H)}\right)=0,
\end{align}
implying that $\forall\epsilon\in(0,1]$, 
\begin{align}
   m \mathcal{F}(\psi,H)\geq \mathcal{F}_{\max}^\epsilon (\psi_{\mathrm{HO},m}),\quad \mathcal{F}_{\min}^\epsilon (\psi_{\mathrm{HO},m})\geq m \mathcal{F}(\psi,H)
\end{align}
hold for all sufficiently large $m$. Therefore, we get
\begin{align}
    \mathcal{F}(\psi,H)\geq\overline{\mathcal{F}}\left(\widehat{\psi}_{\mathrm{iid}},\widehat{H}_{\mathrm{iid}}\right),\quad \overline{\mathcal{F}}\left(\widehat{\psi}_{\mathrm{iid}},\widehat{H}_{\mathrm{iid}}\right)\geq \mathcal{F}(\psi,H).
\end{align}
Thus, in order to prove Eq.~\eqref{eq:iid_consistency}, it is sufficient to show
\begin{align}
    \overline{\mathcal{F}}\left(\widehat{\psi},\widehat{H}\right)\geq \underline{\mathcal{F}}\left(\widehat{\psi},\widehat{H}\right),
\end{align}
or equivalently,
\begin{align}
    C_{\mathrm{cost}}\left(\widehat{\psi},\widehat{H}\right)\geq C_{\mathrm{dist}}\left(\widehat{\psi},\widehat{H}\right).\label{eq:cost_geq_dist}
\end{align}
Although Eq.~\eqref{eq:cost_geq_dist} might seem to be trivial, technically, one must prove
\begin{align}
    \widehat{\phi_{\mathrm{coh}}}(R)\cov \widehat{\phi_{\mathrm{coh}}}(R') \implies  R\geq R'.\label{eq:monotonicity_iid_cbit}
\end{align}
By using the results in \cite{yamaguchi_smooth_2022}, Eq.~\eqref{eq:monotonicity_iid_cbit} can be proven as a consequence of asymptotic monotonicity of smooth metric adjusted skew information rates.

\begin{thebibliography}{70}%
\makeatletter
\providecommand \@ifxundefined [1]{%
 \@ifx{#1\undefined}
}%
\providecommand \@ifnum [1]{%
 \ifnum #1\expandafter \@firstoftwo
 \else \expandafter \@secondoftwo
 \fi
}%
\providecommand \@ifx [1]{%
 \ifx #1\expandafter \@firstoftwo
 \else \expandafter \@secondoftwo
 \fi
}%
\providecommand \natexlab [1]{#1}%
\providecommand \enquote  [1]{``#1''}%
\providecommand \bibnamefont  [1]{#1}%
\providecommand \bibfnamefont [1]{#1}%
\providecommand \citenamefont [1]{#1}%
\providecommand \href@noop [0]{\@secondoftwo}%
\providecommand \href [0]{\begingroup \@sanitize@url \@href}%
\providecommand \@href[1]{\@@startlink{#1}\@@href}%
\providecommand \@@href[1]{\endgroup#1\@@endlink}%
\providecommand \@sanitize@url [0]{\catcode `\\12\catcode `\$12\catcode
  `\&12\catcode `\#12\catcode `\^12\catcode `\_12\catcode `\%12\relax}%
\providecommand \@@startlink[1]{}%
\providecommand \@@endlink[0]{}%
\providecommand \url  [0]{\begingroup\@sanitize@url \@url }%
\providecommand \@url [1]{\endgroup\@href {#1}{\urlprefix }}%
\providecommand \urlprefix  [0]{URL }%
\providecommand \Eprint [0]{\href }%
\providecommand \doibase [0]{http://dx.doi.org/}%
\providecommand \selectlanguage [0]{\@gobble}%
\providecommand \bibinfo  [0]{\@secondoftwo}%
\providecommand \bibfield  [0]{\@secondoftwo}%
\providecommand \translation [1]{[#1]}%
\providecommand \BibitemOpen [0]{}%
\providecommand \bibitemStop [0]{}%
\providecommand \bibitemNoStop [0]{.\EOS\space}%
\providecommand \EOS [0]{\spacefactor3000\relax}%
\providecommand \BibitemShut  [1]{\csname bibitem#1\endcsname}%
\let\auto@bib@innerbib\@empty
\bibitem [{\citenamefont {Giovannetti}\ \emph {et~al.}(2001)\citenamefont
  {Giovannetti}, \citenamefont {Lloyd},\ and\ \citenamefont
  {Maccone}}]{giovannetti_quantum-enhanced_2001}%
  \BibitemOpen
  \bibfield  {author} {\bibinfo {author} {\bibfnamefont {V.}~\bibnamefont
  {Giovannetti}}, \bibinfo {author} {\bibfnamefont {S.}~\bibnamefont {Lloyd}},
  \ and\ \bibinfo {author} {\bibfnamefont {L.}~\bibnamefont {Maccone}},\ }\href
  {\doibase 10.1038/35086525} {\bibfield  {journal} {\bibinfo  {journal}
  {Nature}\ }\textbf {\bibinfo {volume} {412}},\ \bibinfo {pages} {417}
  (\bibinfo {year} {2001})}\BibitemShut {NoStop}%
\bibitem [{\citenamefont {Giovannetti}\ \emph {et~al.}(2006)\citenamefont
  {Giovannetti}, \citenamefont {Lloyd},\ and\ \citenamefont
  {Maccone}}]{giovannetti_quantum_2006}%
  \BibitemOpen
  \bibfield  {author} {\bibinfo {author} {\bibfnamefont {V.}~\bibnamefont
  {Giovannetti}}, \bibinfo {author} {\bibfnamefont {S.}~\bibnamefont {Lloyd}},
  \ and\ \bibinfo {author} {\bibfnamefont {L.}~\bibnamefont {Maccone}},\ }\href
  {\doibase 10.1103/PhysRevLett.96.010401} {\bibfield  {journal} {\bibinfo
  {journal} {Physical Review Letters}\ }\textbf {\bibinfo {volume} {96}},\
  \bibinfo {pages} {010401} (\bibinfo {year} {2006})}\BibitemShut {NoStop}%
\bibitem [{\citenamefont {Giacomini}\ \emph {et~al.}(2019)\citenamefont
  {Giacomini}, \citenamefont {Castro-Ruiz},\ and\ \citenamefont
  {Brukner}}]{giacomini_quantum_2019}%
  \BibitemOpen
  \bibfield  {author} {\bibinfo {author} {\bibfnamefont {F.}~\bibnamefont
  {Giacomini}}, \bibinfo {author} {\bibfnamefont {E.}~\bibnamefont
  {Castro-Ruiz}}, \ and\ \bibinfo {author} {\bibfnamefont {{\v
  C}.}~\bibnamefont {Brukner}},\ }\href {\doibase 10.1038/s41467-018-08155-0}
  {\bibfield  {journal} {\bibinfo  {journal} {Nature Communications}\ }\textbf
  {\bibinfo {volume} {10}},\ \bibinfo {pages} {494} (\bibinfo {year}
  {2019})}\BibitemShut {NoStop}%
\bibitem [{\citenamefont {Schnabel}\ \emph {et~al.}(2010)\citenamefont
  {Schnabel}, \citenamefont {Mavalvala}, \citenamefont {McClelland},\ and\
  \citenamefont {Lam}}]{schnabel_quantum_2010}%
  \BibitemOpen
  \bibfield  {author} {\bibinfo {author} {\bibfnamefont {R.}~\bibnamefont
  {Schnabel}}, \bibinfo {author} {\bibfnamefont {N.}~\bibnamefont {Mavalvala}},
  \bibinfo {author} {\bibfnamefont {D.~E.}\ \bibnamefont {McClelland}}, \ and\
  \bibinfo {author} {\bibfnamefont {P.~K.}\ \bibnamefont {Lam}},\ }\href
  {\doibase 10.1038/ncomms1122} {\bibfield  {journal} {\bibinfo  {journal}
  {Nature Communications}\ }\textbf {\bibinfo {volume} {1}},\ \bibinfo {pages}
  {121} (\bibinfo {year} {2010})}\BibitemShut {NoStop}%
\bibitem [{\citenamefont {Marvian}(2020)}]{marvian_coherence_2020}%
  \BibitemOpen
  \bibfield  {author} {\bibinfo {author} {\bibfnamefont {I.}~\bibnamefont
  {Marvian}},\ }\href {\doibase 10.1038/s41467-019-13846-3} {\bibfield
  {journal} {\bibinfo  {journal} {Nature Communications}\ }\textbf {\bibinfo
  {volume} {11}},\ \bibinfo {pages} {25} (\bibinfo {year} {2020})}\BibitemShut
  {NoStop}%
\bibitem [{\citenamefont {Woods}\ \emph {et~al.}(2019)\citenamefont {Woods},
  \citenamefont {Silva},\ and\ \citenamefont
  {Oppenheim}}]{woods_autonomous_2019}%
  \BibitemOpen
  \bibfield  {author} {\bibinfo {author} {\bibfnamefont {M.~P.}\ \bibnamefont
  {Woods}}, \bibinfo {author} {\bibfnamefont {R.}~\bibnamefont {Silva}}, \ and\
  \bibinfo {author} {\bibfnamefont {J.}~\bibnamefont {Oppenheim}},\ }\href
  {\doibase 10.1007/s00023-018-0736-9} {\bibfield  {journal} {\bibinfo
  {journal} {Annales Henri Poincaré}\ }\textbf {\bibinfo {volume} {20}},\
  \bibinfo {pages} {125} (\bibinfo {year} {2019})}\BibitemShut {NoStop}%
\bibitem [{\citenamefont {Marvian}\ \emph {et~al.}(2016)\citenamefont
  {Marvian}, \citenamefont {Spekkens},\ and\ \citenamefont
  {Zanardi}}]{marvian_quantum_2016}%
  \BibitemOpen
  \bibfield  {author} {\bibinfo {author} {\bibfnamefont {I.}~\bibnamefont
  {Marvian}}, \bibinfo {author} {\bibfnamefont {R.~W.}\ \bibnamefont
  {Spekkens}}, \ and\ \bibinfo {author} {\bibfnamefont {P.}~\bibnamefont
  {Zanardi}},\ }\href {\doibase 10.1103/PhysRevA.93.052331} {\bibfield
  {journal} {\bibinfo  {journal} {Physical Review A}\ }\textbf {\bibinfo
  {volume} {93}},\ \bibinfo {pages} {052331} (\bibinfo {year}
  {2016})}\BibitemShut {NoStop}%
\bibitem [{\citenamefont {Wigner}(1952)}]{wigner_messung_1952}%
  \BibitemOpen
  \bibfield  {author} {\bibinfo {author} {\bibfnamefont {E.~P.}\ \bibnamefont
  {Wigner}},\ }\href@noop {} {\bibfield  {journal} {\bibinfo  {journal} {Z.
  Phys.}\ }\textbf {\bibinfo {volume} {133}},\ \bibinfo {pages} {101} (\bibinfo
  {year} {1952})}\BibitemShut {NoStop}%
\bibitem [{\citenamefont {Araki}\ and\ \citenamefont
  {Yanase}(1960)}]{araki_measurement_1960}%
  \BibitemOpen
  \bibfield  {author} {\bibinfo {author} {\bibfnamefont {H.}~\bibnamefont
  {Araki}}\ and\ \bibinfo {author} {\bibfnamefont {M.~M.}\ \bibnamefont
  {Yanase}},\ }\href {\doibase 10.1103/PhysRev.120.622} {\bibfield  {journal}
  {\bibinfo  {journal} {Physical Review}\ }\textbf {\bibinfo {volume} {120}},\
  \bibinfo {pages} {622} (\bibinfo {year} {1960})}\BibitemShut {NoStop}%
\bibitem [{\citenamefont {Yanase}(1961)}]{yanase_optimal_1961}%
  \BibitemOpen
  \bibfield  {author} {\bibinfo {author} {\bibfnamefont {M.~M.}\ \bibnamefont
  {Yanase}},\ }\href {\doibase 10.1103/PhysRev.123.666} {\bibfield  {journal}
  {\bibinfo  {journal} {Physical Review}\ }\textbf {\bibinfo {volume} {123}},\
  \bibinfo {pages} {666} (\bibinfo {year} {1961})}\BibitemShut {NoStop}%
\bibitem [{\citenamefont
  {Ozawa}(2002{\natexlab{a}})}]{ozawa_conservation_2002}%
  \BibitemOpen
  \bibfield  {author} {\bibinfo {author} {\bibfnamefont {M.}~\bibnamefont
  {Ozawa}},\ }\href {\doibase 10.1103/PhysRevLett.88.050402} {\bibfield
  {journal} {\bibinfo  {journal} {Physical Review Letters}\ }\textbf {\bibinfo
  {volume} {88}},\ \bibinfo {pages} {050402} (\bibinfo {year}
  {2002}{\natexlab{a}})}\BibitemShut {NoStop}%
\bibitem [{\citenamefont {Tajima}\ and\ \citenamefont
  {Nagaoka}(2019)}]{tajima_coherence-variance_2019}%
  \BibitemOpen
  \bibfield  {author} {\bibinfo {author} {\bibfnamefont {H.}~\bibnamefont
  {Tajima}}\ and\ \bibinfo {author} {\bibfnamefont {H.}~\bibnamefont
  {Nagaoka}},\ }\href {http://arxiv.org/abs/1909.02904} {\bibfield  {journal}
  {\bibinfo  {journal} {arXiv:1909.02904 [cond-mat, physics:quant-ph]}\ }
  (\bibinfo {year} {2019})}\BibitemShut {NoStop}%
\bibitem [{\citenamefont
  {Ozawa}(2002{\natexlab{b}})}]{ozawa_conservative_2002}%
  \BibitemOpen
  \bibfield  {author} {\bibinfo {author} {\bibfnamefont {M.}~\bibnamefont
  {Ozawa}},\ }\href {\doibase 10.1103/PhysRevLett.89.057902} {\bibfield
  {journal} {\bibinfo  {journal} {Physical Review Letters}\ }\textbf {\bibinfo
  {volume} {89}},\ \bibinfo {pages} {057902} (\bibinfo {year}
  {2002}{\natexlab{b}})}\BibitemShut {NoStop}%
\bibitem [{\citenamefont {Tajima}\ \emph {et~al.}(2018)\citenamefont {Tajima},
  \citenamefont {Shiraishi},\ and\ \citenamefont
  {Saito}}]{tajima_uncertainty_2018}%
  \BibitemOpen
  \bibfield  {author} {\bibinfo {author} {\bibfnamefont {H.}~\bibnamefont
  {Tajima}}, \bibinfo {author} {\bibfnamefont {N.}~\bibnamefont {Shiraishi}}, \
  and\ \bibinfo {author} {\bibfnamefont {K.}~\bibnamefont {Saito}},\ }\href
  {\doibase 10.1103/PhysRevLett.121.110403} {\bibfield  {journal} {\bibinfo
  {journal} {Physical Review Letters}\ }\textbf {\bibinfo {volume} {121}},\
  \bibinfo {pages} {110403} (\bibinfo {year} {2018})}\BibitemShut {NoStop}%
\bibitem [{\citenamefont {Tajima}\ \emph {et~al.}(2020)\citenamefont {Tajima},
  \citenamefont {Shiraishi},\ and\ \citenamefont
  {Saito}}]{tajima_coherence_2020}%
  \BibitemOpen
  \bibfield  {author} {\bibinfo {author} {\bibfnamefont {H.}~\bibnamefont
  {Tajima}}, \bibinfo {author} {\bibfnamefont {N.}~\bibnamefont {Shiraishi}}, \
  and\ \bibinfo {author} {\bibfnamefont {K.}~\bibnamefont {Saito}},\ }\href
  {\doibase 10.1103/PhysRevResearch.2.043374} {\bibfield  {journal} {\bibinfo
  {journal} {Physical Review Research}\ }\textbf {\bibinfo {volume} {2}},\
  \bibinfo {pages} {043374} (\bibinfo {year} {2020})}\BibitemShut {NoStop}%
\bibitem [{\citenamefont {Tajima}\ and\ \citenamefont
  {Saito}(2021)}]{tajima_universal_2021}%
  \BibitemOpen
  \bibfield  {author} {\bibinfo {author} {\bibfnamefont {H.}~\bibnamefont
  {Tajima}}\ and\ \bibinfo {author} {\bibfnamefont {K.}~\bibnamefont {Saito}},\
  }\href {http://arxiv.org/abs/2103.01876} {\bibfield  {journal} {\bibinfo
  {journal} {arXiv:2103.01876 [cond-mat, physics:quant-ph]}\ } (\bibinfo {year}
  {2021})}\BibitemShut {NoStop}%
\bibitem [{\citenamefont {Tajima}\ \emph {et~al.}(2022)\citenamefont {Tajima},
  \citenamefont {Takagi},\ and\ \citenamefont
  {Kuramochi}}]{tajima_2022_universal}%
  \BibitemOpen
  \bibfield  {author} {\bibinfo {author} {\bibfnamefont {H.}~\bibnamefont
  {Tajima}}, \bibinfo {author} {\bibfnamefont {R.}~\bibnamefont {Takagi}}, \
  and\ \bibinfo {author} {\bibfnamefont {Y.}~\bibnamefont {Kuramochi}},\ }\href
  {https://arxiv.org/abs/2206.11086} {\bibfield  {journal} {\bibinfo  {journal}
  {arXiv preprint arXiv:2206.11086}\ } (\bibinfo {year} {2022})}\BibitemShut
  {NoStop}%
\bibitem [{\citenamefont {Kubica}\ and\ \citenamefont
  {Demkowicz-Dobrza\'nski}(2021)}]{kubica_using_2021}%
  \BibitemOpen
  \bibfield  {author} {\bibinfo {author} {\bibfnamefont {A.}~\bibnamefont
  {Kubica}}\ and\ \bibinfo {author} {\bibfnamefont {R.}~\bibnamefont
  {Demkowicz-Dobrza\'nski}},\ }\href {\doibase 10.1103/PhysRevLett.126.150503}
  {\bibfield  {journal} {\bibinfo  {journal} {Physical Review Letters}\
  }\textbf {\bibinfo {volume} {126}},\ \bibinfo {pages} {150503} (\bibinfo
  {year} {2021})}\BibitemShut {NoStop}%
\bibitem [{\citenamefont {Zhou}\ \emph {et~al.}(2021)\citenamefont {Zhou},
  \citenamefont {Liu},\ and\ \citenamefont {Jiang}}]{zhou_new_2021}%
  \BibitemOpen
  \bibfield  {author} {\bibinfo {author} {\bibfnamefont {S.}~\bibnamefont
  {Zhou}}, \bibinfo {author} {\bibfnamefont {Z.-W.}\ \bibnamefont {Liu}}, \
  and\ \bibinfo {author} {\bibfnamefont {L.}~\bibnamefont {Jiang}},\ }\href
  {\doibase 10.22331/q-2021-08-09-521} {\bibfield  {journal} {\bibinfo
  {journal} {Quantum}\ }\textbf {\bibinfo {volume} {5}},\ \bibinfo {pages}
  {521} (\bibinfo {year} {2021})}\BibitemShut {NoStop}%
\bibitem [{\citenamefont {Yang}\ \emph {et~al.}(2022)\citenamefont {Yang},
  \citenamefont {Mo}, \citenamefont {Renes}, \citenamefont {Chiribella},\ and\
  \citenamefont {Woods}}]{yang_optimal_2022}%
  \BibitemOpen
  \bibfield  {author} {\bibinfo {author} {\bibfnamefont {Y.}~\bibnamefont
  {Yang}}, \bibinfo {author} {\bibfnamefont {Y.}~\bibnamefont {Mo}}, \bibinfo
  {author} {\bibfnamefont {J.~M.}\ \bibnamefont {Renes}}, \bibinfo {author}
  {\bibfnamefont {G.}~\bibnamefont {Chiribella}}, \ and\ \bibinfo {author}
  {\bibfnamefont {M.~P.}\ \bibnamefont {Woods}},\ }\href {\doibase
  10.1103/PhysRevResearch.4.023107} {\bibfield  {journal} {\bibinfo  {journal}
  {Physical Review Research}\ }\textbf {\bibinfo {volume} {4}},\ \bibinfo
  {pages} {023107} (\bibinfo {year} {2022})}\BibitemShut {NoStop}%
\bibitem [{\citenamefont {Liu}\ and\ \citenamefont
  {Zhou}(2022)}]{liu_quantum_2022}%
  \BibitemOpen
  \bibfield  {author} {\bibinfo {author} {\bibfnamefont {Z.-W.}\ \bibnamefont
  {Liu}}\ and\ \bibinfo {author} {\bibfnamefont {S.}~\bibnamefont {Zhou}},\
  }\href {http://arxiv.org/abs/2111.06360} {\bibfield  {journal} {\bibinfo
  {journal} {arXiv:2111.06360 [quant-ph]}\ } (\bibinfo {year}
  {2022})}\BibitemShut {NoStop}%
\bibitem [{\citenamefont {Horodecki}\ \emph {et~al.}(2009)\citenamefont
  {Horodecki}, \citenamefont {Horodecki}, \citenamefont {Horodecki},\ and\
  \citenamefont {Horodecki}}]{horodecki_quantum_2009}%
  \BibitemOpen
  \bibfield  {author} {\bibinfo {author} {\bibfnamefont {R.}~\bibnamefont
  {Horodecki}}, \bibinfo {author} {\bibfnamefont {P.}~\bibnamefont
  {Horodecki}}, \bibinfo {author} {\bibfnamefont {M.}~\bibnamefont
  {Horodecki}}, \ and\ \bibinfo {author} {\bibfnamefont {K.}~\bibnamefont
  {Horodecki}},\ }\href {\doibase 10.1103/RevModPhys.81.865} {\bibfield
  {journal} {\bibinfo  {journal} {Reviews of Modern Physics}\ }\textbf
  {\bibinfo {volume} {81}},\ \bibinfo {pages} {865} (\bibinfo {year}
  {2009})}\BibitemShut {NoStop}%
\bibitem [{\citenamefont {Lostaglio}(2019)}]{lostaglio_introductory_2019}%
  \BibitemOpen
  \bibfield  {author} {\bibinfo {author} {\bibfnamefont {M.}~\bibnamefont
  {Lostaglio}},\ }\href {\doibase 10.1088/1361-6633/ab46e5} {\bibfield
  {journal} {\bibinfo  {journal} {Reports on Progress in Physics}\ }\textbf
  {\bibinfo {volume} {82}},\ \bibinfo {pages} {114001} (\bibinfo {year}
  {2019})},\ \bibinfo {note} {publisher: IOP Publishing}\BibitemShut {NoStop}%
\bibitem [{\citenamefont {Sagawa}(2022)}]{sagawa_entropy_2022}%
  \BibitemOpen
  \bibfield  {author} {\bibinfo {author} {\bibfnamefont {T.}~\bibnamefont
  {Sagawa}},\ }\href {https://link.springer.com/10.1007/978-981-16-6644-5}
  {\emph {\bibinfo {title} {Entropy, {Divergence}, and {Majorization} in
  {Classical} and {Quantum} {Thermodynamics}}}},\ \bibinfo {series}
  {{SpringerBriefs} in {Mathematical} {Physics}}, Vol.~\bibinfo {volume} {16}\
  (\bibinfo  {publisher} {Springer},\ \bibinfo {address} {Singapore},\ \bibinfo
  {year} {2022})\BibitemShut {NoStop}%
\bibitem [{\citenamefont {Gour}\ and\ \citenamefont
  {Spekkens}(2008)}]{gour_resource_2008}%
  \BibitemOpen
  \bibfield  {author} {\bibinfo {author} {\bibfnamefont {G.}~\bibnamefont
  {Gour}}\ and\ \bibinfo {author} {\bibfnamefont {R.~W.}\ \bibnamefont
  {Spekkens}},\ }\href {\doibase 10.1088/1367-2630/10/3/033023} {\bibfield
  {journal} {\bibinfo  {journal} {New Journal of Physics}\ }\textbf {\bibinfo
  {volume} {10}},\ \bibinfo {pages} {033023} (\bibinfo {year}
  {2008})}\BibitemShut {NoStop}%
\bibitem [{\citenamefont {Gour}\ \emph {et~al.}(2009)\citenamefont {Gour},
  \citenamefont {Marvian},\ and\ \citenamefont
  {Spekkens}}]{gour_measuring_2009}%
  \BibitemOpen
  \bibfield  {author} {\bibinfo {author} {\bibfnamefont {G.}~\bibnamefont
  {Gour}}, \bibinfo {author} {\bibfnamefont {I.}~\bibnamefont {Marvian}}, \
  and\ \bibinfo {author} {\bibfnamefont {R.~W.}\ \bibnamefont {Spekkens}},\
  }\href {\doibase 10.1103/PhysRevA.80.012307} {\bibfield  {journal} {\bibinfo
  {journal} {Physical Review A}\ }\textbf {\bibinfo {volume} {80}},\ \bibinfo
  {pages} {012307} (\bibinfo {year} {2009})}\BibitemShut {NoStop}%
\bibitem [{\citenamefont {Marvian}(2022)}]{marvian_operational_2022}%
  \BibitemOpen
  \bibfield  {author} {\bibinfo {author} {\bibfnamefont {I.}~\bibnamefont
  {Marvian}},\ }\href {\doibase 10.1103/PhysRevLett.129.190502} {\bibfield
  {journal} {\bibinfo  {journal} {Physical Review Letters}\ }\textbf {\bibinfo
  {volume} {129}},\ \bibinfo {pages} {190502} (\bibinfo {year}
  {2022})}\BibitemShut {NoStop}%
\bibitem [{\citenamefont {Chitambar}\ and\ \citenamefont
  {Gour}(2019)}]{chitambar_quantum_2019}%
  \BibitemOpen
  \bibfield  {author} {\bibinfo {author} {\bibfnamefont {E.}~\bibnamefont
  {Chitambar}}\ and\ \bibinfo {author} {\bibfnamefont {G.}~\bibnamefont
  {Gour}},\ }\href {\doibase 10.1103/RevModPhys.91.025001} {\bibfield
  {journal} {\bibinfo  {journal} {Reviews of Modern Physics}\ }\textbf
  {\bibinfo {volume} {91}},\ \bibinfo {pages} {025001} (\bibinfo {year}
  {2019})}\BibitemShut {NoStop}%
\bibitem [{\citenamefont {Bartlett}\ \emph {et~al.}(2007)\citenamefont
  {Bartlett}, \citenamefont {Rudolph},\ and\ \citenamefont
  {Spekkens}}]{bartlett_reference_2007}%
  \BibitemOpen
  \bibfield  {author} {\bibinfo {author} {\bibfnamefont {S.~D.}\ \bibnamefont
  {Bartlett}}, \bibinfo {author} {\bibfnamefont {T.}~\bibnamefont {Rudolph}}, \
  and\ \bibinfo {author} {\bibfnamefont {R.~W.}\ \bibnamefont {Spekkens}},\
  }\href {\doibase 10.1103/RevModPhys.79.555} {\bibfield  {journal} {\bibinfo
  {journal} {Reviews of Modern Physics}\ }\textbf {\bibinfo {volume} {79}},\
  \bibinfo {pages} {555} (\bibinfo {year} {2007})}\BibitemShut {NoStop}%
\bibitem [{\citenamefont
  {Marvian~Mashhad}(2012)}]{marvian_mashhad_symmetry_2012}%
  \BibitemOpen
  \bibfield  {author} {\bibinfo {author} {\bibfnamefont {I.}~\bibnamefont
  {Marvian~Mashhad}},\ }\emph {\bibinfo {title} {Symmetry, {Asymmetry} and
  {Quantum} {Information}}},\ \href
  {https://uwspace.uwaterloo.ca/handle/10012/7088} {Ph.D. thesis} (\bibinfo
  {year} {2012})\BibitemShut {NoStop}%
\bibitem [{\citenamefont {Takagi}\ and\ \citenamefont
  {Shiraishi}(2022)}]{takagi_correlation_2022}%
  \BibitemOpen
  \bibfield  {author} {\bibinfo {author} {\bibfnamefont {R.}~\bibnamefont
  {Takagi}}\ and\ \bibinfo {author} {\bibfnamefont {N.}~\bibnamefont
  {Shiraishi}},\ }\href {\doibase 10.1103/PhysRevLett.128.240501} {\bibfield
  {journal} {\bibinfo  {journal} {Physical Review Letters}\ }\textbf {\bibinfo
  {volume} {128}},\ \bibinfo {pages} {240501} (\bibinfo {year}
  {2022})}\BibitemShut {NoStop}%
\bibitem [{\citenamefont {Kudo}\ and\ \citenamefont
  {Tajima}(2023)}]{kudo_fisher_2022}%
  \BibitemOpen
  \bibfield  {author} {\bibinfo {author} {\bibfnamefont {D.}~\bibnamefont
  {Kudo}}\ and\ \bibinfo {author} {\bibfnamefont {H.}~\bibnamefont {Tajima}},\
  }\href {\doibase 10.1103/PhysRevA.107.062418} {\bibfield  {journal} {\bibinfo
   {journal} {Phys. Rev. A}\ }\textbf {\bibinfo {volume} {107}},\ \bibinfo
  {pages} {062418} (\bibinfo {year} {2023})}\BibitemShut {NoStop}%
\bibitem [{\citenamefont {Bennett}\ \emph {et~al.}(1996)\citenamefont
  {Bennett}, \citenamefont {Bernstein}, \citenamefont {Popescu},\ and\
  \citenamefont {Schumacher}}]{bennett_concentrating_1996}%
  \BibitemOpen
  \bibfield  {author} {\bibinfo {author} {\bibfnamefont {C.~H.}\ \bibnamefont
  {Bennett}}, \bibinfo {author} {\bibfnamefont {H.~J.}\ \bibnamefont
  {Bernstein}}, \bibinfo {author} {\bibfnamefont {S.}~\bibnamefont {Popescu}},
  \ and\ \bibinfo {author} {\bibfnamefont {B.}~\bibnamefont {Schumacher}},\
  }\href {\doibase 10.1103/PhysRevA.53.2046} {\bibfield  {journal} {\bibinfo
  {journal} {Physical Review A}\ }\textbf {\bibinfo {volume} {53}},\ \bibinfo
  {pages} {2046} (\bibinfo {year} {1996})}\BibitemShut {NoStop}%
\bibitem [{\citenamefont {Hayashi}(2006)}]{hayashi_general_2006}%
  \BibitemOpen
  \bibfield  {author} {\bibinfo {author} {\bibfnamefont {M.}~\bibnamefont
  {Hayashi}},\ }\href {\doibase 10.1109/TIT.2006.872976} {\bibfield  {journal}
  {\bibinfo  {journal} {IEEE Transactions on Information Theory}\ }\textbf
  {\bibinfo {volume} {52}},\ \bibinfo {pages} {1904} (\bibinfo {year}
  {2006})}\BibitemShut {NoStop}%
\bibitem [{\citenamefont {Bowen}\ and\ \citenamefont
  {Datta}(2008)}]{bowen_asymptotic_2008}%
  \BibitemOpen
  \bibfield  {author} {\bibinfo {author} {\bibfnamefont {G.}~\bibnamefont
  {Bowen}}\ and\ \bibinfo {author} {\bibfnamefont {N.}~\bibnamefont {Datta}},\
  }\href {\doibase 10.1109/TIT.2008.926377} {\bibfield  {journal} {\bibinfo
  {journal} {IEEE Transactions on Information Theory}\ }\textbf {\bibinfo
  {volume} {54}},\ \bibinfo {pages} {3677} (\bibinfo {year}
  {2008})}\BibitemShut {NoStop}%
\bibitem [{\citenamefont {Tajima}\ \emph {et~al.}(2017)\citenamefont {Tajima},
  \citenamefont {Wakakuwa},\ and\ \citenamefont {Ogawa}}]{tajima_large_2017}%
  \BibitemOpen
  \bibfield  {author} {\bibinfo {author} {\bibfnamefont {H.}~\bibnamefont
  {Tajima}}, \bibinfo {author} {\bibfnamefont {E.}~\bibnamefont {Wakakuwa}}, \
  and\ \bibinfo {author} {\bibfnamefont {T.}~\bibnamefont {Ogawa}},\ }\href
  {http://arxiv.org/abs/1611.06614} {\bibfield  {journal} {\bibinfo  {journal}
  {arXiv:1611.06614 [cond-mat, physics:quant-ph]}\ } (\bibinfo {year}
  {2017})}\BibitemShut {NoStop}%
\bibitem [{\citenamefont {Faist}\ \emph {et~al.}(2019)\citenamefont {Faist},
  \citenamefont {Sagawa}, \citenamefont {Kato}, \citenamefont {Nagaoka},\ and\
  \citenamefont {Brandão}}]{faist_macroscopic_2019}%
  \BibitemOpen
  \bibfield  {author} {\bibinfo {author} {\bibfnamefont {P.}~\bibnamefont
  {Faist}}, \bibinfo {author} {\bibfnamefont {T.}~\bibnamefont {Sagawa}},
  \bibinfo {author} {\bibfnamefont {K.}~\bibnamefont {Kato}}, \bibinfo {author}
  {\bibfnamefont {H.}~\bibnamefont {Nagaoka}}, \ and\ \bibinfo {author}
  {\bibfnamefont {F.~G.}\ \bibnamefont {Brandão}},\ }\href {\doibase
  10.1103/PhysRevLett.123.250601} {\bibfield  {journal} {\bibinfo  {journal}
  {Physical Review Letters}\ }\textbf {\bibinfo {volume} {123}},\ \bibinfo
  {pages} {250601} (\bibinfo {year} {2019})}\BibitemShut {NoStop}%
\bibitem [{\citenamefont {Sagawa}\ \emph {et~al.}(2021)\citenamefont {Sagawa},
  \citenamefont {Faist}, \citenamefont {Kato}, \citenamefont {Matsumoto},
  \citenamefont {Nagaoka},\ and\ \citenamefont
  {Brandão}}]{sagawa_asymptotic_2021}%
  \BibitemOpen
  \bibfield  {author} {\bibinfo {author} {\bibfnamefont {T.}~\bibnamefont
  {Sagawa}}, \bibinfo {author} {\bibfnamefont {P.}~\bibnamefont {Faist}},
  \bibinfo {author} {\bibfnamefont {K.}~\bibnamefont {Kato}}, \bibinfo {author}
  {\bibfnamefont {K.}~\bibnamefont {Matsumoto}}, \bibinfo {author}
  {\bibfnamefont {H.}~\bibnamefont {Nagaoka}}, \ and\ \bibinfo {author}
  {\bibfnamefont {F.~G. S.~L.}\ \bibnamefont {Brandão}},\ }\href {\doibase
  10.1088/1751-8121/ac333c} {\bibfield  {journal} {\bibinfo  {journal} {Journal
  of Physics A: Mathematical and Theoretical}\ }\textbf {\bibinfo {volume}
  {54}},\ \bibinfo {pages} {495303} (\bibinfo {year} {2021})}\BibitemShut
  {NoStop}%
\bibitem [{\citenamefont {Ogawa}\ and\ \citenamefont
  {Nagaoka}(2000)}]{ogawa_strong_2000}%
  \BibitemOpen
  \bibfield  {author} {\bibinfo {author} {\bibfnamefont {T.}~\bibnamefont
  {Ogawa}}\ and\ \bibinfo {author} {\bibfnamefont {H.}~\bibnamefont
  {Nagaoka}},\ }\href {\doibase 10.1109/18.887855} {\bibfield  {journal}
  {\bibinfo  {journal} {IEEE Transactions on Information Theory}\ }\textbf
  {\bibinfo {volume} {46}},\ \bibinfo {pages} {2428} (\bibinfo {year}
  {2000})}\BibitemShut {NoStop}%
\bibitem [{\citenamefont {Nagaoka}\ and\ \citenamefont
  {Hayashi}(2007)}]{nagaoka_information-spectrum_2007}%
  \BibitemOpen
  \bibfield  {author} {\bibinfo {author} {\bibfnamefont {H.}~\bibnamefont
  {Nagaoka}}\ and\ \bibinfo {author} {\bibfnamefont {M.}~\bibnamefont
  {Hayashi}},\ }\href {\doibase 10.1109/TIT.2006.889463} {\bibfield  {journal}
  {\bibinfo  {journal} {IEEE Transactions on Information Theory}\ }\textbf
  {\bibinfo {volume} {53}},\ \bibinfo {pages} {534} (\bibinfo {year}
  {2007})}\BibitemShut {NoStop}%
\bibitem [{\citenamefont {Hayashi}\ and\ \citenamefont
  {Nagaoka}(2002)}]{hayashi_general_2002}%
  \BibitemOpen
  \bibfield  {author} {\bibinfo {author} {\bibfnamefont {M.}~\bibnamefont
  {Hayashi}}\ and\ \bibinfo {author} {\bibfnamefont {H.}~\bibnamefont
  {Nagaoka}},\ }in\ \href {\doibase 10.1109/ISIT.2002.1023343} {\emph {\bibinfo
  {booktitle} {Proceedings {IEEE} {International} {Symposium} on {Information}
  {Theory},}}}\ (\bibinfo {year} {2002})\ pp.\ \bibinfo {pages}
  {71--}\BibitemShut {NoStop}%
\bibitem [{\citenamefont {Hayashi}\ and\ \citenamefont
  {Nagaoka}(2003)}]{hayashi_general_channel_2003}%
  \BibitemOpen
  \bibfield  {author} {\bibinfo {author} {\bibfnamefont {M.}~\bibnamefont
  {Hayashi}}\ and\ \bibinfo {author} {\bibfnamefont {H.}~\bibnamefont
  {Nagaoka}},\ }\href {\doibase 10.1109/TIT.2003.813556} {\bibfield  {journal}
  {\bibinfo  {journal} {IEEE Transactions on Information Theory}\ }\textbf
  {\bibinfo {volume} {49}},\ \bibinfo {pages} {1753} (\bibinfo {year}
  {2003})}\BibitemShut {NoStop}%
\bibitem [{\citenamefont {Bowen}\ and\ \citenamefont
  {Datta}(2006)}]{bowen_beyond_2006}%
  \BibitemOpen
  \bibfield  {author} {\bibinfo {author} {\bibfnamefont {G.}~\bibnamefont
  {Bowen}}\ and\ \bibinfo {author} {\bibfnamefont {N.}~\bibnamefont {Datta}},\
  }in\ \href {\doibase 10.1109/ISIT.2006.261709} {\emph {\bibinfo {booktitle}
  {2006 {IEEE} {International} {Symposium} on {Information} {Theory}}}}\
  (\bibinfo {year} {2006})\ pp.\ \bibinfo {pages} {451--455}\BibitemShut
  {NoStop}%
\bibitem [{\citenamefont {Jiao}\ \emph {et~al.}(2018)\citenamefont {Jiao},
  \citenamefont {Wakakuwa},\ and\ \citenamefont
  {Ogawa}}]{jiao_asymptotic_2018}%
  \BibitemOpen
  \bibfield  {author} {\bibinfo {author} {\bibfnamefont {Y.}~\bibnamefont
  {Jiao}}, \bibinfo {author} {\bibfnamefont {E.}~\bibnamefont {Wakakuwa}}, \
  and\ \bibinfo {author} {\bibfnamefont {T.}~\bibnamefont {Ogawa}},\ }\href
  {\doibase 10.1063/1.5013183} {\bibfield  {journal} {\bibinfo  {journal}
  {Journal of Mathematical Physics}\ }\textbf {\bibinfo {volume} {59}},\
  \bibinfo {pages} {022201} (\bibinfo {year} {2018})}\BibitemShut {NoStop}%
\bibitem [{\citenamefont {Datta}\ and\ \citenamefont
  {Renner}(2009)}]{datta_smooth_2009}%
  \BibitemOpen
  \bibfield  {author} {\bibinfo {author} {\bibfnamefont {N.}~\bibnamefont
  {Datta}}\ and\ \bibinfo {author} {\bibfnamefont {R.}~\bibnamefont {Renner}},\
  }\href {\doibase 10.1109/TIT.2009.2018340} {\bibfield  {journal} {\bibinfo
  {journal} {IEEE Transactions on Information Theory}\ }\textbf {\bibinfo
  {volume} {55}},\ \bibinfo {pages} {2807} (\bibinfo {year}
  {2009})}\BibitemShut {NoStop}%
\bibitem [{\citenamefont {Renner}(2005)}]{renner_security_2005}%
  \BibitemOpen
  \bibfield  {author} {\bibinfo {author} {\bibfnamefont {R.}~\bibnamefont
  {Renner}},\ }\emph {\bibinfo {title} {Security of {Quantum} {Key}
  {Distribution}}},\ \href {http://arxiv.org/abs/quant-ph/0512258} {\bibinfo
  {type} {{PhD} {Thesis}}},\ \bibinfo  {school} {ETH Zurich} (\bibinfo {year}
  {2005})\BibitemShut {NoStop}%
\bibitem [{\citenamefont {Nielsen}(1999)}]{nielsen_conditions_1999}%
  \BibitemOpen
  \bibfield  {author} {\bibinfo {author} {\bibfnamefont {M.~A.}\ \bibnamefont
  {Nielsen}},\ }\href {\doibase 10.1103/PhysRevLett.83.436} {\bibfield
  {journal} {\bibinfo  {journal} {Physical Review Letters}\ }\textbf {\bibinfo
  {volume} {83}},\ \bibinfo {pages} {436} (\bibinfo {year} {1999})}\BibitemShut
  {NoStop}%
\bibitem [{Note1()}]{Note1}%
  \BibitemOpen
  \bibinfo {note} {In fact, any completely incoherence-preserving operation
  satisfies Eq.~\protect \textup {\hbox {\mathsurround \z@ \protect \normalfont
  (\ignorespaces \ref {eq:covariant_channel_definition}\unskip \@@italiccorr
  )}} \cite {marvian_coherence_2020}. Here, a channel $\protect \mathcal {E}$
  is called completely incoherence-preserving iff for any ancillary system $A$
  with an arbitrary Hamiltonian $H_A$, the map $\protect \mathcal {E}\otimes
  \protect \mathcal {I}_A$ transforms any symmetric state to a symmetric state,
  where $\protect \mathcal {I}_A$ denotes the identity map on $A$.}\BibitemShut
  {Stop}%
\bibitem [{\citenamefont {Helstrom}(1969)}]{helstrom_quantum_1969}%
  \BibitemOpen
  \bibfield  {author} {\bibinfo {author} {\bibfnamefont {C.~W.}\ \bibnamefont
  {Helstrom}},\ }\href {\doibase 10.1007/BF01007479} {\bibfield  {journal}
  {\bibinfo  {journal} {Journal of Statistical Physics}\ }\textbf {\bibinfo
  {volume} {1}},\ \bibinfo {pages} {231} (\bibinfo {year} {1969})}\BibitemShut
  {NoStop}%
\bibitem [{\citenamefont {Holevo}(2011)}]{holevo_probabilistic_2011}%
  \BibitemOpen
  \bibfield  {author} {\bibinfo {author} {\bibfnamefont {A.}~\bibnamefont
  {Holevo}},\ }\href {\doibase 10.1007/978-88-7642-378-9} {\emph {\bibinfo
  {title} {Probabilistic and {Statistical} {Aspects} of {Quantum} {Theory}}}},\
  \bibinfo {edition} {2nd}\ ed.\ (\bibinfo  {publisher} {Edizioni della
  Normale},\ \bibinfo {address} {Pisa},\ \bibinfo {year} {2011})\BibitemShut
  {NoStop}%
\bibitem [{\citenamefont {Hansen}(2008)}]{hansen_metric_2008}%
  \BibitemOpen
  \bibfield  {author} {\bibinfo {author} {\bibfnamefont {F.}~\bibnamefont
  {Hansen}},\ }\href {\doibase 10.1073/pnas.0803323105} {\bibfield  {journal}
  {\bibinfo  {journal} {Proceedings of the National Academy of Sciences}\
  }\textbf {\bibinfo {volume} {105}},\ \bibinfo {pages} {9909} (\bibinfo {year}
  {2008})}\BibitemShut {NoStop}%
\bibitem [{\citenamefont {Takagi}(2019)}]{takagi_skew_2019}%
  \BibitemOpen
  \bibfield  {author} {\bibinfo {author} {\bibfnamefont {R.}~\bibnamefont
  {Takagi}},\ }\href {\doibase 10.1038/s41598-019-50279-w} {\bibfield
  {journal} {\bibinfo  {journal} {Scientific Reports}\ }\textbf {\bibinfo
  {volume} {9}},\ \bibinfo {pages} {14562} (\bibinfo {year}
  {2019})}\BibitemShut {NoStop}%
\bibitem [{Com()}]{Comment_convex_roof}%
  \BibitemOpen
  \href@noop {} {}\bibinfo {note} {It is known that QFI is given by the convex
  roof of variance: $\mathcal{F}(\rho)=4\min_{\{q_i,\phi_i\}}\sum_iq_i
  V_{\phi_i}(H)$, where $\{q_i,\phi_i\}$ runs over the set of all probability
  distributions $\{q_i\}_i$ and pure states $\{\phi_i\}_i$ such that
  $\rho=\sum_iq_i\phi_i$ \cite{toth_extremal_2013,yu_quantum_2013}. Here,
  $V_{\phi_i}(H)$ denotes the variance of $H$ in $\phi_i$. This property shows
  that QFI quantifies the minimum average quantum fluctuation over the
  ensemble.}\BibitemShut {Stop}%
\bibitem [{SM()}]{SM}%
  \BibitemOpen
  \href@noop {} {}\bibinfo {note} {See Supplemental Material for
  details.}\BibitemShut {Stop}%
\bibitem [{\citenamefont {Marvian}\ and\ \citenamefont
  {Spekkens}(2013)}]{marvian_theory_2013}%
  \BibitemOpen
  \bibfield  {author} {\bibinfo {author} {\bibfnamefont {I.}~\bibnamefont
  {Marvian}}\ and\ \bibinfo {author} {\bibfnamefont {R.~W.}\ \bibnamefont
  {Spekkens}},\ }\href {\doibase 10.1088/1367-2630/15/3/033001} {\bibfield
  {journal} {\bibinfo  {journal} {New Journal of Physics}\ }\textbf {\bibinfo
  {volume} {15}},\ \bibinfo {pages} {033001} (\bibinfo {year}
  {2013})}\BibitemShut {NoStop}%
\bibitem [{\citenamefont {Lieb}\ and\ \citenamefont
  {Yngvason}(2013)}]{lieb_entropy_2013}%
  \BibitemOpen
  \bibfield  {author} {\bibinfo {author} {\bibfnamefont {E.~H.}\ \bibnamefont
  {Lieb}}\ and\ \bibinfo {author} {\bibfnamefont {J.}~\bibnamefont
  {Yngvason}},\ }\href {\doibase 10.1098/rspa.2013.0408} {\bibfield  {journal}
  {\bibinfo  {journal} {Proceedings of the Royal Society A: Mathematical,
  Physical and Engineering Sciences}\ }\textbf {\bibinfo {volume} {469}},\
  \bibinfo {pages} {20130408} (\bibinfo {year} {2013})}\BibitemShut {NoStop}%
\bibitem [{\citenamefont {Ito}\ and\ \citenamefont
  {Dechant}(2020)}]{ito_stochastic_2020}%
  \BibitemOpen
  \bibfield  {author} {\bibinfo {author} {\bibfnamefont {S.}~\bibnamefont
  {Ito}}\ and\ \bibinfo {author} {\bibfnamefont {A.}~\bibnamefont {Dechant}},\
  }\href {\doibase 10.1103/PhysRevX.10.021056} {\bibfield  {journal} {\bibinfo
  {journal} {Physical Review X}\ }\textbf {\bibinfo {volume} {10}},\ \bibinfo
  {pages} {021056} (\bibinfo {year} {2020})},\ \bibinfo {note} {publisher:
  American Physical Society}\BibitemShut {NoStop}%
\bibitem [{\citenamefont {Tan}\ \emph {et~al.}(2021)\citenamefont {Tan},
  \citenamefont {Narasimhachar},\ and\ \citenamefont
  {Regula}}]{tan_fisher_2021}%
  \BibitemOpen
  \bibfield  {author} {\bibinfo {author} {\bibfnamefont {K.~C.}\ \bibnamefont
  {Tan}}, \bibinfo {author} {\bibfnamefont {V.}~\bibnamefont {Narasimhachar}},
  \ and\ \bibinfo {author} {\bibfnamefont {B.}~\bibnamefont {Regula}},\ }\href
  {\doibase 10.1103/PhysRevLett.127.200402} {\bibfield  {journal} {\bibinfo
  {journal} {Physical Review Letters}\ }\textbf {\bibinfo {volume} {127}},\
  \bibinfo {pages} {200402} (\bibinfo {year} {2021})}\BibitemShut {NoStop}%
\bibitem [{\citenamefont {T\'oth}\ and\ \citenamefont
  {Petz}(2013)}]{toth_extremal_2013}%
  \BibitemOpen
  \bibfield  {author} {\bibinfo {author} {\bibfnamefont {G.}~\bibnamefont
  {T\'oth}}\ and\ \bibinfo {author} {\bibfnamefont {D.}~\bibnamefont {Petz}},\
  }\href {\doibase 10.1103/PhysRevA.87.032324} {\bibfield  {journal} {\bibinfo
  {journal} {Physical Review A}\ }\textbf {\bibinfo {volume} {87}},\ \bibinfo
  {pages} {032324} (\bibinfo {year} {2013})}\BibitemShut {NoStop}%
\bibitem [{\citenamefont {Yu}(2013)}]{yu_quantum_2013}%
  \BibitemOpen
  \bibfield  {author} {\bibinfo {author} {\bibfnamefont {S.}~\bibnamefont
  {Yu}},\ }\href {http://arxiv.org/abs/1302.5311} {\bibfield  {journal}
  {\bibinfo  {journal} {arXiv:1302.5311 [quant-ph]}\ } (\bibinfo {year}
  {2013})}\BibitemShut {NoStop}%
\bibitem [{\citenamefont {Wilf}(2011)}]{wilf_generatingfunctionology_2011}%
  \BibitemOpen
  \bibfield  {author} {\bibinfo {author} {\bibfnamefont {H.~S.}\ \bibnamefont
  {Wilf}},\ }\href@noop {} {\emph {\bibinfo {title} {generatingfunctionology:
  {Third} {Edition}}}},\ \bibinfo {edition} {3rd}\ ed.\ (\bibinfo  {publisher}
  {A K Peters/CRC Press},\ \bibinfo {address} {New York},\ \bibinfo {year}
  {2011})\BibitemShut {NoStop}%
\bibitem [{\citenamefont {Hardy}\ \emph {et~al.}(1929)\citenamefont {Hardy},
  \citenamefont {Littlewood},\ and\ \citenamefont
  {Pólya}}]{hardy_simple_1929}%
  \BibitemOpen
  \bibfield  {author} {\bibinfo {author} {\bibfnamefont {G.~H.}\ \bibnamefont
  {Hardy}}, \bibinfo {author} {\bibfnamefont {J.~E.}\ \bibnamefont
  {Littlewood}}, \ and\ \bibinfo {author} {\bibfnamefont {G.}~\bibnamefont
  {Pólya}},\ }\href@noop {} {\bibfield  {journal} {\bibinfo  {journal}
  {Messenger Math.}\ }\textbf {\bibinfo {volume} {58}},\ \bibinfo {pages} {145}
  (\bibinfo {year} {1929})}\BibitemShut {NoStop}%
\bibitem [{\citenamefont {Birkhoff}(1946)}]{birkhoff_tres_1946}%
  \BibitemOpen
  \bibfield  {author} {\bibinfo {author} {\bibfnamefont {G.}~\bibnamefont
  {Birkhoff}},\ }\href {https://ci.nii.ac.jp/naid/10006506380/} {\bibfield
  {journal} {\bibinfo  {journal} {Univ. Nac. Tucuman, Ser. A}\ }\textbf
  {\bibinfo {volume} {no.5}},\ \bibinfo {pages} {147} (\bibinfo {year}
  {1946})}\BibitemShut {NoStop}%
\bibitem [{\citenamefont {Marshall}\ \emph {et~al.}(2011)\citenamefont
  {Marshall}, \citenamefont {Olkin},\ and\ \citenamefont
  {Arnold}}]{marshall_inequalities_2011}%
  \BibitemOpen
  \bibfield  {author} {\bibinfo {author} {\bibfnamefont {A.~W.}\ \bibnamefont
  {Marshall}}, \bibinfo {author} {\bibfnamefont {I.}~\bibnamefont {Olkin}}, \
  and\ \bibinfo {author} {\bibfnamefont {B.~C.}\ \bibnamefont {Arnold}},\
  }\href {http://link.springer.com/10.1007/978-0-387-68276-1} {\emph {\bibinfo
  {title} {Inequalities: {Theory} of {Majorization} and {Its}
  {Applications}}}},\ Springer {Series} in {Statistics}\ (\bibinfo  {publisher}
  {Springer New York},\ \bibinfo {address} {New York},\ \bibinfo {year}
  {2011})\BibitemShut {NoStop}%
\bibitem [{\citenamefont {Hayashi}(2003)}]{hayashi_general_2003}%
  \BibitemOpen
  \bibfield  {author} {\bibinfo {author} {\bibfnamefont {M.}~\bibnamefont
  {Hayashi}},\ }in\ \href {\doibase 10.1109/ISIT.2003.1228448} {\emph {\bibinfo
  {booktitle} {{IEEE} {International} {Symposium} on {Information} {Theory},
  2003. {Proceedings}.}}}\ (\bibinfo {year} {2003})\ pp.\ \bibinfo {pages}
  {431--}\BibitemShut {NoStop}%
\bibitem [{\citenamefont {Barbour}\ and\ \citenamefont
  {Ćekanavićius}(2002)}]{barbour_total_2002}%
  \BibitemOpen
  \bibfield  {author} {\bibinfo {author} {\bibfnamefont {A.~D.}\ \bibnamefont
  {Barbour}}\ and\ \bibinfo {author} {\bibfnamefont {V.}~\bibnamefont
  {Ćekanavićius}},\ }\href {\doibase 10.1214/aop/1023481001} {\bibfield
  {journal} {\bibinfo  {journal} {The Annals of Probability}\ }\textbf
  {\bibinfo {volume} {30}},\ \bibinfo {pages} {509} (\bibinfo {year}
  {2002})}\BibitemShut {NoStop}%
\bibitem [{\citenamefont {Roos}(2003)}]{roos_improvements_2003}%
  \BibitemOpen
  \bibfield  {author} {\bibinfo {author} {\bibfnamefont {B.}~\bibnamefont
  {Roos}},\ }\href {\doibase 10.1016/S0378-3758(02)00095-2} {\bibfield
  {journal} {\bibinfo  {journal} {Journal of Statistical Planning and
  Inference}\ }\textbf {\bibinfo {volume} {113}},\ \bibinfo {pages} {467}
  (\bibinfo {year} {2003})}\BibitemShut {NoStop}%
\bibitem [{\citenamefont {Adell}\ and\ \citenamefont
  {Jodrá}(2006)}]{adell_exact_2006}%
  \BibitemOpen
  \bibfield  {author} {\bibinfo {author} {\bibfnamefont {J.~A.}\ \bibnamefont
  {Adell}}\ and\ \bibinfo {author} {\bibfnamefont {P.}~\bibnamefont {Jodrá}},\
  }\href {\doibase 10.1155/JIA/2006/64307} {\bibfield  {journal} {\bibinfo
  {journal} {Journal of Inequalities and Applications}\ }\textbf {\bibinfo
  {volume} {2006}},\ \bibinfo {pages} {1} (\bibinfo {year} {2006})}\BibitemShut
  {NoStop}%
\bibitem [{\citenamefont {Keyl}\ and\ \citenamefont
  {Werner}(1999)}]{keyl_optimal_1999}%
  \BibitemOpen
  \bibfield  {author} {\bibinfo {author} {\bibfnamefont {M.}~\bibnamefont
  {Keyl}}\ and\ \bibinfo {author} {\bibfnamefont {R.~F.}\ \bibnamefont
  {Werner}},\ }\href {\doibase 10.1063/1.532887} {\bibfield  {journal}
  {\bibinfo  {journal} {Journal of Mathematical Physics}\ }\textbf {\bibinfo
  {volume} {40}},\ \bibinfo {pages} {3283} (\bibinfo {year}
  {1999})}\BibitemShut {NoStop}%
\bibitem [{\citenamefont {Yamaguchi}\ and\ \citenamefont
  {Tajima}(2023)}]{yamaguchi_smooth_2022}%
  \BibitemOpen
  \bibfield  {author} {\bibinfo {author} {\bibfnamefont {K.}~\bibnamefont
  {Yamaguchi}}\ and\ \bibinfo {author} {\bibfnamefont {H.}~\bibnamefont
  {Tajima}},\ }\href {\doibase 10.22331/q-2023-05-22-1012} {\bibfield
  {journal} {\bibinfo  {journal} {{Quantum}}\ }\textbf {\bibinfo {volume}
  {7}},\ \bibinfo {pages} {1012} (\bibinfo {year} {2023})}\BibitemShut
  {NoStop}%
\end{thebibliography}
\end{document}